%% file: PDF_MAIN.tex
\documentclass[acmsmall]{acmart}
\usepackage{chngcntr}
\usepackage{multicol}
\usepackage[shortlabels]{enumitem}
\makeatletter
\@addtoreset{equation}{section}
\makeatother
\numberwithin{equation}{section}
\counterwithin{figure}{section}

\acmYear{2022}
\acmDOI{} 
\startPage{1}

\allowdisplaybreaks

\input{TEX/defs}

\usepackage{tikz-cd}
\usepackage{ifluatex}
\input diagxy        
\xyoption{curve}     
\xyoption{all}
\xyoption{2cell}

\begin{document}

	
\ifluatex
\directlua{adddednatlualoader = function ()
     require = function (stem)
         local fname = dednat6dir..stem..".lua"
         package.loaded[stem] = package.loaded[stem] or dofile(fname) or fname
       end
   end}
\catcode`\^^J=10
\directlua{dofile "dednat6load.lua"}
\else
%
\def\diagxyto{\ifnextchar/{\toop}{\toop/>/}}
\def\to     {\rightarrow}
\def\defded#1#2{\expandafter\def\csname ded-#1\endcsname{#2}}
\def\ifdedundefined#1{\expandafter\ifx\csname ded-#1\endcsname\relax}
\def\ded#1{\ifdedundefined{#1}
    \errmessage{UNDEFINED DEDUCTION: #1}
  \else
    \csname ded-#1\endcsname
  \fi
}
\def\defdiag#1#2{\expandafter\def\csname diag-#1\endcsname{\bfig#2\efig}}
\def\defdiagprep#1#2#3{\expandafter\def\csname diag-#1\endcsname{{#2\bfig#3\efig}}}
\def\ifdiagundefined#1{\expandafter\ifx\csname diag-#1\endcsname\relax}
\def\diag#1{\ifdiagundefined{#1}
    \errmessage{UNDEFINED DIAGRAM: #1}
  \else
    \csname diag-#1\endcsname
  \fi
}
\newlength{\celllower}
\newlength{\lcelllower}
\def\cellfont{}
\def\lcellfont{}
\def\cell #1{\lower\celllower\hbox to 0pt{\hss\cellfont${#1}$\hss}}
\def\lcell#1{\lower\celllower\hbox to 0pt   {\lcellfont${#1}$\hss}}
\def\expr#1{\directlua{output(tostring(#1))}}
\def\eval#1{\directlua{#1}}
\def\pu{\directlua{pu()}}
%

\defdiag{diag-new-terms-for-the-universal-property}{   
  \morphism(0,0)|a|/->/<1800,0>[{\RCBVcat\left(\left(\SynVr{,}\SynTr{,}\SynffixpointRecT\right){,}\left(\catV{,}\monadT{,}\ffixpointRecT\right)\right)}`{\CBVcat\left(\rCBVtoCBV\left(\SynVr{,}\SynTr{,}\SynffixpointRecT\right){,}\rCBVtoCBV\left(\catV{,}\monadT{,}\ffixpointRecT\right)\right)};{\rCBVtoCBV}]
  \morphism(1800,0)|r|/->/<0,-450>[{\CBVcat\left(\rCBVtoCBV\left(\SynVr{,}\SynTr{,}\SynffixpointRecT\right){,}\rCBVtoCBV\left(\catV{,}\monadT{,}\ffixpointRecT\right)\right)}`{\CBVcat\left(\left(\SynV{,}\SynT{,}\Synfix{,}\Synit\right){,}\rCBVtoCBV\left(\catV{,}\monadT{,}\ffixpointRecT\right)\right)};{\CBVcat\left(\incCBVtorCBV,\rCBVtoCBV\left(\catV{,}\monadT{,}\ffixpointRecT\right)\right)}]
}
\defdiag{unit-pullback}{   
  \morphism(0,0)/->/<600,0>[{D}`{\unlift{\LRMONAD\left(D,C,j\right)}};]
  \morphism(0,0)|b|/->/<0,-450>[{D}`{{G(C)}};{j}]
  \morphism(600,0)|r|/->/<0,-450>[{\unlift{\LRMONAD\left(D,C,j\right)}}`{{G(T'(C))}};{\unlift{\LRMONAD{j}}}]
  \morphism(0,-450)|b|/->/<600,0>[{{G(C)}}`{{G(T'(C))}};{{G(\ee_C)}}]
}
\defdiag{rolling-wCPO-pairs}{   
  \morphism(0,0)|a|/->/<1350,0>[{J\circ\paE{E}{A}{\catV}\left(\colim\left(\unlifteppair{\pEDchain{E}{A}}\right){,}\colim\left(\unlifteppair{\pEDchain{E}{A}}\right)\right)}`{\paE{E}{A}{\catC}\left(\colim\left(\pEDchain{E}{A}\right){,}\colim\left({\pEDchain{E}{A}}\right)\right)};{\cong}]
  \morphism(1350,0)|a|/->/<975,0>[{\paE{E}{A}{\catC}\left(\colim\left(\pEDchain{E}{A}\right){,}\colim\left({\pEDchain{E}{A}}\right)\right)}`{\colim\left(\pEDchain{E}{A}\right)};{\cong}]
  \morphism(2325,0)|a|/->/<675,0>[{\colim\left(\pEDchain{E}{A}\right)}`{J\colim\left(\unlifteppair{\pEDchain{E}{A}}\right)};{\cong}]
}
\defdiag{roll-eq1-diagram}{   
  \morphism(0,0)|a|/->/<1050,0>[{{\left({\catV}^\op\times\catV\right)^{n-1}}}`{{\left({\catV}^\op\times\catV\right)^n}};{{\pairL{\id},\pE{\fixpointRecT{E}}{\catV}^\op{,}\pE{\fixpointRecT{E}}{\catV}\pairR}}]
  \morphism(0,0)|b|/->/<1050,-450>[{{\left({\catV}^\op\times\catV\right)^{n-1}}}`{{\catV}};{{\pE{\fixpointRecT{E}}{\catV}}}]
  \morphism(1050,0)|r|/->/<0,-450>[{{\left({\catV}^\op\times\catV\right)^n}}`{{\catV}};{{\pE{E}{\catV}}}]
  \morphism(1050,-450)|b|/->/<-1050,0>[{{\catV}}`{{{\catV}'}};{{H}}]
  \morphism(638,-225)|a|/<=/<375,0>[{\phantom{O}}`{\phantom{O}};{{\rollReT}^{\pEE{E}}}]
}
\defdiag{diag-chain-of-ep-pairs-projections}{   
  \morphism(0,0)|a|/<-/<300,0>[{\epzeroo}`{\objectD{A}{1}};{\comochain{0}}]
  \morphism(300,0)|a|/<-/<300,0>[{\objectD{A}{1}}`{\objectD{A}{2}};{\comochain{1}}]
  \morphism(600,0)|a|/<-/<300,0>[{\objectD{A}{2}}`{\objectD{A}{3}};{\comochain{2}}]
  \morphism(900,0)|a|/<-/<270,0>[{\objectD{A}{3}}`{{\cdots}};{\comochain{3}}]
}
\defdiag{diag-chain-of-ep-pairs-for-the-F}{   
  \morphism(0,0)|a|/->/<300,0>[{\initiall}`{\FforREC\left(\initiall\right)};{\coproje{}}]
  \morphism(300,0)|a|/->/<480,0>[{\FforREC\left(\initiall\right)}`{\FforREC^2\left(\initiall\right)};{\FforREC\left(\coproje{}\right)}]
  \morphism(780,0)|a|/->/<480,0>[{\FforREC^2\left(\initiall\right)}`{\FforREC^3\left(\initiall\right)};{\FforREC^2\left(\coproje{}\right)}]
  \morphism(1260,0)/->/<300,0>[{\FforREC^3\left(\initiall\right)}`{{\cdots}};]
}
\defdiag{diag-chain-of-ep-pairs}{   
  \morphism(0,0)|a|/->/<300,0>[{\epzeroo}`{\objectD{A}{1}};{\mochain{0}}]
  \morphism(300,0)|a|/->/<300,0>[{\objectD{A}{1}}`{\objectD{A}{2}};{\mochain{1}}]
  \morphism(600,0)|a|/->/<300,0>[{\objectD{A}{2}}`{\objectD{A}{3}};{\mochain{2}}]
  \morphism(900,0)|a|/->/<270,0>[{\objectD{A}{3}}`{{\cdots}};{\mochain{3}}]
}
\defdiag{roll-eq2-diagram}{   
  \morphism(0,0)|a|/->/<1125,0>[{{\left({{\catV}'}^\op\times{\catV}'\right)^{n-1}}}`{{\left({{\catV}'}^\op\times{\catV}'\right)^n}};{{\pairL{\id},\pE{\fixpointRecT{E'}}{{\catV}'}^\op{,}\pE{\fixpointRecT{E'}}{{\catV}'}\pairR}}]
  \morphism(0,0)|b|/->/<1125,-450>[{{\left({{\catV}'}^\op\times{\catV}'\right)^{n-1}}}`{{{\catV}'}};{{\pE{\fixpointRecT{E'}}{{\catV}'}}}]
  \morphism(1125,0)|r|/->/<0,-450>[{{\left({{\catV}'}^\op\times{\catV}'\right)^n}}`{{{\catV}'}};{{\pE{E}{{\catV}'}}}]
  \morphism(0,-450)|l|/->/<0,450>[{{\left({\catV}^\op\times\catV\right)^{n-1}}}`{{\left({{\catV}'}^\op\times{\catV}'\right)^{n-1}}};{{\left({H}^\op\times{H}\right)^{n-1}}}]
  \morphism(675,-225)|a|/<=/<375,0>[{\phantom{O}}`{\phantom{O}};{{\rollReT}^{\pEE{E'}}}]
}
\defdiag{roll-basic-diagram}{   
  \morphism(0,0)|a|/->/<1125,0>[{{\left({\catV}^\op\times\catV\right)^{n}}}`{{\left({\catV}^\op\times\catV\right)^{n+1}}};{{\pairL{\id},\pE{\fixpointRecT{E}}{\catV}^\op{,}\pE{\fixpointRecT{E}}{\catV}\pairR}}]
  \morphism(0,0)|b|/->/<1125,-450>[{{\left({\catV}^\op\times\catV\right)^{n}}}`{{\catV}};{{\pE{\fixpointRecT{E}}{\catV}}}]
  \morphism(1125,0)|r|/->/<0,-450>[{{\left({\catV}^\op\times\catV\right)^{n+1}}}`{{\catV}};{{\pE{E}{\catV}}}]
  \morphism(1125,-450)|b|/->/<-1125,0>[{{\catV}}`{{\catC}};{{J}}]
  \morphism(675,-225)|a|/<=/<375,0>[{\phantom{O}}`{\phantom{O}};{{\rollReT}^{\pEE{E}}}]
}
\defdiag{diag-H-compatibility-of-paramatric-types}{   
  \morphism(0,0)|a|/->/<1125,0>[{{\left({\catV}^\op\times\catV\right)^{n}}}`{{\catV}};{{\pE{E}{{\catV}}}}]
  \morphism(0,0)|b|/->/<0,-300>[{{\left({\catV}^\op\times\catV\right)^{n}}}`{{\left({\catV'}^\op\times\catV'\right)^n}};{{\left(H^\op\times{H}\right)^n}}]
  \morphism(1125,0)|r|/->/<0,-300>[{{\catV}}`{{{\catV}'}};{{H}}]
  \morphism(0,-300)|b|/->/<1125,0>[{{\left({\catV'}^\op\times\catV'\right)^n}}`{{{\catV}'}};{{\pE{E'}{{\catV}'}}}]
}
\defdiag{basic-logicalrelations}{   
  \morphism(0,0)|a|/->/<2025,0>[{{\left(\SynV,{\SynT},{\Synfix},\Synit\right)}}`{{\left(\SynV,{\SynT},{\Synfix},\Synit\right)\times\left(\SynVt{,}\SynTt{,}\Synfixt{,}\Synitt\right)}};{{\left(\ID{,}\DSyn\right)}}]
  \morphism(2025,0)|r|/->/<0,-300>[{{\left(\SynV,{\SynT},{\Synfix},\Synit\right)\times\left(\SynVt{,}\SynTt{,}\Synfixt{,}\Synitt\right)}}`{{\CBVU\left(\wCpo\times\wCpo{,}\monadwP{-}\right)}};{{\sem{-}\times\semt{k}{-}}}]
  \morphism(0,0)|l|/->/<0,-300>[{{\left(\SynV,{\SynT},{\Synfix},\Synit\right)}}`{{\CBVU\left(\SUBscone{\wCpo}{\sconeFUNCTOR{n,k}},\monadLR{n,k}{-}\right)}};{{\semLR{n,k}{-}}}]
  \morphism(0,-300)|b|/->/<2025,0>[{{\CBVU\left(\SUBscone{\wCpo}{\sconeFUNCTOR{n,k}},\monadLR{n,k}{-}\right)}}`{{\CBVU\left(\wCpo\times\wCpo{,}\monadwP{-}\right)}};{{\CBVU\left({\forgetfulSub}_{n,k}\right)}}]
}
\defdiag{basic-logicalrelations-recursive-types}{   
  \morphism(0,0)|a|/->/<2025,0>[{{\left(\SynVr,\SynTr{,}\SynffixpointRecT\right)}}`{{{\left(\SynVr,\SynTr{,}\SynffixpointRecT\right)}\times\left(\SynVrt{,}\SynTrt{,}\SyntffixpointRecT\right)}};{{\left(\ID{,}\DSynrec\right)}}]
  \morphism(2025,0)|r|/->/<0,-300>[{{{\left(\SynVr,\SynTr{,}\SynffixpointRecT\right)}\times\left(\SynVrt{,}\SynTrt{,}\SyntffixpointRecT\right)}}`{{\rCBVU\left(\wCpo\times\wCpo{,}\monadwP{-}\right)}};{{\sem{-}\times\semt{k}{-}}}]
  \morphism(0,0)|l|/->/<0,-300>[{{\left(\SynVr,\SynTr{,}\SynffixpointRecT\right)}}`{{\rCBVU\left(\SUBscone{\wCpo}{\sconeFUNCTOR{n,k}},\monadLR{n,k}{-}\right)}};{{\semLR{n,k}{-}}}]
  \morphism(0,-300)|b|/->/<2025,0>[{{\rCBVU\left(\SUBscone{\wCpo}{\sconeFUNCTOR{n,k}},\monadLR{n,k}{-}\right)}}`{{\rCBVU\left(\wCpo\times\wCpo{,}\monadwP{-}\right)}};{{\rCBVU\left({\forgetfulSub}_{n,k}\right)}}]
}
\defdiag{forgetful-subscone}{   
  \morphism(0,0)/->/<525,0>[{{\SUBscone{\catD}{G}}}`{{\catD\downarrow{G}}};]
  \morphism(525,0)|a|/->/<450,0>[{{\catD\downarrow{G}}}`{{\catD\times{\catB}}};{{\forgetfulS}}]
  \morphism(975,0)|a|/->/<375,0>[{{\catD\times{\catB}}}`{{\catB}};{{\pi_{\catB}}}]
}
\defdiag{forgetful-indentity-on-C-monad-subscone}{   
  \morphism(150,0)|a|/->/<600,0>[{{\catD\downarrow{G}}}`{{\catD\downarrow{G}}};{\monadSub}]
  \morphism(750,0)|r|/->/<150,-300>[{{\catD\downarrow{G}}}`{{\catD\times{\catB}}};{{\forgetfulS}}]
  \morphism(150,0)|l|/->/<-150,-300>[{{\catD\downarrow{G}}}`{{\catD\times{\catB}}};{{\forgetfulS}}]
  \morphism(0,-300)|b|/->/<450,0>[{{\catD\times{\catB}}}`{{\catB}};{{\pi_{\catB}}}]
  \morphism(900,-300)|b|/->/<-450,0>[{{\catD\times{\catB}}}`{{\catB}};{{\pi_{\catB}}}]
}
\defdiag{kleisli-morphism-pair}{   
  \morphism(0,-225)|l|/->/<0,225>[{\catV}`{\catC};{J}]
  \morphism(1500,-225)|r|/->/<0,225>[{{\catV}'}`{{\catC}'};{J'}]
  \morphism(0,0)|a|/->/<1500,0>[{\catC}`{{\catC}'};{\lift{H}}]
  \morphism(0,-225)|b|/->/<1500,0>[{\catV}`{{\catV}'};{H}]
}
\defdiag{obvious-diagram-parametric-types}{   
  \morphism(0,-450)|l|/->/<0,450>[{{\left({\catV}^\op\times\catV\right)^{n}}}`{{\left({\catC}^\op\times\catC\right)^{n}}};{{\left({J^\op}\times{J}\right)^{n}}}]
  \morphism(675,-450)|a|/->/<0,450>[{{\catV}}`{{\catC}};{{J}}]
  \morphism(0,-450)|b|/->/<675,0>[{{\left({\catV}^\op\times\catV\right)^{n}}}`{{\catV}};{{\pE{E}{\catV}}}]
  \morphism(0,0)|a|/->/<675,0>[{{\left({\catC}^\op\times\catC\right)^{n}}}`{{\catC}};{{\pE{E}{\catC}}}]
}
\defdiag{obvious-diagram-parametric-types-n}{   
  \morphism(0,-450)|l|/->/<0,450>[{{\left({\catV}^\op\times\catV\right)^{n+1}}}`{{\left({\catC}^\op\times\catC\right)^{n+1}}};{{\left({J^\op}\times{J}\right)^{n+1}}}]
  \morphism(675,-450)|a|/->/<0,450>[{{\catV}}`{{\catC}};{{J}}]
  \morphism(0,-450)|b|/->/<675,0>[{{\left({\catV}^\op\times\catV\right)^{n+1}}}`{{\catV}};{{\pE{E}{\catV}}}]
  \morphism(0,0)|a|/->/<675,0>[{{\left({\catC}^\op\times\catC\right)^{n+1}}}`{{\catC}};{{\pE{E}{\catC}}}]
}
\defdiag{obvious-diagram-parametric-types-(n-1)-recursive}{   
  \morphism(0,-450)|l|/->/<0,450>[{{\left({\catV}^\op\times\catV\right)^{n}}}`{{\left({\catC}^\op\times\catC\right)^{n}}};{{\left({J^\op}\times{J}\right)^{n}}}]
  \morphism(675,-450)|a|/->/<0,450>[{{\catV}}`{{\catC}};{{J}}]
  \morphism(0,-450)|b|/->/<675,0>[{{\left({\catV}^\op\times\catV\right)^{n}}}`{{\catV}};{{\pE{\fixpointRecT{E}}{\catV}}}]
  \morphism(0,0)|a|/->/<675,0>[{{\left({\catC}^\op\times\catC\right)^{n}}}`{{\catC}};{{\pE{\fixpointRecT{{E}}}{\catC}}}]
}
\defdiag{wCpo-category-of-morphisms-pullback-comma}{   
  \morphism(0,0)|l|/->/<0,-450>[{{\catD\downarrow{G}}}`{{\catC}};{{\cCPULL}}]
  \morphism(0,0)|a|/->/<675,0>[{{\catD\downarrow{G}}}`{{\morphismcatTwo\pitchfork\catD}};{{\morPULL}}]
  \morphism(675,0)|r|/->/<0,-450>[{{\morphismcatTwo\pitchfork\catD}}`{{\catD}};{{\codomm}}]
  \morphism(0,-450)|b|/->/<675,0>[{{\catC}}`{{\catD}};{{G}}]
}
\defdiag{wCpo-category-of-morphisms}{   
  \morphism(0,0)|l|/->/<0,-225>[{D_0}`{{D_1}};{{f}}]
  \morphism(0,0)|a|/->/<1200,0>[{D_0}`{{D_0}'};{{\alpha}_0}]
  \morphism(1200,0)|r|/->/<0,-225>[{{D_0}'}`{{D_1'}};{{g}}]
  \morphism(0,-225)|b|/->/<1200,0>[{{D_1}}`{{D_1'}};{{\alpha}_1}]
}
\defdiag{morphism-comma-category}{   
  \morphism(0,0)|l|/->/<0,-225>[{D}`{{G(C)}};{{j}}]
  \morphism(0,0)|a|/->/<1200,0>[{D}`{{D}'};{{{\alpha}_0}}]
  \morphism(1200,0)|r|/->/<0,-225>[{{D}'}`{{G(C')}};{{h}}]
  \morphism(0,-225)|b|/->/<1200,0>[{{G(C)}}`{{G(C')}};{{G\left(\alpha_1\right)}}]
}
\defdiag{basic-wcpo-chain}{   
  \morphism(0,0)/->/<300,0>[{\leastelement}`{{q\left(\leastelement\right)}};]
  \morphism(300,0)/->/<300,0>[{{q\left(\leastelement\right)}}`{\cdots};]
  \morphism(600,0)/->/<300,0>[{\cdots}`{{q^n\left({\leastelement}\right)}};]
  \morphism(900,0)/->/<300,0>[{{q^n\left({\leastelement}\right)}}`{\cdots};]
}
\defdiag{strong-monad-compatibility-identity}{   
  \morphism(0,0)|l|/->/<0,-900>[{TW}`{\terminall\times{TW}};{\cong}]
  \morphism(0,0)|r|/->/<675,-900>[{TW}`{T(\terminall\times{W})};{\cong}]
  \morphism(0,-900)|b|/->/<675,0>[{\terminall\times{TW}}`{T(\terminall\times{W})};{\strengthT_{\terminall{,}W}}]
}
\defdiag{strong-monad-compatibility-associativity}{   
  \morphism(0,0)|a|/->/<1050,0>[{\left({W\times{Y}}\right)\times{TZ}}`{W\times\left({Y}\times{TZ}\right)};{\cong}]
  \morphism(1050,0)|r|/->/<0,-450>[{W\times\left({Y}\times{TZ}\right)}`{W\times{T\left({Y\times{Z}}\right)}};{W\times{\strengthT_{Y,Z}}}]
  \morphism(1050,-450)|r|/->/<0,-450>[{W\times{T\left({Y\times{Z}}\right)}}`{T\left({W}\times\left({Y\times{Z}}\right)\right)};{\strengthT_{W,Y\times{Z}}}]
  \morphism(0,0)|l|/->/<0,-900>[{\left({W\times{Y}}\right)\times{TZ}}`{T\left(\left({W\times{Y}}\right)\times{Z}\right)};{\strengthT_{W\times{Y},Z}}]
  \morphism(0,-900)|l|/->/<1050,0>[{T\left(\left({W\times{Y}}\right)\times{Z}\right)}`{T\left({W}\times\left({Y\times{Z}}\right)\right)};{\cong}]
}
\defdiag{strong-monad-compatibility-unit}{   
  \morphism(0,0)|l|/->/<0,-900>[{W\times{Y}}`{W\times{TY}};{W\times\ee_{Y}}]
  \morphism(0,0)|r|/->/<675,-900>[{W\times{Y}}`{T\left(W\times{Y}\right)};{\ee_{W\times{Y}}}]
  \morphism(0,-900)|b|/->/<675,0>[{W\times{TY}}`{T\left(W\times{Y}\right)};{\strengthT_{W{,}Y}}]
}
\defdiag{strong-monad-compatibility-multiplication}{   
  \morphism(0,0)|a|/->/<1050,0>[{W\times{TTY}}`{W\times{TY}};{W\times{\mm_Y}}]
  \morphism(1050,0)|r|/->/<0,-900>[{W\times{TY}}`{T\left({W\times{Y}}\right)};{W\times{\strengthT_{Y,Z}}}]
  \morphism(0,-900)|r|/->/<1050,0>[{TT\left({W\times{Y}}\right)}`{T\left({W\times{Y}}\right)};{\mm_{W\times{Y}}}]
  \morphism(0,0)|l|/->/<0,-450>[{W\times{TTY}}`{T\left(W\times{TY}\right)};{\strengthT_{W,TY}}]
  \morphism(0,-450)|l|/->/<0,-450>[{T\left(W\times{TY}\right)}`{TT\left({W\times{Y}}\right)};{T\left(\strengthT_{W,Y}\right)}]
}
\defdiag{twocellpushoutdefinitionrightside}{   
  \morphism(675,0)|l|/{@{->}@/_25pt/}/<0,-525>[{b_1}`{y};{h_0}]
  \morphism(675,0)|r|/{@{->}@/^25pt/}/<0,-525>[{b_1}`{y};{h_0'}]
  \morphism(0,0)|l|/->/<0,-525>[{e}`{b_0};{p_0}]
  \morphism(0,0)|a|/->/<675,0>[{e}`{b_1};{p_1}]
  \morphism(0,-525)|r|/->/<675,0>[{b_0}`{y};{h_1}]
  \morphism(98,-262)/=/<180,0>[{\phantom{O}}`{\phantom{O}};]
  \morphism(525,-262)|a|/=>/<300,0>[{\phantom{O}}`{\phantom{O}};{\xi_0}]
}
\defdiag{twocellofpushoutdefinitionleftside}{   
  \morphism(0,-525)|a|/{@{->}@/^25pt/}/<675,0>[{b_0}`{y};{h_1'}]
  \morphism(0,0)|a|/->/<675,0>[{e}`{b_1};{p_1}]
  \morphism(0,-525)|b|/{@{->}@/_25pt/}/<675,0>[{b_0}`{y};{h_1}]
  \morphism(0,-525)|l|/<-/<0,525>[{b_0}`{e};{p_0}]
  \morphism(675,-525)|r|/<-/<0,525>[{y}`{b_1};{h_0'}]
  \morphism(338,-22)/=/<0,-180>[{\phantom{O}}`{\phantom{O}};]
  \morphism(338,-375)|r|/<=/<0,-300>[{\phantom{O}}`{\phantom{O}};{\xi_1}]
}

\def\pu{}
\fi	




\title[AD for ML-family languages: correctness via logical relations]{Automatic Differentiation for ML-family languages: correctness via logical relations}

\author{Fernando Lucatelli Nunes}
\affiliation{
	\department{Department of Information and Computing Sciences} 
	\institution{Utrecht University}          
	\country{Netherlands}                    
}
\email{f.lucatellinunes@uu.nl}

\author{Matthijs Vákár}
\affiliation{
	\department{Department of Information and Computing Sciences} 
	\institution{Utrecht University}          
	\country{Netherlands}                    
}
\email{m.i.l.vakar@uu.nl}

\settocdepth{part}


\begin{abstract}
We give a simple, direct and reusable logical relations technique for languages with term and type recursion and partially defined differentiable functions. We demonstrate it by working out the case of  Automatic Differentiation (AD) correctness: namely, we present a correctness proof of a dual numbers style AD code transformation for realistic functional languages in the ML-family. We also show how this code transformation provides us with  correct forward- and {reverse-mode} AD.

The starting point is to interpret a functional programming language as a suitable freely generated categorical structure. In this setting, by the universal property of the syntactic categorical structure, the dual numbers AD code transformation and the basic $\wCpo$-semantics arise as structure preserving functors. The proof follows, then, by a novel logical relations argument.

The key to much of our contribution is a powerful monadic logical relations technique for term recursion and recursive types. It provides us with a semantic correctness proof based  on a simple approach for denotational semantics, making use only of the very basic concrete model of \wcpos. 
\end{abstract}

\maketitle


%
\pu

%
\pu

%
\pu

%
\pu

%
\pu

%
\pu

%
\pu

%
\pu

%
\pu

%
\pu

%
\pu

%
\pu

%
\pu

%
\pu

%
\pu

%
\pu

%
\pu

%
\pu

%
\pu

%
\pu

%
\pu

%
\pu

%
\pu

%
\pu

%
\pu

%
\pu

%
\pu
%
\pu


\tableofcontents
\newpage

\setcounter{secnumdepth}{-1}

\input{TEX/introduction}

\resettocdepth

\setcounter{secnumdepth}{5}
\clearpage
\input{TEX/Keyideas}

\input{TEX/overview}

\input{TEX/BasicSemantics}

\input{TEX/DomainSemanticsSimplified}

\input{TEX/coarse-language}

\input{TEX/dualnumbers_macro}

\input{TEX/Semantics_Language}

\input{TEX/sconing}

\input{TEX/CorrectnessAD}

\input{TEX/recursive-types}

\input{TEX/almost-everywhere-differentiability}

\input{TEX/related-work}

\begin{acks}                           
	This project has received funding via NWO Veni grant number VI.Veni.202.124
	as well as the European Union’s Horizon 2020 research and innovation
	programme under the Marie Skłodowska-Curie grant agreement No. 895827. 
	
	This research was supported through the programme ``Oberwolfach Leibniz Fellows'' by the Mathematisches Forschungsinstitut Oberwolfach in 2022. It was also partially supported  by the CMUC, Centre for Mathematics of the University of Coimbra - UIDB/00324/2020, funded by the Portuguese Government through FCT/MCTES. 

  We are grateful to the anonymous reviewers for their helpful comments on this manuscript.
\end{acks}

\bibliographystyle{plain-abb}
\bibliography{TEX/bibliography}

\clearpage
\appendix 

\input{TEX/fine-language}
\input{TEX/efficient-sign-derivative}

\input{TEX/wCPO-enriched-scone}

\newpage

\input{TEX/app-rnn}




\pu

\end{document}



%% file: TEX/defs.tex
\usepackage{xcolor}
\usepackage{amsmath,amsthm}
\usepackage{mathtools}
\usepackage{suffix}
\usepackage{amsfonts}
\usepackage{mathpartir}
\usepackage{enumitem}
\usepackage{stmaryrd}
\usepackage{twoopt}
\usepackage{amsthm}
\usepackage{okcategory}

\usepackage{tocvsec2}  

\usepackage{listings}

\AtEndPreamble{%
	\theoremstyle{acmdefinition}
    \newtheorem{remark}[theorem]{Remark}}
\AtEndPreamble{%
	\theoremstyle{acmdefinition}
	\newtheorem{assum}[theorem]{Assumption}}

\newcommand\gdefinedby{::=}

\newcommand{\morPULL}{\mathsf{proj}_{2\pitchfork\catD } }
\newcommand{\cCPULL}{\mathsf{proj}_{\catC}}

\newcommand{\KI}{\mathfrak{L} }

\newcommand{\defeqq}{\stackrel {\mathrm{def}}=}
\newcommand{\defeq}{\stackrel {\mathrm{def}}=}

\newcommand\FforREC{\mathfrak{F}_X}
\newcommand\limm[2]{ \mathsf{lim}\left( #1 , #2 \right) }
\newcommand{\initiall}{\mathsf{0}}
\newcommand{\epzeroo}{\mathfrak{O}}

\newcommand{\uniqMor}{\iota}
\newcommand{\terminall}{\mathsf{1}}
\newcommand{\strengthT}{\mathsf{s}}

\newcommand{\wcpo}{$\omega$-cpo}
\newcommand{\wcpos}{$\omega$-cpos}
\newcommand{\wcont}{$\omega$-continuous}
\newcommand{\wCpo}{\mathbf{\boldsymbol\omega Cpo}}
\renewcommand{\wCPO}{\mathbf{\boldsymbol\omega CPO}}

\newcommand{\ww}{\omega}
\newcommand\LRMONAD{\lift{\monadT}}

\newcommand\epto{\stackrel{\hookrightarrow}{\leftharpoondown}}

\newcommand\eppair[1]{\left( \emb #1 , \prj #1 \right) }
\newcommand\emb[1]{#1^e}
\newcommand\prj[1]{#1^p}

\newcommand\lift[1]{\overline{#1}}
\newcommand\unlift[1]{\underline{#1}}

\newcommand\unlifteppair[1]{\hat{{#1}}}
\newcommand\secondM[1]{#1 ^\ast }

\newcommand\colim{\mathrm{colim}}

\newcommand\semanticshandler[2]{\mathfrak{p}_{#1 \to #2} }

\newcommand\ifelse[3]{\mathbf{if}\,#1\,\mathbf{then}\,#2\,\mathbf{else}\,#3\,}
\newcommand\letin[3]{\mathbf{let}\,#1=\,#2\,\mathbf{in}\,#3}
\newcommand\letrec[4]{\mathbf{let} \,\mathbf{rec}\,#1\!\left( #2\right) = #3\,\mathbf{in}\,#4}

\newcommand\intle[1]{\varphi_{#1}}

\newcommand\coproje[1]{\iota _{#1} }

\newcommand\cpairL{[ }
\newcommand\cpairR{] }

\newcommand\TyAlph[1]{
	\ifcase #1\or \tau\or \sigma\or \rho\else \@ctrerr \fi%
}
\newcommand\ty[1]{{\TyAlph{#1}}}

\newcommand\TySchAlph[1]{
	\ifcase #1\or \phi\or \psi \or \upsilon\else \@ctrerr \fi%
}

\newcommand\var[1]{{\VarAlph{#1}}}
\newcommand\VarAlph[1]{%
	\ifcase #1\or x\or y\or z\else \@ctrerr \fi%
}

\newcommand\tvar[1]{{\TVarAlph{#1}}}
\newcommand\TVarAlph[1]{%
	\ifcase #1\or \alpha\or \beta\or \gamma\else \@ctrerr \fi%
}

\newcommand\val[1]{%
	\ifcase #1\or v\or w\or u\else \@ctrerr \fi%
}

\newcommand\trm[1]{{\TermAlph{#1}}}
\newcommand\TermAlph[1]{%
	\ifcase #1\or t\or s\or r\else \@ctrerr \fi%
}

\newcommand\monadLR[2]{\mathcal{P}_{#1}\monadwP{#2} }
\newcommand\openLift[1]{\mathsf{Diff}_{\left( #1 \right) } } 
\newcommand\topologyO[1]{\mathfrak{O}_{#1}}

\newcommand\morLRmonad[1]{\mathtt{j}_{#1}}

\newcommand\wrapSyncat[1]{\mathbf{wrap}_{#1}}
\newcommand\wrapSemcat[1]{\mathtt{wrap}_{#1}}

\newcommand\ee{\eta } 
\newcommand\mm{\mathrm{m}} 
\newcommand\ID{\mathrm{id}} 

\newcommand\RRsyntax{\mathtt{R}}
\newcommand\monadT{\mathcal{T}} 
\newcommand\monadSub{\mathfrak{T}_{sub}}

\newcommand\monadS{\mathcal{S}} 
\newcommand\monadwP[1]{\left( #1\right) _\leastelement } 

\newcommand\monadphi{\varphi } 

\newcommand\catV{\Cat{V}}
\newcommand\catW{\Cat{W}}
\newcommand\catZ{\Cat{Z}}
\newcommand\catB{\Cat{B}}
\newcommand\catC{\Cat{C}}
\newcommand\catD{\Cat{D}}
\newcommand\catE{\Cat{E}}
\newcommand\catCat{\mathbf{Cat}}
\newcommand\catPos{\mathbf{Pos}}

\newcommand\ecat[1]{{#1}\textrm{-}\catCat}

\newcommand\incCBVtorCBV{\mathtt{s} }
\newcommand\incTCBVtorCBV{\mathtt{s}^\mathsf{t}  }

\newcommand\RCBVcat{\mathfrak{C}_{\mathcal{RBV}} }
\newcommand\CBVcat{\mathfrak{C}_{\mathcal{BV}} }
\newcommand\CBVcatund{\mathfrak{C}_{\mathtt{p}} }

\newcommand\wCBVcat{\wCPO\textrm{-}\mathfrak{C}_{\mathcal{BV}} }

\newcommand\wrCBVcat{\wCPO\textrm{-}\mathfrak{C}_{r\mathcal{BV}} }

\newcommand\morphismcatTwo{\mathsf{2}}

\newcommand\SUBscone[2]{\mathbf{Sub}\left( #1\downarrow #2 \right)}

\newcommand\obb[1]{\mathsf{ob}\, #1}

\newcommand\canonicalbasise[2]{e^{#1}_{#2}}
\newcommand\sconeFUNCTOR[1]{G_{#1}}

\newcommand\chang[1]{\mathfrak{G}_{#1}}

\newcommand\rCBVtoCBV{\mathcal{R} }

\newcommand\CBVU{\mathcal{U}_{\mathcal{BV}} }
\newcommand\rCBVU{\mathcal{U}_{r\mathcal{BV}} }
\newcommand\CBVUp{\mathcal{U}_{\mathtt{p}} }
\newcommand\CBVUrp{\mathcal{U}_{r\mathtt{p}} }

\newcommand\forgetfulS{\mathcal{L}}

\newcommand\forgetfulSub{\underline{\mathcal{L}}}

\newcommand\codomm{\mathrm{codom}}
\newcommand\domm{\mathrm{dom}}
\newcommand\ihom[3]{#1\left[ #2 , #3 \right] }
\newcommand\ihomC{\ihom{\catC} }
\newcommand\ihomV{\ihom{\catV} }

\newcommand\ehom[3]{#1\left( #2 , #3 \right) }
\newcommand\ehomC{\ehom{\catC} }
\newcommand\ehomV{\ehom{\catV} }

\newcommand\expo[2]{\left( #1 \Rightarrow #2 \right) }
\newcommand\expk[2]{\left( #1 \Rightarrow ^k #2 \right) }

\newcommand\kleislitimes{\otimes}

\newcommand\op{\mathrm{op} }  

\newcommand\eva[2]{\mathsf{eval}_{ #1 , #2  } }

\newcommand\rolling{\unlift{\rollReT} }
\newcommand\wrolling{\unlift{\rollReT}_{\omega } }

\newcommand\paE[3]{{#1}^{#2}_{#3}}
\newcommand\pE[2]{{#1}_{#2}}
\newcommand\mind[1]{\treeindexing{m}{#1}}
\newcommand\sind[1]{\treeindexing{s}{#1}}

\newcommand\uF{\tilde{F}}
\newcommand\uE{\tilde{E}}

\newcommand\treeindexing[2]{\mathsf{#1}_{\left( #2\right) }}
\newcommand\TreeIndex{\mathsf{Tree}}

\newcommand\pEE[1]{#1}
\newcommand\fddt{\textit{fddt} }
\newcommand\ParamVd[1]{\mathfrak{P}^{\mathfrak{d}}\left( #1 \right)  }
\newcommand\Paramd[2]{\mathsf{Param}^{\mathfrak{d}}\left( #1 , #2\right)  }
\newcommand\ParamAll[2]{\mathsf{Param}\left( #1 , #2\right)  }
\newcommand\ParamAlll{\mathsf{Param}}
\newcommand\prjEDchain[2]{\mathcal{P}^{\pEE{#1}}_{#2} }

\newcommand\pEDchain[2]{\mathcal{E}^{\pEE{#1}}_{#2} }
\newcommand\pEDchainn[1]{\mathcal{E}^{\pEE{#1}} }

\newcommand\epchain[1]{a_{#1} }
\newcommand\mochain[1]{\emb{\epchain{#1}} }
\newcommand\comochain[1]{\prj{\epchain{#1}}}
\newcommand\objectD[2]{\mathfrak{#1}_{#2}}

\newcommand\rollReT{\mathsf{roll} } 
\newcommand\wrollReT{\omega\mathsf{roll} } 
\newcommand\pwrollReT[1]{\omega\mathsf{roll}^{#1}  } 

\newcommand\ffixpointRecT{\unlift{\fixpointRecT} } 
\newcommand\wffixpointRecT{\unlift{\fixpointRecT}_{\omega } }

\newcommand\fixpointRecT{\nu }
\newcommand\wfixpointRecT{\nu _{\omega} }

\newcommand\fixpoint{\mu }  

\newcommand\iterationn{{\mathsf{itt}}}  
\newcommand\wfixpoint{\mu _ \omega }  
\newcommand\witerate{\mathsf{it}_\omega }  

\newcommand\pwfixpoint{\overline{\mu} }  
\newcommand\pwiterate{\overline{\mathsf{it}} }  

\newcommand\leastelement{\bot}                              
\newcommand\lfpoint[1]{\mathrm{lfp}\left({#1}\right)}       

\newcommand\diagk[1]{\mathrm{diag}_{#1} }
\newcommand\proj[1]{\pi_{#1}}            
\newcommand\projY{\proj{Y}}
\newcommand\diagv[1]{\partial_{#1} }
\newcommand\pairL{\left( }
\newcommand\pairR{\right) }

\newcommand{\RR}{\mathbb{R}}
\newcommand{\NN}{\mathbb{N}}
\newcommand\NNN[1]{\mathbb{I}_{#1}}

\newcommand{\signR}{\mathsf{sign}}
\newcommand\semanc[1]{\mathsf{#1} }
\newcommand\seman[1]{ {#1}_\op }

\newcommand\targetreals{\reals }  
\newcommand\tangentprojection[2]{\mathfrak{h} _{#1} #2 }
\newcommand\tangentreals{\mathbf{vect} }  
\renewcommand\reals{\mathbf{real}}            
\newcommand\ints{\mathbf{int}}   
\newcommand\cnst[1]{\underline{#1}}           
\newcommand\cnzero{\overline{0}}
\newcommand\cncanoni[1]{\overline{e} _{#1}}

\newcommand\Op{\mathrm{Op}}                   

\newcommand\syncat[1]{\mspace{-25mu}\synname{#1}}
\newcommand\synname[1]{\qquad\text{#1}}

\newenvironment{syntax}[1][]{%
	\(
	\begin{array}[t]{#1l@{\quad\!\!}*3{l@{}}@{\,}l}
	}{
	\end{array}
	\)%
}

\newcommand\gor{\mathrel{\lvert}}
\newcommand\vor{\mathrel{\big\lvert}}

\newcommand\Init{\mathbf{0}}
\newcommand\tReturn[1]{\mathbf{return}\,#1}
\newcommand\toin[3]{#2\,\mathbf{to}\,#1.\,#3}
\newcommand\tTuple[1]{\langle #1\rangle}
\newcommand\tUnit{\tTuple{\,}}

\newcommand\pMatch[5][\,]{\mathbf{case}\,#2\,\mathbf{of}#1\tPair{#3}{#4}\To#5}
\newcommand\tMatch[4][\,]{\mathbf{case}\,#2\,\mathbf{of}#1\tTuple{#3}\To#4}
\newcommand\vMatch[3][\,]{\mathbf{case}\,#2\,\mathbf{of}#1\{#3\}}
\newcommand\nvMatch[2][\,]{\mathbf{case}\,#2\,\mathbf{of}#1\{\,\}}
\newcommand\bvMatch[6][\,]{\mathbf{case}\,#2\,\mathbf{of}#1\{\tInl#3\To #4\, \mid \tInr#5\To #6\}}
\newcommand\rMatch[4][\,]{\mathbf{case}\,#2\,\mathbf{of}#1\mathbf{roll}\,#3\To#4}
\newcommand\tItFrom[3]{\mathbf{iterate}\,#1\,\mathbf{from}\,#2=#3}
\newcommand\To{\to}
\newcommand\tRoll{\mathbf{roll}\,}
\newcommand\tUnroll{\mathbf{unroll}\,}
\newcommand\tInl{\mathbf{inl}\,}
\newcommand\tInr{\mathbf{inr}\,}
\newcommand\tFst{\mathbf{fst}\,}
\newcommand\tSnd{\mathbf{snd}\,}
\newcommand\tPair[2]{\langle #1, #2\rangle}
\newcommand\rec[2]{\mu#1.#2}
\newcommand\trec[2]{\mathbf{\mu}#1.#2}

\WithSuffix\newcommand\t*{\,{\mathop{\times}}\,}
\WithSuffix\newcommand\t+{\,{\mathop{\sqcup}}\,}
\newcommand\tSign{\mathbf{sign}\,}
\newcommand\Unit{\mathbf{1}}
\newcommand\fun[1]{\lambda #1.}

\newcommand\zero{ { 0}^{\mathbf v}}
\newcommand\plus{{+}^\mathbf v}

\newcommand\subst[2]{#1{}[#2]}
\newcommand\sfor[2]{^{#2}\!/\!_{#1}}

\newcommand\tinf{\vdash}
\newcommand\ctx{\Gamma}
\newcommand\Ginf[3][]{\ctx #1\tinf #2 : #3}
\newcommand\Ginfv[3][]{\ctx #1\tinf^v #2 : #3}

\newcommand\Ginfc[3][]{\ctx #1\tinf^c #2 : #3}

\newcommand\DGinf[3][]{\Delta\mid\ctx #1\tinf #2 : #3}

\newcommand\freeeq[1]{\overset{\# #1}{\equiv}}
\newcommand\beeq{{\equiv}}

\newcommand\Dsynsymbol[1][]{\scalebox{0.8}{$\mathcal{D}$}_{#1}}
\newcommand\Dsyn[2][]{\Dsynsymbol[#1](#2)}
\newcommand\Dsynrevsymbol[1][]{\scalebox{0.8}{$\overleftarrow{\mathcal{D}}$}_{#1}}
\newcommand\Dsynrev[2][]{\Dsynrevsymbol[#1](#2)}

\newcommand\vectoraslineartransformation[1]{\tilde{#1}}

\newcommand\Dsynplainsymbol[1][]{\scalebox{0.8}{${\mathcal{D}}$}_{#1}}
\newcommand\Dsynplain[2][]{\Dsynplainsymbol[#1](#2)}

\newcommand\DsynV[2][]{\Dsynsymbol[#1]{}_{\Cat{V}}(#2)}
\newcommand\DsynC[2][]{\Dsynsymbol[#1]{}_{\Cat{C}}(#2)}

\newcommand\DSyn{\mathbb{D}}
\newcommand\DSynrec{\mathbb{ID}}

\newcommand\dDSyn[1]{\overrightarrow{d}#1 }

\newcommand\dDSemtotaltra[2]{\mathfrak{D}^{#1}{#2} }

\newcommand\dDSemtra[2]{\mathfrak{d}^{#1}\left( #2\right)  }
\newcommand\dDSemj[2]{\mathfrak{d}_{#2}\left( #1\right)  }

\newcommand\SynVt{\Syn_V^{\target}}

\newcommand\SynTt{\Syn_\monadS^{\target}}
\newcommand\Synfixt{\Syn_\mu^{\target} }  
\newcommand\Synitt{\Syn_\mathsf{it}^{\target} }  

\newcommand\target{{\mathbf{tr}}}
\newcommand\recSyn{{\mathsf{R}}}

\newcommand\SynVr{\Syn_V^\recSyn}
\newcommand\SynCr{\Syn_C^\recSyn}
\newcommand\SynTr{\Syn_\monadS^\recSyn}

\newcommand\SynffixpointRecT{\ffixpointRecT_\Syn}
\newcommand\fddtpointRecT{\fixpointRecT_\Syn}

\newcommand\SynVrt{\Syn_V^{\recSyn\target}}

\newcommand\SynTrt{\Syn_\monadS^{\recSyn\target}}
\newcommand\SyntffixpointRecT{\ffixpointRecT^{\target}_\Syn}

\newcommand\Syn{\mathbf{Syn}}
\newcommand\SynV{\Syn_V}
\newcommand\SynC{\Syn_C}
\newcommand\SynJ{\Syn_J}
\newcommand\SynG{\Syn_G}
\newcommand\SynT{\Syn_\monadS}

\newcommand\Synfix{\Syn_\mu }  
\newcommand\Synit{\Syn_\mathsf{it} }  

\newcommand\semt[2]{[\hspace{-2.5pt}[#2]\hspace{-2.5pt}] _ {#1} }
\newcommand\sem[1]{[\hspace{-2.5pt}[#1]\hspace{-2.5pt}] }
\newcommand\semLR[2]{ \overline{[\hspace{-2.5pt}\sem{#2}\hspace{-2.5pt}]}_{#1} }
\newcommand\semLRG[2]{ \overline{[\hspace{-2.5pt}\sem{#2}\hspace{-2.5pt}]}_{#1}}

\newcommand\Lift[1]{{#1}_{\bot}}

\newcommand\SScone{\mathbf{SScone}}
\newcommand\Array[1]{#1[]}

\newcommand\Wrap[1]{\mathrm{Wrap}_{#1}}

\newcommand\lincomp[2]{#1*^{\mathbf v} #2}

\newcommand\typeproj[1]{\mathbf{p}_{#1}}
\newcommand\typezero[1]{\mathbf{z}_{#1}}

\newcommand\semroll{\mathrm{roll}}

\newcommand\extendedH[1]{\mathcal{#1}}

\makeatletter
\let\c@equation\c@figure    
\makeatother

%% file: TEX/introduction.tex
\section*{Introduction}
\subsubsection*{Logical relations}
Logical relations arguments (see e.g. \cite{mitchell1992notes} for a survey) are proof techniques 
that can be used to demonstrate properties of typed programming languages,
ranging from strong normalisation to canonicity and adequacy. 
The arguments are essentially type-guided forms of induction.
They seem to have been reinvented several times by different research communities, and are also known under various other names, including Tait's method of computability, reducibility candidates, Artin gluing, Sierpinski cone, (sub)sconing, and Freyd cover.

Category theory gives a useful way to organise logical relations arguments: by viewing them as ways of building a new categorical semantics of a programming language out of an existing ones.
The new semantics then equips objects with predicates of some form, and restricts the morphisms to those morphisms that respect the predicates.
By choosing the right notion of predicates, we can ensure that the existence of this new semantics gives us the property we are hoping to prove about our programming language.

In this paper we present novel logical relations methods for  languages with recursive features, together with an application of these techniques to correctness proofs for Automatic Differentiation.

\subsubsection*{AD and the PL community}
Automatic Differentiation (AD, see e.g. \cite{griewank2008evaluating} for a survey)  is a popular family of techniques for computing derivatives of 
functions implemented by a piece of code, particularly when efficiency, scaling to high dimensions, and numerical stability are important.
It has been studied in the scientific computing community for many decades and has been heavily used in machine learning for the last decade. 
In the last years, the programming languages (PL) community has turned towards studying AD from a new perspective.
Much progress has been made towards giving a formulation of (forward and) reverse mode AD that
\begin{enumerate}
\item is simple and purely functional;\label{point1-AD} 
\item scales to the expressive ML-family functional languages that are popular in practice;\label{point2-AD} 
\item admits a simple correctness proof that shows that AD computes the derivative;\label{point3-AD} 
\item provably has the correct asymptotic complexity and is performant in practice;\label{point4-AD}
\item is parallelism preserving. \label{point5-AD}
\end{enumerate}

\subsubsection*{Our contributions}
In this paper, we present a simple solution to problems \eqref{point1-AD}-\eqref{point3-AD}, \emph{our first major contribution}.

We give a proof of the correctness of the reverse and forward mode dual numbers style Automatic Differentiation (AD)  in a semantically unified way, making use only of the very simple concrete denotational model of \wcpos{}.
In doing so, we simplify existing techniques that relied on sheaf-theoretic machinery \cite{vakar2020denotational,huot2023omegapap}.

A key challenge that we tackle to achieve the correctness proofs of this paper is to have sufficiently strong categorical logical relations techniques for reasoning about partially defined differentiable functions, and term and type recursion. 
We believe that our novel methods can be simpler than existing alternatives such as \cite{pitts1996relational,DBLP:conf/esop/Ahmed06}, and they are still widely applicable, \emph{our second major contribution}.

We refer to the companion paper \cite{smeding2022} for a performant implementation of the dual numbers reverse-mode AD technique proved correct in the present paper. It shows that it efficiently differentiates most of Haskell98, contributing towards point \eqref{point4-AD}.
Parallelism preservation (point \eqref{point5-AD}) for this AD technique is discussed in \cite{smeding2024parallel}.

In our work, we ensure to keep all constructions sufficiently simple such that they can easily be generalized to more advanced AD algorithms such as CHAD \cite{vakar2021chad,DBLP:journals/toplas/VakarS22,VAKAR-LUCATELLI2021}, which is one of our key motivations for this work.

\subsubsection*{Why care and why is this difficult?}
Given the central role that AD plays in modern scientific computing and
machine learning, the ideal of differential programming has been emerging \cite{meijer2018behind,plotkin2018some}:
compilers for general purpose programming languages should provide built-in support for automatic differentiation of any programs written in the language.
Such general purpose programming languages tend to include many language features, however, which we then need to be able to differentiate.
What a correct and efficient notion of derivative is of such features might not be so straightforward as they often go beyond what is studied in traditional calculus.
In this paper we focus on the challenge posed, in particular, by partial language features: partial primitive operations, lazy conditionals on real numbers, iteration, recursion and recursive types.

Partial primitive operations are certainly key. Indeed, even the basic operations of division and logarithm are examples.
(Lazy) conditionals on real numbers are useful in practice for pasting together various existing smooth functions, a basic example being the ReLU function
$$
ReLU(x) \defeq \ifelse{x}{0}{x}=\vMatch{(\tSign\,x)}{\tInl\_\To 0\mid \tInr\_\To x},
$$
which is a key component of many neural networks.
Conditionals are also frequently used in probabilistic programming to paste together density functions of different distributions \cite{betancourt_2019}.
People have long studied the subtle issue of how one should algorithmically differentiate such functions with ``kinks'' under the name of \emph{the if-problem in automatic differentiation} \cite{beck1994if}.  
Our solution is the one also employed  by \cite{abadi-plotkin2020}: to treat the functions as semantically undefined at their kinks (at $x=0$ in the case of $ReLU(x)$).
This is justified given how coarse the semantic treatment of floating point numbers as real numbers is already.
Our semantics based on partial functions defined on real numbers is sufficient to prove many high-level correctness properties. 
However, like any semantics based on real numbers, it fails to capture many of the low-level subtleties introduced by the floating point implementation.
Our key insight that we use to prove correctness of AD of partial programs is to construct a suitable lifting of the partiality monad to a variant of \cite{hsv-fossacs2020}'s category of $\RR^k$-indexed logical relations used to relate programs to their derivatives.
This particular monad lifting for derivatives of partial functions can be seen as our solution to the if-problem in AD.
In Section \ref{sec:almost-everywhere-differentiability}, we briefly discuss how the more ambitious solution to the if-problem in the style of \cite{lee2020correctness,DBLP:journals/pacmpl/MazzaP21,huot2023omegapap} can also be achieved with our methods.
In that solution, we show that the set of non-differentiable points where AD does not compute a correct derivative is of measure zero, which we achieve by choosing a different monad-lifting.

Similarly, iteration constructs, or while-loops, are necessary for implementing iterative algorithms with dynamic stopping criteria.
Such algorithms are frequently used in programs that AD is applied to.
For example, AD is applied to iterative differential equation solvers to perform Bayesian inference in
SIR models\footnote{For this particular case of an iterative algorithm it is actually possible and better (but more laborious!) to implement a custom derivative rather than differentiating through the while-loop \cite{margossian2021efficient}.}.
This technique played a key role in modelling the Covid19-pandemic \cite{flaxman2020estimating}.
For similar reasons, AD through iterative differential equation solvers 
is important for probabilistic modelling of  pharmacokinetics \cite{tsiros2019population}.
Other common use-cases of iterative algorithms that need to be
AD'ed are eigen-decompositions and algebraic equation solvers, such as those employed in Stan \cite{carpenter2015stan}. 
Finally, iteration gives a convenient way of achieving numerically stable approximations to complex functions (such as the  Conway-Maxwell-Poisson density function \cite{goodrich_2017}).
The idea is to construct, using iteration, a Taylor approximation that terminates once the next term in the series causes floating-point underflow.
Indeed, for a function whose $i$-th terms in the Taylor expansion can be represented by a program
$$i : \ints, x : \reals \vdash
t(i, x) : \reals,$$
we would define the underflow-truncated Taylor series by
\begin{equation}\label{eqn:taylor-series-underflow-example}\tItFrom{\Big(\begin{array}{l}
\pMatch{ z}{i}{y'}{}
\letin{y}{
t(i, x)}{}\\
\vMatch{-\epsilon < y < \epsilon}{\tInl\_ \To \tInr x\mid \tInr\_\To \tInl \tPair{i + 1}{y +y'}}\end{array}\Big)}
{z}{\tPair{0}{0}},
\end{equation}
where $\epsilon$ is a cut-off for floating-point underflow.

Next, recursive neural networks \cite{tai2015improved} are often mentioned as a use case of AD applied to recursive programs.
While basic Child-Sum Tree-LSTMs can also be implemented with
primitive recursion (a fold) over an inductively defined tree (which can be defined as a recursive type), there are other related models such
as Top-Down-Tree-LSTMs that require an iterative or general recursive approach \cite{zhang2016top}.
In fact, \cite{jeong2018improving} has shown that a recursive approach is preferable as it better
exposes the available parallelism in the model.
In Appendix \ref{ap:rnn}, we show some Haskell code for the recursive neural network of \cite{socher2011parsing}, to give an idea of how iteration and recursive types (in the form of inductive types of labelled trees) naturally arise in a functional implementation of such neural net architectures.
We imagine that more applications of AD applied to recursive programs with naturally will emerge as the technique becomes available to machine learning researchers and engineers.
Finally, we speculate that coinductive types like streams of real numbers, which can be encoded
using recursive types as $\mu \alpha.\Unit \To (\reals* \alpha)$, provide a useful API for on-line machine learning applications \cite{shalev2012online}, where data is processed in real time as it becomes available.
Recursion and more notably recursive types introduce one final challenge into the correctness proof of AD of such expressive functional programs: the required logical relations arguments are notoriously technical, limiting the audience of any work using them and frustrating application to more complicated AD algorithms like CHAD.
To mend this problem, we introduce a novel, simple but powerful logical relations technique for open semantic logical relations for recursive types.

\paragraph{Prerequisites}
We assume some familiarity with
 category theory  (see, for instance, \cite{mac2013categories}): the concepts of and basic facts about
categories, functors, natural transformations, (co)limits, adjunctions, and (co) monads.
We also assume that the reader knows the most basic definitions in enriched category theory (see, for instance, \cite{kelly1982basic}): the concepts of $\catV$-categories, $\catV$-functors, and $\catV$-natural transformations.
We recall the definitions and results we need for $\catV$-monads and their Keisli  $\catV$-categories (the interested reader can find more details in \cite{MR0280560}).
Later in this paper, we will also consider $\catV$-(co)limits, $\catV$-adjunctions, and $\catV$-(co)monadicity, but only for the specific case of $\catV=\wCpo$ with its cartesian structure.
In these cases, we ensure to spell out all details to make the paper as self-contained as possible.

\paragraph{Convention}
Whenever we talk about
\textit{strict preservation of some structure} (like products, coproducts or exponentials), we are assuming that we have chosen structures (chosen products, coproducts or exponentials) and the preservation is on the nose, that is to say, the canonical comparison is the identity.

%% file: TEX/Keyideas.tex
\section{Key ideas}
In this paper, we consider how to perform forward and reverse mode 
dual numbers automatic differentiation on a functional language with expressive partial features,
by using a dual numbers technique.

\subsection*{Language}
We consider an idealised functional language with product types $\ty{1}\t*\ty{2}$, sum types
$\ty{1}\t+\ty{2}$, function types $\ty{1}\To\ty{2}$ generated by 
\begin{itemize}
\item a primitive \emph{type $\reals$ of real numbers} (in practice, implemented as floating point numbers);
\item \emph{constants} $\vdash \cnst{c}:\reals$ for $c\in \RR$; 
\item sets $(\Op_n)_{n\in \NN}$ of $n$-ary \emph{primitive operations} $\op$, for which we include computations\\  $\var{1}_1:\reals,\ldots, \var{1}_n:\reals\vdash \op(\var{1}_1,\ldots,\var{1}_n):\reals$;
we think of these as implementing partial functions $\RR^n\rightharpoonup \RR$ with open domain of definition, on which they are differentiable;
for example, we can include mathematical operations $\log,\exp\in\Op_1$ and $(+),(*),(/)\in\Op_2$;
\item a construct $\var{1}:\reals\vdash \tSign(\var{1}):\Unit\t+\Unit$,
where we write $\Unit$ for the empty product; $\tSign{\trm{1}}$ computes the \emph{sign of a real number} $\trm{1}$ and is undefined at $\trm{1}=\cnst{0}$; we can use it to define a lazy conditional on real numbers
$\ifelse{\trm{3}}{\trm{1}}{\trm{2}}\defeq 
\vMatch{\tSign\trm{3}}{{
	\_\To{\trm{1}}
	\vor \_\To{\trm{3}}
}}$ of the kind that is often used in AD libraries like Stan \cite{carpenter2015stan}.
\end{itemize}
Next, we include two more standard mechanisms for defining partial functions:
\begin{itemize}
\item \emph{(purely functional) iteration}: given a computation $\Gamma, \var{1} : \ty{1} \vdash \trm{1} : \ty{1}\t+\ty{2}$ to iterate and a
starting value $\Gamma\vdash \trm{2} : \ty{1}$, we have a computation $\Gamma\vdash \tItFrom{\trm{1}}{\var{1}}{\trm{2}} : \ty{2}$ which
repeatedly calls $\trm{1}$, starting from the value of $\trm{2}$ until the result lies in $\ty{2}$;
\item \emph{recursion}: given a computation $\Gamma,\var{1}:\ty{1}\To\ty{2}\vdash \trm{1}:\ty{1}\To\ty{2}$,
we have a program $\Gamma\vdash\rec{\var{1}}{\trm{1}}:\ty{1}\to\ty{2}$ that recursively computes to $\letin{\var{1}}{\rec{\var{1}}{\trm{1}}}{\trm{1}}$;
note that we can define iteration with recursion.
\end{itemize}

\subsection*{Dual numbers forward AD code transform}
Let us assume that we have programs $\partial_i\op(\var{1}_1,\ldots,\var{1}_n)$ that compute the 
$i$-th partial derivative of each $n$-ary primitive operation $\op$.
For example, we can define $\partial_1(*)(\var{1}_1,\var{1}_2)=\var{1}_2$ and $\partial_2(*)(\var{1}_1,\var{1}_2)=\var{1}_1$.
Then, we can define a very straightforward forward mode AD code transformation 
$\Dsynsymbol$ 
by replacing all primitive types $\reals$ by a pair $\Dsyn{\reals}\defeq \reals\t*\reals$ 
of reals
and by replacing all constants $\cnst{c}$,  $n$-ary primitive operations $\op$ and sign function $\tSign$ 
in the program as\footnote{Actually, while our definition for $\Dsyn{\tSign{\trm{3}}}$ given here is correct, there exist more efficient implementation techniques, as we discuss in Appendix \ref{sec:efficient-sign-derivative}.}
\[
    \begin{array}{ll}
        \Dsyn{\cnst{c}} \defeq & \tPair{\cnst{c}}{\cnst{0}}\\
        \Dsyn{\op(\trm{3}_1,\ldots,\trm{3}_n)}\defeq~
                           &\pMatch{\Dsyn{\trm{3}_1}}{\var{1}_1}{\var{1}_1'}
                           { \ldots \to\pMatch{\Dsyn{\trm{3}_n}}{\var{1}_n}{\var{1}_n'}
                           {\\
                            &{\tPair{ \op(\var{1}_1,\ldots,\var{1}_n)}{\var{1}_1' *\partial_1\op(\var{1}_1,\ldots,\var{1}_n)+\ldots
                           +\var{1}_n' *\partial_n\op(\var{1}_1,\ldots,\var{1}_n)}}}}\\
        \Dsyn{\tSign{\trm{3}}}\defeq & \tSign{(\tFst\Dsyn{\trm{3}})}.
    \end{array} 
\]
We extend $\Dsynsymbol$ to all other types and programs in the unique homomorphic (structure preserving way),
by using structural recursion.
So, for example, $\Dsyn{\ty{1}\To\ty{2}}\defeq \Dsyn{\ty{1}}\To\Dsyn{\ty{2}}$, $\Dsyn{\var{1}}\defeq \var{1}$, $\Dsyn{\letin{\var{1}}{\trm{1}}{\trm{2}}}=\letin{\var{1}}{\Dsyn{\trm{1}}}{\Dsyn{\trm{2}}}$ and $\Dsyn{\trm{1}\,\trm{2}}=\Dsyn{\trm{1}}\,\Dsyn{\trm{2}}$.
We like to think of $\Dsynsymbol$ as a structure preserving functor $\Dsynsymbol:\Syn\to\Syn$ on the syntax.

\subsection*{Semantics}
To formulate correctness of the AD transformation $\Dsynsymbol$, we need to assign a formal 
denotational semantics $\sem{-}$ to our language.
We use the standard interpretation of types $\ty{1}$ as \wcpos{} $\sem{\ty{1}}$ (partially ordered sets with suprema of countable chains) and programs $\var{1}_1:\ty{1}_1,\ldots,\var{1}_n:\ty{1}_n\vdash \trm{1}:\ty{2}$ as monotone \wcont{} partial functions $\sem{\trm{1}}:\sem{\ty{1}_1}\times\cdots\times \sem{\ty{1}_n}\rightharpoonup\sem{\ty{2}}$.
We interpret $\reals$ as the discrete \wcpo{} $\sem{\reals}\defeq \RR$ of real numbers, in which $r\leq r'$ if and only if $r=r'$.
We interpret ${\cnst{c}}$ as the constant $\sem{\cnst{c}}\defeq c\in\RR$.
We interpret $\op$ as the partial differentiable function $\sem{\op(\var{1}_1,\ldots,\var{1}_n)}:\RR^n\rightharpoonup \RR$ that it is intended to implement.
And, finally, we interpret $\tSign$ as the partial function $\sem{\tSign(\var{1})}:\RR\rightharpoonup \Unit \sqcup \Unit$ that sends $r<0$ to the left copy of $\Unit$, $r>0$ to the right copy and is undefined for $r=0$.
Having fixed these definitions, the rest of the semantics is entirely compositional and standard.
In particular, we interpret iteration and recursion using Kleene's Fixpoint Theorem.
We think of this semantics as a structure preserving functor $\sem{-}:\Syn\to\wCpo$ from the syntax 
 to the category of \wcpos{} and monotone \wcont{} functions.

\subsection*{Correctness statement}
Having defined a semantics, we can phrase what it means for $\Dsynsymbol$ to be correct.
We prove the following, showing that $\Dsyn{\trm{1}}$ implements the usual calculus derivative $D\sem{\trm{1}}$ of $\sem{\trm{1}}$.
\begin{theorem}[Forward AD Correctness, Theorem \ref{theo:main-theorem-section-proof} with $k=1$ in main text]
For any program $\var{1}:\ty{1}\vdash \trm{1}:\ty{2}$ for $\ty{1}=\reals^{m},\ty{2}=\reals^l$ (where we write $\reals^n$ for the type $\reals\t*\cdots\t*\reals$ of length $n$ tuples of reals), we have that $\sem{\trm{1}}$ is differentiable on its domain and
\begin{align*}
&\sem{\Dsyn{\trm{1}}}((x_{1},v_{1}),\ldots,(x_{m},v_{m}))=\\
&\Big(\pi_1(\sem{\trm{1}}(x_1,\ldots,x_m)), \pi_1(D\sem{\trm{1}}((x_{1},\ldots,x_m),(v_{1},\ldots,v_m))),\ldots,\\
&\;\;\pi_l(\sem{\trm{1}}(x_1,\ldots,x_m)), \pi_l(D\sem{\trm{1}}((x_{1},\ldots,x_m),(v_{1},\ldots,v_m)))
\Big)
\end{align*}
for any $(x_1,\ldots,x_m)$ in the domain of definition of $\sem{\trm{1}}$ and any tangent vector $(v_1,\ldots, v_m)$ to $\sem{\ty{1}}$ at $x$. 
\end{theorem}
Importantly, the program $\trm{1}$ might use higher-order functions, iteration, recursion, etc..
In fact,  we also establish the theorem above for general types $\ty{1}$ and $\ty{2}$ not containing function types, but its phrasing requires slight bookkeeping that might distract from the simplicity of the theorem.

\subsection*{A proof via logical relations}
The proof of the correctness theorem follows a logical relations argument
that we found using categorical methods, but which can be phrased entirely in elementary terms.
Let us fix some $n\in \NN$.
We define for all types $\ty{1}$ of our language, by induction, relations $T^{n}_{\ty{1}}\subseteq 
(\RR^n\To\sem{\ty{1}})\times ((\RR^n\times \RR^n)\To \sem{\Dsyn{\ty{1}}})$
and 
$P^{n}_{\ty{1}}\subseteq 
(\RR^n\rightharpoonup{\sem{\ty{1}}})\times ((\RR^n\times \RR^n)\rightharpoonup {\sem{\Dsyn{\ty{1}}}})$
that relate a (partial) $n$-curve to its derivative $n$-curve:
\begin{align*}
T^{n}_{\reals}&\defeq \set{(\gamma,\gamma')\mid \gamma\text{ is differentiable and } \gamma'=(x,v)\mapsto (\gamma(x),D\gamma(x,v))}\\
T^{n}_{\ty{1}\t*\ty{2}} & \defeq \set{(x\mapsto (\gamma_1(x),\gamma_2(x)), (x,v)\mapsto(\gamma'_1(x,v),\gamma'_2(x,v)))\mid (\gamma_1,\gamma'_1)\in T^{n}_{\ty{1}}\text{ and } (\gamma_2,\gamma'_2)\in T^{n}_{\ty{2}}}\\
T^{n}_{\ty{1}\t+\ty{2}} & \defeq \set{(\iota_1\circ \gamma_1,\iota_1\circ \gamma'_1)\mid (\gamma_1,\gamma'_1)\in T^{n}_{\ty{1}}}\cup \set{(\iota_2\circ \gamma_2,\iota_2\circ \gamma'_2)\mid (\gamma_2,\gamma'_2)\in T^{n}_{\ty{2}}}\\
T^{n}_{\ty{1}\To\ty{2}}& \defeq \set{(\gamma,\gamma')\mid \forall (\delta, \delta')\in T^{n}_{\ty{1}}. (x\mapsto \gamma(x)(\delta(x)), (x,v)\mapsto \gamma'(x,v)(\delta'(x,v)))\in P^{n}_{\ty{2}}}\\
P^{n}_{\ty{1}} & \defeq \Big\{(\gamma,\gamma')\mid \gamma^{-1}(\sem{\ty{1}})\times \RR^n=\gamma'^{-1}(\sem{\Dsyn{\ty{1}}})\text{ is open and for all differentiable}\\&\qquad\qquad\quad \delta:\RR^n\to \gamma^{-1}(\sem{\ty{1}}) \text{ we have }
 (\gamma\circ\delta,(x,v)\mapsto (\gamma(\delta(x)),\gamma'(D\delta(x,v)))) \in T^{n}_{\ty{1}}\Big\}.
\end{align*}
We then prove the following ``fundamental lemma'', using induction on the typing derivation of $\trm{1}$:
\begin{quote}
If $\var{1}_1:\ty{1}_1,\ldots,\var{1}_n:\ty{1}_n\vdash \trm{1}:\ty{2}$ and, for $1\leq i\leq n$,
$(f_i, f_i')\in T^{n}_{\ty{1}_i}$, then \\
$(x\mapsto \sem{\trm{1}}(f_1(x),\ldots, f_n(x)), (x,v)\mapsto \sem{\Dsyn{\trm{1}}}(f_1'(x,v),\ldots, f_n'(x,v)))\in P^{n}_{\ty{2}}$.
\end{quote}
For example, we use that, by assumption, $\sem{\partial_i\op(\var{1}_1,\ldots,\var{1}_n)}$ equals the $i$-th partial derivative of $\sem{\op(\var{1}_1,\ldots,\var{1}_n)}$ combined with the chain-rule, to show that primitive operations $\op$ respect the logical relations.
Crucial features to enable the inductive steps for iteration and recursion in the proof of the fundamental lemma are that $T^n_{\reals}$ and $P^n_{\ty{1}}$ are closed under suprema of countable chains and that $P^n_{\ty{1}}$ contains the least element. 

As $T^{k}_{\reals^k}$ contains, in particular, $(\id, ((x_1,\ldots,x_k), (v_1,\ldots, v_k))\mapsto ((x_1,v_1),\ldots ,(x_n,v_k)))$, our correctness theorem follows.

\subsection*{Extending to recursive types via a novel categorical logical relations technique}
Next, we extend our language with ML-style polymorphism and recursive types.
That is, we allow the formation of types $\ty{1}$ with free type variables $\tvar{1}$ and we include 
a type variable binder $\trec{\tvar{1}}{\ty{1}}$, which binds $\tvar{1}$ in $\ty{1}$ and computes a canonical fixpoint of $\tvar{1}\mapsto\ty{1}$.
We extend our AD transformation homomorphically on terms and types. For example, on types, we define
\begin{align*}
\Dsyn{\tvar{1}}\defeq\tvar{1}&& \Dsyn{\trec{\tvar{1}}{\ty{1}}}\defeq\trec{\tvar{1}}{\Dsyn{\ty{1}}}.
\end{align*}
A type $\ty{1}$ with $n$ free type variables gets interpreted in our \wcpo{-}semantics as an $n$-ary mixed-variance endofunctor $\sem{\ty{1}}$ on the 
category of \wcpos{} and partial morphisms that restricts to that of \wcpos{} and total morphisms.
Programs with types that have free variables get interpreted as (di)natural transformations.
As the
category of \wcpos{} and partial morphisms has the structure to interpret recursive types, we have a canonical \emph{minimal invariant} 
$$
\semroll:\sem{\ty{1}}(\mu\sem{\ty{1}}, \mu\sem{\ty{1}})\xto{\cong}\mu\sem{\ty{1}}
$$
for the mixed-variance endofunctors $\sem{\ty{1}}$ on $\wCpo$ that types $\ty{1}$ denote \cite{levy2012call}.
We interpret $\sem{\trec{\tvar{1}}{\ty{1}}}\defeq \mu \sem{\ty{1}}$.

To extend the correctness proof to this larger language, we would like to define the 
logical relation
$$
T^n_{\trec{\tvar{1}}{\ty{1}}}\defeq \set{(\semroll\circ\gamma,\semroll\circ\gamma')\mid (\gamma,\gamma')\in T^n_{\subst{\ty{1}}{\sfor{\tvar{1}}{\trec{\tvar{1}}{\ty{1}}}}}}.
$$
That is, we would like to be able to \emph{define relations using type recursion}.
If we can do so, then extending the proof of the fundamental lemma is straightforward.
We can then establish the correctness theorem also for $\ty{1}$ and $\ty{2}$ that involve recursive types.

The traditional method is to follow the technical recipes of \cite{pitts1996relational}.
Instead, we develop a powerful new logical relations technique for recursive types, which we believe
to be more conceptually clear and easier to use in situations like ours.
To be precise, we prove a general result saying that under mild conditions, that we can interpret recursive types in the category of logical relations over a category that models recursive types itself. 
For simplicity, we state an important special case that we need for our application here.

Given a right adjoint $\wCpo$-enriched functor $G:\wCpo^n\to \wCpo$ (such as $G(X)=\wCpo^n(Y, X)$ for some $Y\in \wCpo^n$), consider the category 
$\SScone$ of $n$-ary logical relations, which has objects $(X, T)$, where $X\in\wCpo^n$ and $T$ is a  (full) sub-$\omega$-cpo  of 
$GX$, and morphisms $(T,X)\to (T',X')$ are $\wCpo^n$-morphisms $f:X\to X'$ such that $y\in T$ implies $Gf(y)\in T'$.
\begin{theorem}[Logical relations for recursive types, special case of Theorem \ref{the:maybe-the-main-result-on-LR-recursive} in main text]
Let $L$ be a strong monad on $\SScone$ that lifts the usual partiality monad $\Lift{(-)}$ on $\wCpo^n$ 
along the projection functor $\SScone \to \wCpo^n$. We assume that $L$ takes the initial object to the terminal one, and that $G(\eta_X)^{-1}(L(T,X))=T$, where we write $\eta_X:X\to\Lift{X}$ for the unit of the partiality monad on $\wCpo^n$.
Then, the Kleisli functor $\SScone\hookrightarrow \SScone_L$ gives a model for recursive types.
\end{theorem}
Spelled out in non-categorical terms, we are considering logical relations 
\[
T_{\ty{1}}\subseteq G(\sem{\ty{1}})\qquad\qquad P_{\ty{1}}=L(T_{\ty{1}},\sem{\ty{1}})\subseteq G(\Lift{\sem{\ty{1}}})
\]
and we require that the relation $P_{0}$ at the initial object (empty type) is precisely the singleton set $\{\bot\}$ (containing the least element) and $G({\sem{\ty{1}}}\subseteq \Lift{{\sem{\ty{1}}}})^{-1}(P_{\ty{1}})$ (which we think of as the total elements in $P_{\ty{1}}$) coincide with $T_{\ty{1}}$.

In particular, in our case, we work binary logical relations ($n=2$) with $$G(X,X')=\wCpo^2((\RR^n,\RR^n\times\RR^n), (X,X'))$$ and the monad lifting
\begin{align*}L(T,(X,X')) &=  \Big\{(\gamma,\gamma')\mid \gamma^{-1}(X)\times \RR^n=\gamma'^{-1}(X')\text{ is a proper open subset and for all differentiable}\\&\qquad\qquad\quad \delta:\RR^n\to \gamma^{-1}(X) \text{ we have }
 (\gamma\circ\delta,(x,v)\mapsto (\gamma(\delta(x)),\gamma'(D\delta(x,v)))) \in T\Big\}\\
 &\phantom{==}\cup T.
\end{align*}
Consequently, we can define the logical relations $T_{\trec{\tvar{1}}{\ty{1}}}$ using type recursion, as desired.

\subsection*{Dual numbers reverse AD}
Similarly to dual numbers forward AD $\Dsynsymbol$, we can define a 
reverse AD code transformation $\Dsynrevsymbol$: we define 
$
\Dsynrev{\reals}\defeq \reals\t*\tangentreals  
$
and 
\begin{align*}
\Dsynrev{\cnst{c}}&\defeq \tPair{\cnst{c}}{\zero}\\
\Dsynrev{\op(\trm{1}_1,\ldots,\trm{1}_n)}
&\defeq\pMatch{\Dsynrev{\trm{1}_1}}{\var{1}_1}{\var{1}_1'}{\ldots}
\pMatch{\Dsynrev{\trm{1}_n}}{\var{1}_n}{\var{1}_n'}{}\\
&\quad\quad\langle\op(\var{1}_1,\ldots,\var{1}_n),
\lincomp{\var{1}_1'}{\partial_1\op(\var{1}_1,\ldots,\var{1}_n)} \plus\ldots
 \plus
\lincomp{\var{1}_n'}{\partial_n\op(\var{1}_1,\ldots,\var{1}_n)} \rangle \\
\Dsynrev{\tSign{\trm{1}}} &\defeq \tSign(\tFst\Dsynrev{\trm{1}}).
\end{align*}
and extend homomorphically to all other type and term formers, as we did before.
In fact, this algorithm is exactly the same as dual numbers forward AD 
in code with the only differences being that 
\begin{enumerate}
\item the type $\reals$ of real numbers for tangents has been replaced with a new type 
$\tangentreals  $, which we think of as representing (dynamically sized) cotangent vectors to the global input of the program;
\item the zero $\cnst 0 $ and addition $(+)$ of type $\reals$ have been replaced by the zero $\zero$ and addition $(\plus)$ of cotangents of type $\tangentreals  $;
\item the multiplication $(*):\reals\t*\reals\to \reals$ has been replaced by the operation 
$(\lincomp{~}{~}):\tangentreals  \t* \reals \to\tangentreals  $: $(\lincomp{v}{r})$ is the rescaling of a cotangent $v$ by the scalar $r$.
\end{enumerate}
We write $\cncanoni{i}$ for program representing the $i$-th canonical basis vector $e_i$ of type $\tangentreals  $
and we write 
\begin{equation}\label{eq:wrap-k-x-Dk} 
\Wrap{s}(\var{1})\defeq \tMatch{\var{1}}{\var{1}_1,\ldots,\var{1}_s}{\tTuple{\tPair{\var{1}_1}{\cncanoni{1}},\ldots,\tPair{\var{1}_s}{\cncanoni{s}}}}.
\end{equation} 
We define $\sem{\tangentreals  }\defeq \RR^\infty\defeq\sum_{k=0}^\infty \RR^k$ as the infinite (vector space) coproduct of $k$-dimensional real vector spaces.
That is, we interpret $\tangentreals$ as the type of dynamically sized real vectors\footnote{Note that, in practice, \cite{smeding2022} actually implements $\tangentreals$ as a type of 
ASTs of simple expressions computing a dynamically sized vector.
This allows us to first build up the expression during execution of the program (the forward pass)
and to only evaluate this cotangent expression later (in a reverse pass) making clever use of a distributivity law of addition and multiplication (also known as the linear factoring rule in \cite{brunel2019backpropagation}) to achieve the correct computational complexity of reverse AD.
}.
We show that $\Dsynrev{\trm{1}}$ implements the transposed derivative $D\sem{\trm{1}}^t$ of $\sem{\trm{1}}$ in the following sense.
\begin{theorem}[Reverse AD Correctness, Theorem \ref{theo:main-theorem-section-proof} with $k=\infty$ in main text]
    For any program $\var{1}:\ty{1}\vdash \trm{1}:\ty{2}$ for $\ty{1}=\reals^{s},\ty{2}=\reals^l$,
    \begin{align*}    &\sem{\letin{\var{1}}{\Wrap{k}(\var{1})}{\Dsynrev{\trm{1}}}}(x_1,\ldots, x_s)=\\
   &\Big((\pi_1(\sem{\trm{1}}(x_1,\ldots,x_s)),
      D\sem{\trm{1}}^t((x_{1},\ldots,x_s),e_1)),\ldots, (\pi_l(\sem{\trm{1}}(x_1,\ldots,x_s)), D\sem{\trm{1}}^t((x_{1},\ldots,x_s),e_l))\Big)
    \end{align*}
    for any $(x_1,\ldots,x_s)$ in the domain of definition of $\sem{\trm{1}}$. 
    \end{theorem}
We prove this theorem again using a similar logical relations argument, defining $T^{n}_{\ty{1}}\subseteq 
(\RR^n\To\sem{\ty{1}})\times ((\RR^n\times {(\RR^\infty )}^n)\To \sem{\Dsynrev{\ty{1}}})$
and 
$P^{n}_{\ty{1}}\subseteq 
(\RR^n\rightharpoonup{\sem{\ty{1}}})\times (\RR^n\times  {(\RR^\infty)}^n)\rightharpoonup {\sem{\Dsynrev{\ty{1}}}})$ as before for all types $\ty{1}$ of language, setting 
\begin{align*}
T^{n}_{\reals}&\defeq \set{(\gamma, \gamma')\mid \gamma\text{ is differentiable and }\gamma'=(x, L)\mapsto (\gamma(x),  L(D\gamma^t(x, e_1)))}\\
T^{n}_{\ty{1}\t*\ty{2}} & \defeq \set{(x \mapsto (\gamma_1(x),\gamma_2(x)), (x,L)\mapsto(\gamma'_1(x,L),\gamma'_2(x,L)))\mid (\gamma_1,\gamma'_1)\in T^{n}_{\ty{1}}\text{ and } (\gamma_2,\gamma'_2)\in T^{n}_{\ty{2}}}\\
T^{n}_{\ty{1}\t+\ty{2}} & \defeq \set{(\iota_1\circ \gamma_1,\iota_1\circ \gamma'_1)\mid (\gamma_1,\gamma'_1)\in T^{n}_{\ty{1}}}\cup \set{(\iota_2\circ \gamma_2,\iota_2\circ \gamma'_2)\mid (\gamma_2,\gamma'_2)\in T^{n}_{\ty{2}}}\\
T^{n}_{\ty{1}\To\ty{2}}& \defeq \set{(\gamma,\gamma')\mid \forall (\delta, \delta')\in T^{n}_{\ty{1}}. (x\mapsto \gamma(x)(\delta(x)), (x,L)\mapsto \gamma'(x,L)(\delta'(x,L)))\in P^{n}_{\ty{2}}}\\
T^n_{\trec{\tvar{1}}{\ty{1}}}&\defeq \set{(\semroll\circ\gamma,\semroll\circ\gamma')\mid (\gamma,\gamma')\in T^n_{\subst{\ty{1}}{\sfor{\tvar{1}}{\trec{\tvar{1}}{\ty{1}}}}}}
\\
P^{n}_{\ty{1}} & \defeq \Big\{(\gamma,\gamma')\mid \gamma^{-1}(\sem{\ty{1}})\times (\RR^\infty){}^n=\gamma'^{-1}(\sem{\Dsynrev{\ty{1}}})\text{ is open and for all differentiable}\\&\qquad\quad\!\!\! \delta:\RR^n\to \gamma^{-1}(\sem{\ty{1}}) \text{ we have }
 (\gamma\circ\delta,(x,L)\mapsto
 \gamma'(\delta(x),L\circ D\delta^t(x,-))) \in T^{n}_{\ty{1}}\Big\},
\end{align*}
where we consider $(\RR^\infty){}^n$ as a type of linear transformations from $\RR^n$ to $\RR^\infty$ and we write $e_i$ for the $i$-th standard basis vector of $\RR^n$.
We then prove the following ``fundamental lemma'', using induction on the typing derivation of $\trm{1}$:
\begin{quote}
If $\var{1}_1:\ty{1}_1,\ldots,\var{1}_n:\ty{1}_n\vdash \trm{1}:\ty{2}$ and, for $1\leq i\leq n$,
$(f_i, f_i')\in T^{n}_{\ty{1}_i}$, then \\
$(x\mapsto \sem{\trm{1}}(f_1(x),\ldots, f_n(x)), (x,L)\mapsto \sem{\Dsynrev{\trm{1}}}(f_1'(x,L),\ldots, f_n'(x,L)))\in P^{n}_{\ty{2}}$.
\end{quote}

As $T^{s}_{\reals^s}$ contains, in particular, $$(\id, ((x_1,\ldots,x_s), L)\mapsto ((x_1,L e_1),\ldots ,(x_s,L e_s))),$$ our theorem follows.

\subsection*{Extending to arrays}
AD tends to be applied to programs that manipulate large arrays of reals.
Seeing that such arrays are denotationally equivalent to lists $\trec{\tvar{1}}{\Unit\t+\tvar{1}\t*\reals}$,
while only the computational complexity of operations differs, our correctness result also applies to functional languages with arrays.
We thus differentiate array types $\Array{\ty{1}}$ with elements of type $\ty{1}$ in the obvious 
structure preserving way, e.g.
\begin{align*}
\Dsyn{\Array{\ty{1}}}\defeq \Array{\Dsyn{\ty{1}}}   \qquad \Dsyn{\mathbf{generate}}\defeq \mathbf{generate}\qquad \Dsyn{\mathbf{map}}\defeq\mathbf{map}\qquad \Dsyn{\mathbf{foldr}}\defeq \mathbf{foldr}
\end{align*}
and similarly for dual numbers reverse AD.

\subsection*{Almost everywhere differentiability}
Taking inspiration from \cite{lee2020correctness,DBLP:journals/pacmpl/MazzaP21,huot2023omegapap} we can increase our ambitions and show that, given sufficiently nice primitive operations, our AD methods compute correct 
derivatives \emph{almost everywhere} for (almost everywhere) terminating programs in our language.
In fact, a minor adaptation of our methods yields these results.
Indeed, as long as we assume that all our (partial) primitive operations denote functions that denote functions that are piecewise analytic under analytic partition (PAP) and  are defined on a $c$-analytic subset (meaning: a countable union of analytic subsets) of $\RR^n$, then
we can simply redefine our logical relations 
\begin{align*}
    T^n_{\reals} &\defeq \{(\gamma,\gamma')\mid \gamma \text{ is PAP and } \gamma'=(x,v)\mapsto  (\gamma(x),\gamma''(x,v))\text{ for an intensional derivative $\gamma''$ of $\gamma$} \}\\
    &\{A_i\subseteq \RR^n\}_{i\in I}\text{ of }\gamma'^{-1}(\sem{\Dsyn{\ty{1}}})\text{ and there exist open neighbourhoods $U_i$ of $A_i$ with functions }\\
    &\gamma_i:U_i\to \sem{\ty{1}}, \gamma'_i:U_i\times \RR^n\to \sem{\Dsyn{\ty{1}}}\text{ such that }\gamma|_{A_i}=\gamma_i|_{A_i}\text{ and } \gamma'|_{A_i\times \RR^n}=\gamma'_i|_{A_i\times \RR^n} \text{ and}\\
    & \text{for all analytic }\delta:\RR^n\to U_i\text{ we have that } 
    (\gamma_i\circ \delta, (x,v)\mapsto (\gamma_i(\delta(x)), \gamma'_i(D\delta(x,v))))\in T^n_{\ty{1}}\Big\}.
    \end{align*}
to conclude that
\begin{itemize}
    \item any program $\var{1}:\ty{1}\vdash \trm{1}:\ty{2}$ for $\ty{1}=\reals^{m},\ty{2}=\reals^l$ in our language denotes a partial PAP function $\sem{\trm{1}$} defined on a $c$-analytic subset;
    \item our AD transformation $\Dsyn{\trm{1}}$ computes $\pairL \sem{\trm{1}},g\pairR$ for an intensional derivative $g$ of $\sem{\trm{1}}$, which coincides almost everywhere in the domain with the (standard) derivative  $D\sem{\trm{1}}$ of $\sem{\trm{1}}$.
\end{itemize}
Consequently, if our program terminates almost everywhere, the AD transformation computes the correct derivative almost everywhere.

%% file: TEX/overview.tex
\section{Overview}\label{sec:overview}
We briefly sketch the high-level plan of attack 
that we will follow in this paper.
In this work, our guiding philosophy is to consider categorical models of functional programming languages as categories with a certain kind of structure:
\begin{itemize}
\item certain chosen types $\reals$ and morphisms $\op$ for basic types of real numbers and primitive operations (such as $\cos$ and multiplication) between real numbers;
\item finite products, to represent tuples;
\item finite coproducts, to represent variants;
\item exponentials, to build types of curried and higher order functions;
\item a (partiality) monad such that the Kleisli category supports certain morphism-level fixpoint operators to represent (call-by-value) iteration (while-loops) and recursion;
\item certain object-level fixpoint operators to represent recursively defined types.
\end{itemize}
Due to its technical complexity, we isolate the discussion of the last feature (recursive types) in Section \ref{sec:recursive-types}.

A crisp formulation for the last two bullet points above is hard to find in the literature.
Therefore, we develop such a formulation in detail in Sections \ref{sec:basic-categorical-semantics}, \ref{ssec:cat-models-recursive-types}.
We further, in Sections \ref{sec:owcpos-enriched-stuff} and \ref{ssec:wcpo-enriched-cat-models-rec-types}, show that there are particularly well-behaved models of these features if we have enrichment over $\omega$-cpos ($\omega$-chain complete partial orders).
All the models we consider, except for the syntax, will fall into this well-behaved class.

We will generally identify the syntax of  a programming language, up to $\beta\eta$-equality, with the \emph{freely generated} (or initial) such category $\Syn$.
We can then understand automatic differentiation (AD) as a structure preserving functor (preserving all the structure described above) $$\Dsynsymbol:\Syn\to\Syn$$
that sends $\reals$ to a type of pairs $\reals\times\tangentreals$ (for storing both values and derivatives) and each primitive operation $\op$ to its derivative.
We discuss this in Sections \ref{sec:our-cbv-language}, \ref{ssec:rec-types-syntax}, and \ref{sub:extended-macro}.

In order to phrase the correctness of automatic differentiation,
we first need to fix the meaning of the programs in our language.
That is sometimes done using an operational semantics that describes 
how programs are evaluated in time.
Here, we work, instead, with a denotational semantics that systematically assigns
spaces (in our case, $\omega$-cpos) to types and mathematical functions (in our case, $\omega$-cocontinuous, monotone functions) to programs.
We can understand such a denotational semantics as a structure preserving functor to the category $\wCpo$ of $\omega$-cpos: 
$$
\sem{-}:\Syn\to\wCpo,
$$
which sends $\reals$ to the real numbers $\RR$ and each primitive operation $\op$ to the function $\sem{\op}$ it intends to implement.
Importantly, we are now in a position to phrase a correctness theorem for AD by relating the semantics of an AD-transformed program to the mathematical derivative of that program.
We discuss this in Sections \ref{sec:semantics-AD-transformation} and \ref{subsect:concrete-semantics-for-the-recursive-types}.

Our proof strategy for this correctness theorem is a logical relations proof, which we again phrase categorically.
Given a functor $G:\wCpo^n\to \wCpo$, we can consider the category 
$\SScone$ of $n$-ary logical relations, which has objects $(X, T)$, where $X\in\wCpo^n$ and $T$ is a (full) sub-$\omega$-cpo 
$GX$, and morphisms $(T,X)\to (T',X')$ are $\wCpo^n$-morphisms $f:X\to X'$ such that $y\in T$ implies $Gf(y)\in T'$.
Our proof proceeds by making a sensible choice of $G$ (Section \ref{ssec:basic-setting}) and giving a new categorical semantics  $\semLR{}{-}:\Syn\to \SScone$ in the category of logical relations, such that the following diagram commutes and that the commuting of this diagram immediately implies the correctness of AD (Sections \ref{ssec:ad-log-rel-data-types},
\ref{subsec:proof-basic-correctness-theorem},
\ref{sub:forward-mode-types-correctness}, and
\ref{sub:reverse-mode-types-correctness}):
\begin{figure}[!h]\begin{tikzcd}
    \Syn \arrow[d, "{(\id,\Dsynsymbol)}"'] \arrow[r, "\semLR{}{-}"] & \SScone \arrow[d, "\mathrm{forget}"] \\
    \Syn\times\Syn \arrow[r, "\sem{-}\times\sem{-}"']             & \wCpo\times\wCpo.                    
    \end{tikzcd}
\end{figure}

How do we construct such a semantics though? For that, we need to show that the category of logical relations has all the structure needed to interpret our language. That is:
\begin{itemize}
\item we show that products, coproducts, exponentials, morphism-level fixpoint operators for iteration and recursion exist in our category of logical relations (Sections \ref{sec:LogicalRelations-basic-proof} and \ref{subsect:the-definition-of-the-monad-for-the-logical-relations});
\item we show that object-level fixpoint operators for recursive types exist in our category of logical relations (Section \ref{ssec:rec-types-sscone});
\item we choose a sensible logical relation $\semLR{}{\reals}$ for $\reals$ to precisely capture correct differentiation 
and we demonstrate that each primitive operation respects the chosen logical relations (Section \ref{subsec:LR-assignment}).
\end{itemize}
The only choices that need to be made to construct this interpretation are the choice of logical relations associated with $\reals$ and with partial functions (in the form of a lifting of the partiality monad to logical relations).
All other required structure is unique thanks to a universal property.
Further, the commutativity of the diagram above follows automatically from the initiality of $\Syn$ among all categorical models.

We believe that our methods for interpreting morphism- and object-level fixpoint combinators in categories of logical relations can be a simplification compared to existing methods.
We therefore aim to present them in a reusable way.

%% file: TEX/BasicSemantics.tex
\section{Categorical models for CBV languages: $CBV$ pairs and models}\label{sec:basic-categorical-semantics}

The aim of this section is to establish a class of categorical models for call-by-value (CBV) languages with free notions of recursion and iteration.
This material will be of importance as we will later consider particular examples of such models constructed from (1) the syntax of programming languages, (2) a concrete denotational semantics for programming languages in terms of $\omega$-cpos, and (3) that concrete semantics decorated with logical relations to enable correctness proofs of AD.

Given a cartesian closed category $\catV$, we can see it as a $\catV $-enriched category w.r.t. the cartesian structure. 
Recall that a strong \textit{monad} $\monadT $ on a cartesian closed category $\catV $ 
is the same as a $\catV$-monad on $\catV$. More precisely, it is a triple 
\begin{equation}
\monadT = \left( T : \catV \to \catV , \mm : T^2\to T  , \ee  : \ID _ {\catV} \to T \right) ,
\end{equation}
where $T $ is a $\catV $-endofunctor and $\mm, \eta  $ are $\catV$-natural transformations, 
satisfying the usual associativity and identity equations, that is to say, $\mm\cdot \left( \mm T\right) = \mm \cdot \left( T\mm \right) $ and $\mm\cdot\left( \eta T \right) = \ID _ {T}= \mm\cdot\left( T\eta \right) $.\footnote{See \cite[pag.~60]{MR0280560} for the classical enriched case. For the general case of monads in $2$-categories, see  \cite[pag.~150]{MR299653} or, for instance, \cite[Section~3]{2019arXiv190201225L}.} 

Let $\monadT = \left( T, \mm , \ee \right) $ and $\monadT ' = \left( T', \mm ', \ee ' \right) $
be monads on $\catV $ and $\catV '$ respectively. Recall that an oplax morphism (or a monad op-functor) between $\mathcal{T} $ and $\mathcal{T} '$ is a pair 
\begin{equation}\label{eq:morphism-of-monads}
	\left( H : \catV\to \catV ' , \monadphi : HT\to T'H \right) ,
\end{equation} 
where $H$ is a functor and $\monadphi $ is a natural transformation, such that 
\begin{equation}\label{eq:definition-for-the-oplax-monad-morphism}
\monadphi\cdot \left( H\ee \right)   =  \left( \ee ' H \right)  \qquad \mbox{and} \qquad 
 \left( \mm 'H\right)\cdot \left( T'\monadphi \right)\cdot \left( \monadphi T \right) = \monadphi\cdot \left( H \mm \right)    .
\end{equation}

By the universal property of Kleisli categories, denoting by $J: \catV\to \catC$ and $J: \catV ' \to \catC ' $ the universal
Kleisli functors,  the oplax morphims \eqref{eq:morphism-of-monads}  correspond bijectively with pairs of functors $\left( H : \catV\to \catV ', \lift{H} : \catC\to \catC '  \right) $ such that the diagram \eqref{eq:Kleisli-morphism} commutes.

\begin{definition}[$CBV$ pair] \label{def:CBVpair}
A   $CBV$ pair is a pair $\left(\catV , \monadT \right) $
where $\catV $ is  bicartesian closed category (i.e., a cartesian closed category with finite coproducts) and $\monadT$ is a $\catV$-monad on $\catV $. We further require that $\catV $ has chosen finite products, coproducts and exponentials.

A \textit{$CBV $ pair morphism} between the $CBV$ pairs $\left(\catV , \monadT \right) $ and $\left(\catV  ', \monadT ' \right) $ is a strictly bicartesian closed functor $H:\catV\to \catV'$ that is a \emph{strict} morphism between $\monadT$ and $\monadT'$, i.e. such that $HT = T'H$ and 
 $\left( H, \ID \right)  $ defines a monad op-functor  \eqref{eq:morphism-of-monads}. This defines a category of $CBV$ pairs and $CBV$ pair morphisms, denoted herein by $\CBVcatund$. 
\end{definition}

\begin{remark}
	If $\left( \catV , \monadT \right) $ is a $CBV$ pair, since $\monadT $ is $\catV$-enriched, we get
	a $\catV$-enriched Kleisli category $\catC $. We denote by
	\begin{equation}  
		\ihomC{-}{-} = \expk{-}{-} : \catC ^\op \times \catC \to \catV  
	\end{equation}	
	the $\catV$-enriched hom functor. It should be noted that, if we denote by $\expo{X}{Y} = \ihomV{X}{Y} $ the exponential in $\catV $, we have that $\ihomC{X}{Y} = \expk{X}{Y} =  \expo{X}{TY} $ 
	which is the so called \textit{Kleisli exponential} and corresponds to the function types for our language.
\end{remark}	

Denoting by $\catC $ and $\catC ' $ the respective Kleisli categories, each morphism $$H : \left( \catV   , \monadT \right) \to \left( \catV ', \monadT ' \right) $$ of $CBV$ pairs gives rise to a commutative square
\begin{equation}\label{eq:Kleisli-morphism}
	\diag{kleisli-morphism-pair}
\end{equation}
where $J$ and $J '$ are, respectively, the universal Kleisli functors of $\monadT $ and $\monadT '$. In this case, $\lift{H}$ strictly preserves Kleisli exponentials, finite coproducts and the action of $\catV $ on $\catC$. That is to say, $\left( H, \lift{H}\right) $ strictly preserves the distributive closed Freyd-categorical structure\footnote{Although this level of generality is not needed in our work, the interested reader can find more about Freyd-categorical structures and basic aspects of the modelling of call-by-value languages in \cite{levy2003modelling}}. 


\subsection{$CBV$ models: term recursion and iteration}

In order to interpret our language defined in Section \ref{sec:our-cbv-language}, we need an additional support for term recursion and iteration.
Since we do not impose further equations for the iteration or recursion constructs in our language, the following definitions establish our class of models for term recursion and iteration.
In contrast with most other references we are aware of, we give an explicit discussion of the case of iteration, even though it is definable in terms of recursion.
The reason is that there are interesting programming languages with iteration but without recursion, of which we might want to prove properties.

\begin{definition}[Free Recursion and Iteration]
Let $\left( \catV , \monadT \right) $ be a $CBV$ pair and $ \catC $ the corresponding $\catV$-enriched Kleisli category. 
\begin{itemize}
\item A \textit{free recursion} for $\left( \catV , \monadT \right) $ is a family of morphisms
\begin{equation}\label{eq:basic-recursion}
	\fixpoint = \left( \fixpoint^{W,Y} : \ihomV{\ihomC{W}{Y} }{\ihomC{W}{Y} } \longrightarrow  \ihomC{W}{Y}\right) _{(W,Y)\in \catC\times\catC}   
\end{equation} 
in $\catV $.
\item A \textit{free iteration} for $\left( \catV , \monadT \right) $ is a family of morphisms 
\begin{equation}\label{eq:basic-iteration}
	\iterationn = \left( \iterationn^{W,Y} : \ihomC{ W }{  W\sqcup Y  } \longrightarrow \ihomC{ W } {Y} \right) _{(W,Y)\in \catC\times\catC }   
\end{equation} 
in $\catV $.
\end{itemize}
\end{definition}	

\begin{definition}[$CBV$ model]\label{def:CBVmodel}
	A \textit{$CBV$ model} is a quadruple $\left( \catV , \monadT, \fixpoint , \iterationn \right) $ in which 
	$\left( \catV , \monadT\right) $ is a $CBV$ pair, $\fixpoint $ is a free recursion, and $\iterationn $ is a free iteration for $\left( \catV , \monadT\right) $.
	
	A \textit{$CBV$ model morphism between} the $CBV$ models $\left( \catV , \monadT, \fixpoint , \iterationn \right) $ and $\left( \catV ' , \monadT ', \fixpoint ' , \iterationn ' \right) $ is a morphism $H$ between the underlying $CBV$ pairs such that $ H \left( \fixpoint ^{W,Y} \right) = \fixpoint '^{HW, HY}  $ and $ H \left( \iterationn  ^{W, Y} \right) = \iterationn '^{HW, HY}   $, for any $\left( W, Y\right)\in \catV\times\catV  $. 
	This defines a category of $CBV$ models, denoted herein by $\CBVcat $.
\end{definition} 

It should be noted that $\CBVcat $ has finite products. Given two $CBV$ models $\left( \catV _0 , \monadT _ 0, \fixpoint _0 , \iterationn _0 \right) $ and\linebreak $\left( \catV _1 , \monadT _ 1, \fixpoint _1 , \iterationn _1 \right) $, the product is given by
\begin{equation}
	\left( \catV _0\times \catV _1 , \monadT _ 0\times \monadT _ 1, \left( \fixpoint _0 , \fixpoint _1\right) , \left(\iterationn _0, \iterationn _1\right)  \right) 
\end{equation} 
where $\left( \fixpoint _0 \times \fixpoint _1\right)^{\left( W, W'\right), \left( Y, Y' \right) } =  \fixpoint _0 ^{W,Y}\times \fixpoint _1 ^{W',Y'}  $ and $\left( \iterationn _0 \times \iterationn _1\right)^{\left( W, W'\right), \left( Y, Y' \right) } = \left(  \iterationn _0 ^{W,Y}\times \iterationn _1 ^{W',Y'} \right)  $.

%% file: TEX/DomainSemanticsSimplified.tex
\section{Canonical fixed points from 2-dimensional structure}\label{sec:owcpos-enriched-stuff}
The aim of this section is to specialise to a class of particularly well-behaved $CBV$ pairs and models, as they possess canonical iteration and recursion constructs that arise from a universal property of being a least fixed point.
To phrase this universal property (and thus obtain uniqueness), we need the homsets of our models to be categories (or, posets, if we do not care to distinguish different ways of comparing morphisms), leading us to consider $\catV$ that are $\catCat$- or $\catPos$-enriched. 
To obtain existence of these fixed points, it is sufficient to have colimits of countable chains, leading us to specialise to $\catV$ that are enriched over $\omega$-cocomplete categories or posets.

We denote by $ \wCpo $ the usual category of \wcpos. The objects of  $ \wCpo $ are the
partially ordered sets with colimits of $\ww$-chains, while the morphisms are functors preserving these colimits.
 An \wcpo{} is called \textit{pointed} if it has a least element, denoted herein by $\leastelement $. We say that
 $f\in\wCpo\left( W, Y\right) $ is a pointed $\wCpo$-morphism if $W $ is pointed and $f$ preserves the least element.

It is well known that $ \wCpo $ is bicartesian closed\footnote{In fact, it is locally presentable, hence complete and cocomplete, and its internal hom $X\Rightarrow Y$ is given by using the order $f\leq_{X\Rightarrow Y} g$ defined as $\forall x\in X.f(x)\leq_Y g(x)$ on the homset $\wCpo(X,Y)$.
Its products $\bigsqcap_{i\in I}X_i$ carry the lexicographic order and its coproducts $\bigsqcup_{i\in I}X_i$ have the disjoint union of the orders of all $X_i$, making elements in different components incomparable.
See, for example, \cite{vakar2019domain} for more details.}.  
We consider $ \wCpo $-enriched categories w.r.t. the cartesian structure.
Henceforth, if $\catV $ is an $\wCpo $-enriched category and $W, Y$ are objects of $\catV $, we denote by
$\catV\left( W, Y\right) $ the $ \wCpo $-enriched hom, that is to say, the \wcpo{} of morphisms between $W$ and $Y$. 

An $\wCpo$-category $\catV $ has finite $\wCpo$-products if it has a terminal object and we have a natural isomorphism of $\omega$-cpos
$$
\pairL -,-\pairR:\catV(Z,W)\times \catV(Z, Y)\cong \catV(Z, W\times Y),
$$
or, equivalently, if it has finite products and tupling is an $\wCpo$-morphism.
Dually, an $\wCpo$-category $\catV $ has finite $\wCpo$-coproducts if it has an initial object and we have a natural isomorphism of $\omega$-cpos
$$
\cpairL -,-\cpairR:\catV(W,Z)\times \catV(Y, Z)\cong \catV(W\sqcup Y, Z),
$$
or, equivalently, if it has finite coproducts and cotupling is an $\wCpo$-morphism.
We say that an $\wCpo$-functor $F:\catV\to \catV'$ has an $\wCpo$-right adjoint $U:\catV'\to \catV$ if we have a natural isomorphism of $\omega$-cpos
$$\catV'(FZ, W') \cong \catV(Z, UW'),$$
or, equivalently, if it has a right adjoint functor $U:\catV'\to \catV$ such that the homset bijection $\catV'(FZ, W') \to \catV(Z, UW')$ is an $\wCpo$-morphism. 
An $\wCpo$-category $\catV $ is \textit{$\wCpo$-cartesian closed} if $\catV $ has finite $\wCpo $-products and, moreover, for each object $Z\in \catV$, the $\wCpo$-functor $\left(Z\times -\right) : \catV\to\catV $ has a right $\wCpo$-adjoint $\ihomV{Z}{-}$, called, herein, the $\wCpo$-exponential of $Z$. An $\wCpo$-functor $H : \catV\to \catV '$ 
is \textit{strictly $\wCpo$-cartesian closed} if it strictly preserves the $\wCpo$-products and the induced comparison  $H\circ\ihomV{-}{-}\to \ihom{\catV '}{H(-)}{H(-)}$  
is the identity.

Let $\catV $ be $\wCpo$-cartesian closed. For any $Z\in \catV$, since the hom-functor
$\ehomV{Z}{-} : \catV\to\wCpo$  is cartesian, it induces the \textit{change of enriching base  $2$-functors}
\begin{equation}
\chang{\ehomV{Z}{-}}: \ecat{\catV}\to\ecat{\wCpo} 
\end{equation}
between the $2$-categories of enriched categories w.r.t. the cartesian structures.
Therefore, taking $Z = \terminall $ (the terminal object of $\catV$), we get that every $\catV $-category  ($\catV$-functor/$\catV$-monad) has a \textit{suitable underlying} $\wCPO $-category  ($\wCpo$-functor/$\wCpo$-monad), given by its image by $\chang{\wCpo} := \chang{\ehomV{\terminall}{-}}$. 

\begin{definition}[$CBV$ $\wCpo$-pair]\label{def:CBV-mode-wcpo-pair}
	A \textit{$CBV$ $\wCpo$-pair} is a $CBV$ pair $\left( \catV , \monadT \right) $ in which $\catV$ is an $ \wCpo $-bicartesian closed category, such that $\ehomV{W}{TY} $ is a pointed \wcpo{} for any $(W,Y)\in \catV\times\catV$.  

	A  \textit{$CBV$ $\wCpo$-pair morphism} between $\left(\catV , \monadT \right) $ and $\left(\catV  ', \monadT ' \right) $ is an $\wCpo$-functor 
	$H : \catV\to\catV ' $ whose underlying functor yields a morphism between the $CBV$ pairs, and such that 
	$ H : \ehomV{W}{TY}\to \ehomV {HW}{H TY}$
	is a pointed $\wCpo$-morphism for any $\left( W, Y\right)\in \obb{\catV}\times\obb{\catV} $. This defines a category of $CBV$ $\wCpo$-pairs, denoted herein by $\wCBVcat$.
\end{definition} 

There is, then, an obvious forgetful functor $\CBVUp : \wCBVcat\to \CBVcatund $.

\subsection{Fixpoints: term recursion and iteration}

 Recall that, if $A$ is a pointed $\omega$-cpo and $q : A\to A $ is an endomorphism in $\wCpo $, then $q$ has a \textit{least fixed point} given by the 
 colimit of the chain 
\begin{equation}\label{basic-least-element-chain} 
\diag{basic-wcpo-chain}
\end{equation}
by Kleene's Fixpoint Theorem.  Given such an endomorphism, we denote by $\lfpoint{q} $ its least fixed point.

Henceforth, let $\left( \catV , \monadT \right) $ be a $CBV$ $\wCpo$-pair, and $J: \catV \to \catC $ the corresponding $\catV$-enriched universal Kleisli functor. We denote by
$-\kleislitimes - : \catV\times\catC \to \catC $
the $\catV$-tensor product in $\catC$, also called the $\catV $-copower, which, in this case, correspond to the usual action of $\catV$ on $\catC $.

By hypothesis, for any $W,Y,Z\in \catV $, the \wcpo{} $\catC \left( Z\kleislitimes W,  Y \right)\cong\catV\left( Z, \catC\left[ W,Y \right] \right) $
 is pointed. Let us write $\Lambda_Z^{W,Y}$ for the isomorphism from left to right. Then, we can define
\begin{eqnarray}
\pwfixpoint^{W,Y}_Z :&\ehomV{ Z\times\ihomC{W}{Y} }{\ihomC{W}{Y}} & \to  \ehomV{ Z} { \ihomC{W}{Y}} \label{def:rec-wcpo-pair} \\
                   & f & \mapsto \lfpoint { h\mapsto  f\circ \left( Z\times h\right)\circ  \diagv{Z} }   \nonumber\\
	\pwiterate^{W,Y}_Z :&\ehomC{Z\kleislitimes W}{ W\sqcup Y } & \to \ehomC{ Z\kleislitimes W} {Y} \label{def:it-wcpo-pair}\\
	& f & \mapsto \lfpoint { 
	\begin{array}{l}
		h\mapsto \cpairL h, J\left( \projY\right)  \cpairR\circ d_{Z,W,Y} \circ\left( Z\kleislitimes f\right)\\
		\qquad \quad\circ a_{Z,Z,W} \circ \left( \diagv{Z}\kleislitimes \ID _ W \right)
	\end{array}	
	}     \nonumber
\end{eqnarray}
where $\diagv{Z} = \left(\ID_ Z , \ID _Z \right) : Z\to Z\times Z $ is the diagonal morphism, $d_{Z,W,Y}:Z\otimes (W\sqcup Y)\to (Z\otimes W)\sqcup (Z\otimes Y)$ is the distributor, and $a_{Z,W,Y}:(Z\times W)\otimes Y\to Z\otimes (W\otimes Y)$ is the associator.
 Since the morphisms above are  $\wCpo $-natural in $Z\in\catV $, they give rise to the families of morphisms 
\begin{eqnarray}
\wfixpoint = \left( \wfixpoint^{W,Y} \right) _{(W,Y)\in \catC\times\catC} &\defeqq  & \left(  \pwfixpoint^{W,Y} _{\ihomV{\ihomC{W}{Y} }{\ihomC{W}{Y}}}\left( \eva{\ihomC{W}{Y}}{\ihomC{W}{Y}} \right) \right) _{(W,Y)\in \catC\times\catC}\label{def:free-recursion-wCPO}\\
	\witerate = \left( \witerate^{W,Y}  \right) _{(W,Y)\in \catC\times\catC } &\defeqq & \Lambda_{\ihomV{ W }{T\left( W\sqcup Y\right) }}^{W,Y}\left(  \pwiterate^{W,Y} _{\ihomV{ W }{T\left( W\sqcup Y\right) }} \left( \eva{W}{T\left(W\sqcup Y\right)}  \right)\right) _{(W,Y)\in \catC\times\catC}\label{def:free-iteration-wCPO}
\end{eqnarray} 	
by the Yoneda Lemma, where $\eva{A}{B} :  \ihomV{A}{B} \times A\to B $ is the evaluation morphism given by the cartesian closed structure.

\begin{lemma}[Underlying $CBV$ model]\label{UnderlyingCBVmodel}
	There is a forgetful functor $\CBVU : \wCBVcat\to \CBVcat $ defined by $\CBVU\left( \catV , \monadT \right) = \left( \catV , \monadT , \wfixpoint , \witerate \right)  $, taking every morphism $H$ to its underlying morphism of $CBV$ models.
\end{lemma}	
\begin{proof}
	Since  $H$ is a $\wCpo $-functor and, for any $\left( W, Y\right)\in \obb{\catV}\times\obb{\catV} $,
	$$ H : \ehomV{W}{TY}\to \ehom{\catV'}{HW}{T'HY} 
	$$
	is a pointed $\wCpo$-morphism,  we get that, indeed, $H$ respects the free iteration and free recursion as defined in \eqref{def:free-recursion-wCPO} and \eqref{def:free-iteration-wCPO}. 
\end{proof}

It should be noted that, given $CBV$ $\wCpo$-pairs $\left( \catV _0 , \monadT _0 \right) $ and $\left( \catV _1 , \monadT _1\right)  $, 
\begin{equation}
	\left( \catV _0 , \monadT _0\right)\times\left( \catV _1 , \monadT _1 \right) =  \left( \catV _0\times \catV _1 , \monadT _0\times \monadT _1 \right) 
\end{equation}
is the product in $\wCBVcat$. Moreover, $\CBVU $ preserves finite products.

%% file: TEX/coarse-language.tex
\section{Automatic Differentiation for term recursion and iteration}\label{sec:our-cbv-language}
For our purpose, we could define our macro
in terms of total derivatives. However, we choose to present it  in terms of partial derivatives, in order to keep our treatment as close as possible to 
the starting point of the efficient implementation of the reverse mode given in \cite{smeding2022}. 

Following this choice of presentation,  it is particularly convenient to establish our
\textit{AD} macro $\Dsynsymbol $ as a \textit{program transformation} between a \textit{source language} and  a \textit{target language} (see Section \ref{sec:fwdmode-AD-transformation}). 
The source language contains the programs that we differentiate, while we use the target language to represent those derivatives.

\subsection{Source language as a standard call-by-value language with iteration and recursion}\label{subsect:source-language}
We consider a standard (coarse-grain) call-by-value language over a ground type $\reals$, certain real constants $\cnst{c}\in\Op_0$,
certain primitive operations $\op\in\Op_n$ for each nonzero natural number $n\in\NN ^\ast $, and $\tSign$.
We denote $\displaystyle\Op := \bigcup _{n\in\NN } \Op _ n $.

As it is clear from the semantics defined in Section \ref{subsect:semantics-source-language},
$\reals $ intends to implement the real numbers. Moreover,
 for each $n\in\NN$, the operations in $\Op _n $ intend to implement partially defined functions
 $\RR ^n \rightharpoonup \RR $. Finally, $\tSign $ intends to implement the partially defined function $\RR\rightharpoonup \Unit\sqcup \Unit$ defined on $\RR ^- \cup \RR ^+ $
 which takes $\RR ^-$ to the left component and $\RR ^+ $ to the right component.

 Although it is straightforward to consider more general settings, we also add the assumption that the primitive operations implement differentiable functions (see Assumption \ref{ass:differentiable-functions}).


We treat this operations in a schematic way as this reflects the reality of 
practical Automatic Differentiation libraries, which are constantly being expanded with new primitive operations.

The types $\ty{1},\ty{2},\ty{3}$, values $\val{1},\val{2},\val{3}$, and computations $\trm{1},\trm{2},\trm{3}$ of our language are as follows.\\
\input{TEX/types-values-computations}\\
We use sugar  
$\ifelse{\trm{3}}{\trm{1}}{\trm{2}}\defeq 
\vMatch{\tSign\trm{3}}{{
		\_\To{\trm{1}}
		\vor \_\To{\trm{3}}
}}$,
$\tFst\trm{1}\defeq \pMatch{\trm{1}}{\var{1}}{\_}{\var{1}}
$,
$\tSnd\trm{1}\defeq \pMatch{\trm{1}}{\_}{\var{1}}{\var{1}}
$ and 
$\letrec{f}{\var{1}}{\trm{1}}{\trm{2}}  
\defeq 
\letin{f}{\rec{f}{\fun{\var{1}}{\trm{1}}}}{\trm{2}}$.
In fact, we can consider iteration as syntactic sugar as well:\\
$
\tItFrom{\trm{1}}{\var{1}}{\trm{2}}\defeq(\rec{\var{3}}{\fun{\var{1}}{
		\vMatch{\trm{1}}{\tInl\var{1}'\To \var{3}\,\var{1}'\mid \tInr\var{1}''\To \var{1}''}
}})\,\trm{2}
$.

The computations are typed according to the rules of Figure~\ref{fig:types1} and Figure~\ref{fig:typestermrecursion-iteration}, \textit{where $\RRsyntax\subset\RR$ is a fixed set of real numbers containing $0$.}
For now, the reader may ignore the kinding contexts $\Delta$.
They will serve to support our treatment of ML-style polymorphism later.
\begin{figure}[!ht]
	\fbox{\parbox{0.98\linewidth}{\begin{minipage}{\linewidth}\noindent\input{TEX/type-system}\end{minipage}}}
	\caption{Typing rules for a basic source language with real conditionals, where $\RRsyntax\subset\RR$ is a fixed set of real numbers containing $0$.
		\label{fig:types1}}
\end{figure}
\begin{figure}[!ht]
	\fbox{\parbox{0.98\linewidth}{\begin{minipage}{\linewidth}\noindent\input{TEX/type-system-termrecursion-iteration}\end{minipage}}}
	\caption{Typing rules for term recursion and iteration.
		\label{fig:typestermrecursion-iteration}}
\end{figure}

We consider the standard CBV $\beta\eta$-equational theory of \cite{moggi1988computational}
for our language, which we list in Figure \ref{fig:beta-eta1}.
We could impose further equations for the iteration construct as is done in  \cite{bloom1993iteration,goncharov2015unguarded}
as well as for the basic operations\, $\op$\, and the sign function $\tSign$.
However, such equations are unnecessary for our development.
\begin{figure}[!ht]
	\fbox{\parbox{0.98\linewidth}{\scalebox{0.92}{\begin{minipage}{\linewidth}\noindent\input{TEX/beta-eta1}\end{minipage}}}
	}\caption{\label{fig:beta-eta1} Basic $\beta\eta$-equational theory for  our language.
		We write $\beta\eta$-equality as $\equiv$ to distinguish it from equality in let-bindings.
		We write $\freeeq{\var{1}_1,\ldots,\var{1}_n}$ to indicate that the variables are fresh in the left hand side.
		In the top right rule, $\var{1}$ may not be free in $\trm{3}$.
		Equations hold on pairs of computations of the same type.
	}
\end{figure}

\subsection{Target language}\label{sub:target-language-syntax}

We define our \textit{target language} by
extending the source language adding the following syntax, with the typing rules of Figure \ref{fig:types-reverse}. 
\\
\input{TEX/types-values-computations-reverse}\\
\begin{figure}[!ht]
	\fbox{\parbox{0.98\linewidth}{\begin{minipage}{\linewidth}\noindent\input{TEX/type-system-reverse}\end{minipage}}}
	\caption{Extra typing rules for the target language with iteration and recursion, where we denote $\NN ^\ast := \NN - \left\{ 0 \right\}$, $\reals ^1 := \reals $ and $\reals ^{i+1} = \reals ^i \times \reals $.
		\label{fig:types-reverse}}
\end{figure}

The operational semantics  of the target language depends on the intended behavior for the $AD$ macro $\Dsynsymbol $ defined in Section \ref{sec:fwdmode-AD-transformation}. In our context,
we want $\tangentreals $ to implement a vector space ((co)tangent vectors), with the respective operations and the usual laws between the operations such as distributivity of the scalar multiplication over the vector addition (which is particularly useful for efficient implementations~\cite{smeding2022}).

The terms $\tangentprojection{i}{\trm{1}}$ are irrelevant for the definition and correctness of the macro $\Dsynsymbol $, but it is particularly useful to illustrate the expected types in Section \ref{sub:forward-mode-types-correctness} and \ref{sub:reverse-mode-types-correctness}. Although this perspective is unimportant for our correctness statement, the reader might want to view $\tangentreals $ as a computation type encompassing \emph{computational effects} for the vector space operations $\overline{e}_i$, $(*)$, and $(+)$ with \textit{handlers} given by the terms $\tangentprojection{i}{\trm{1}}$.

We are particularly interested in the case that $\left( \tangentreals , + , \ast , \cnzero\right) $ implements the vector space $\left( \RR ^k, + , \ast , 0\right) $, for some $k\in\NN\cup\left\{ \infty \right\}$,\footnote{$\RR ^ \infty $ is the vector space freely generated by the infinite set $\left\{ \canonicalbasise{}{i}: i\in\NN  ^\ast \right\} $. In other words, it is the infinity coproduct of $\RR ^i $ ($i\in\NN ^\ast$). In order to implement it, one can use lists/arrays and pattern matching for the vector addition.} where $\cncanoni{i}$ implements the $i$-th element  $\canonicalbasise{k}{i}\in\RR ^k $  of the canonical basis if $k=\infty $ or if $i\leq k $, and $0\in\RR^k$ otherwise. In this case, $\tangentprojection{i}{\trm{1}}$ is supposed to
implement
\begin{equation}\label{eq:respective-coprojections} 
\semanticshandler{k}{i}: \RR ^k \to \RR ^i,
\end{equation} 
which denotes the canonical projection if $i\leq k $ and the coprojection otherwise.

For short, we say that $\tangentreals $ implements the vector space $\RR ^k $ to refer to the case above. It corresponds to the $k$-semantics for the target language defined in Section \ref{subsect:k-semantics-target-language}.


\subsection{The $CBV$ models $\left( \SynV, \SynT,\Synfix , \Synit   \right) $ and $\left( \SynVt , \SynTt ,\Synfixt , \Synitt   \right) $} \label{subsec:Syntax-as-a-categoricalstructure}

As discussed in Appendix \ref{appx:fine-grain-cbv}, we can translate our coarse-grain  languages to  fine-grain call-by-value languages. 
The 
fine-grain languages corresponding to the source and target languages correspond to the $CBV$ models 
\begin{equation} \label{eq:CBV-models-syntax} 
	\left( \SynV, \SynT,\Synfix , \Synit   \right)\qquad\mbox{and}\qquad 
	\left( \SynVt , \SynTt ,\Synfixt  , \Synitt   \right)
\end{equation} 
with the following universal properties.

\begin{proposition}[Universal Property of $CBV$ models \eqref{eq:CBV-models-syntax}]\label{prop:section-universal-property-syntax}
	Let $\left( \catV , \monadT, \fixpoint , \iterationn \right) $  be a $CBV$ model. Assume that Figure~\ref{fig:assignment-functor-universalproperty-syntax} and Figure~\ref{fig:assignment-functor-universalproperty-syntax-target} are given consistent assignments.
\begin{enumerate}	
		\item There is a unique $CBV$ model morphism $H: \left( \SynV, \SynT,\Synfix , \Synit   \right)\to \left( \catV , \monadT, \fixpoint , \iterationn \right) $ respecting the assignment of Figure~\ref{fig:assignment-functor-universalproperty-syntax}.	 	
		\item There is a unique $CBV$ model morphism $\extendedH{H}: \left( \SynVt , \SynTt ,\Synfixt  , \Synitt    \right)\to \left( \catV , \monadT, \fixpoint , \iterationn \right) $ that extends $H$ and respects the assignment of Figure~\ref{fig:assignment-functor-universalproperty-syntax-target}.
	\end{enumerate}
\end{proposition}
\begin{figure}[!ht]
	\fbox{\parbox{0.98\linewidth}{\begin{minipage}{\linewidth}\noindent\input{TEX/assignment-functor-universalproperty-syntax}\end{minipage}}}
	\caption{Assignment that gives the universal property of the source language.
		\label{fig:assignment-functor-universalproperty-syntax}}
\end{figure} 
\begin{figure}[!ht]
	\fbox{\parbox{0.98\linewidth}{\begin{minipage}{\linewidth}\noindent\input{TEX/assignment-universalproperty-syntax-target}\end{minipage}}}
	\caption{Assignment that gives the universal property of the target language.
		\label{fig:assignment-functor-universalproperty-syntax-target}}
\end{figure}

%% file: TEX/types-values-computations.tex
\begin{syntax}
    \ty{1}, \ty{2}, \ty{3} & \gdefinedby & & \syncat{types}                          \\
    &\gor& \reals                      & \synname{numbers}\\
    &\gor & \Init  \gor \ty{1} \t+ \ty{2}  & \synname{sums}\\
    &&&\\
    \val{1}, \val{2}, \val{3} & \gdefinedby & & \syncat{values}                          \\
    &\gor& \var{1},\var{2},\var{3}                      & \synname{variables}\\
    &\gor& \cnst{c}                   & \synname{constants}\\
    &\gor& \tInl{\val{1}} \gor   \tInr{\val{1}} & \synname{sum inclusions}\\
    &&&\\
    \trm{1}, \trm{2}, \trm{3} & \gdefinedby & & \syncat{computations}                          \\
    &\gor& \var{1},\var{2},\var{3}                      & \synname{variables}\\
    &\gor & \letin{\trm{1}}{\var{1}}{\trm{2}} & \synname{sequencing}\\
    &\gor& \cnst{c}                   & \synname{constant}\\
    &\gor & \op(\trm{1}_1,\ldots,\trm{1}_n) & \synname{operation}\\
    &\gor & \nvMatch{\trm{1}} & \synname{sum match}\\
	&\gor& \tInl{\trm{1}} \gor   \tInr{\trm{1}} & \synname{sum inclusions}\\
    &\gor & \vMatch{\trm{3}}{\begin{array}{l}\;\;\tInl\var{1}\To\trm{1}\\
    \gor \tInr\var{2}\To \trm{2}\end{array}}\hspace{-15pt}\; & \synname{sum match}\\
  \end{syntax}%
~
 \begin{syntax}
	&\gor\quad\, & \Unit  \gor  \ty{1}_1 \t* \ty{1}_2 & \synname{products}\\
  &\gor& \ty{1} \To \ty{2}              & \synname{function}      \\
	& & &\\
	&&&\\
	&\gor\quad\, & \tUnit \ \gor  \tPair{\val{1}}{\val{2}} & \synname{tuples}\\
  &\gor& \fun{\var{1}}{\trm{1}}             & \synname{abstractions}      \\
  &\gor &\rec{\var{1}}{\trm{1}} & \synname{term recursion}\\
  &&&\\\\
	&\gor\quad\, & \tUnit \ \gor  \tPair{\trm{1}}{\trm{2}} & \synname{tuples}\\
	&\gor\quad\, & \pMatch{\trm{2}}{\var{1}}{\var{2}}{\trm{1}}\hspace{-10pt}\; & \synname{product match}\\
	&\gor& \fun{\var{1}}{\trm{1}}             & \synname{abstractions}      \\
	&\gor& \trm{1}\ \trm{2}             & \synname{function app.}      \\
	&\gor&\tItFrom{\trm{1}}{\var{1}}{\trm{2}}\hspace{-10pt}\; & \synname{iteration}\\
	& \gor & \rec{\var{1}}\trm{1} & \synname{term recursion}\\
	&\gor&\tSign\trm{1} & \synname{sign function}
\end{syntax}

%% file: TEX/type-system.tex
\[
  \begin{array}{c}
\inferrule{
  ((\var{1} : \ty{1}) \in \ctx)
}{
  \DGinf {\var{1}}{\ty{1}}    
}
\quad
\inferrule{
  \DGinf{\trm{1}}{\ty{2}}\quad
  \DGinf[,\var{1}:\ty{2}]{\trm{2}}{\ty{1}}
}{
  \DGinf{\letin{\var{1}}{\trm{1}}{\trm{2}}}{\ty{1}}
}
\quad 
\inferrule{ (c\in\RRsyntax)}{\DGinf {\cnst{c}}\reals}\\ \\
   \inferrule{
     \set{\DGinf {\trm{1}_i} \reals}_{i=1}^n\quad (\op\in\Op_n)
   }{
     \DGinf {\op(\trm{1}_1,\ldots,\trm{1}_n)}\reals
   }
\quad 
\inferrule{\DGinf{\trm{1}}\Init}{\DGinf {\nvMatch{\trm{1}}}{\ty{1}}}\quad 
   \quad
   \inferrule{\DGinf {\trm{1}} {\ty{1}}}
   {\DGinf{\tInl\trm{1}}\ty{1}\t+\ty{2}}
   \\\\
   \inferrule{\DGinf {\trm{1}} {\ty{1}}}
   {\DGinf{\tInr\trm{1}}\ty{1}\t+\ty{2}}
   \quad
   \inferrule{\DGinf{\trm{3}}{\ty{2}\t+\ty{3}}\quad
   \DGinf[,\var{1}:\ty{2}]{\trm{1}}{\ty{1}}\quad 
   \DGinf[,{\var{2}}:\ty{3}]{\trm{2}}{\ty{1}}
   }{\DGinf{\bvMatch{\trm{3}}{\var{1}}{\trm{1}}{\var{2}}{\trm{2}}}{\ty{1}}}
   \quad 
   \inferrule{~}{\DGinf \tUnit \Unit}\\ \\ 
   \inferrule{\DGinf {\trm{1}} {\ty{1}}\quad \DGinf {\trm{2}} {\ty{2}}}
   {\DGinf {\tPair{\trm{1}}{\trm{2}}} {\ty{1}\t*\ty{2}}}
   \quad 
   \inferrule{
   \DGinf {\trm{3}}{\ty{2}\t*\ty{3}}\quad 
   \DGinf[{,\var{1}:\ty{2},\var{2}:\ty{3}}] {\trm{1}}{\ty{1}}}{\DGinf{\pMatch{\trm{3}}{\var{1}}{\var{2}}{\trm{1}}}{\ty{1}}}
   \quad 
   \inferrule{\DGinf[,\var{1}:{\ty{2}}]{\trm{1}}{\ty{1}}}{\DGinf{\fun{\var{1}}{\trm{1}}}{\ty{2}\To\ty{1}}} \\ \\ 
   \inferrule{\DGinf {\trm{1}} {\ty{2}\To\ty{1}}
   \quad\DGinf {\trm{2}}{\ty{2}} 
   }{\DGinf {\trm{1}\,\trm{2}}{\ty{1}}}
   \quad 
   \inferrule{\DGinf {\trm{3}} \reals}{\DGinf {\tSign\trm{3}} {\Unit\t+ \Unit}
   }
\end{array}
\]

%% file: TEX/type-system-termrecursion-iteration.tex
\[
  \begin{array}{c}
  \inferrule{\DGinf[,\var{1}:{\ty{2}}] {\trm{1}}{\ty{2}\t+\ty{1}}
   \quad \DGinf {\trm{3}} {\ty{2}} }{\DGinf {\tItFrom{\trm{1}}{\var{1}}{\trm{3}}}{\ty{1}}}
   \quad
   \inferrule{ \DGinf[,\var{1}:\ty{1}]{\trm{1}}{\ty{1}}}
   {\DGinf{\rec{\var{1}}{\trm{1}}}{\ty{1}}}(\ty{1}=\ty{2}\To\ty{3})
\end{array}
\]

%% file: TEX/beta-eta1.tex
\[
\begin{array}{l@{\ }l@{\ }l@{\ }l@{\ }l@{\ }l}
    \letin{\var{2}}{(\letin{\var{1}}{\trm{1}}{\trm{2}})}{\trm{3}}
    &\beeq& \letin{\var{1}}{\trm{1}}{(\letin{\var{2}}{\trm{2}}{\trm{3}})}\hspace{-18pt}\\
    \letin{\var{1}}{\val{1}}{\trm{1}}&\beeq&\subst{\trm{1}}{\sfor{\var{1}}{\val{1}}}
    &
    \\
    \bvMatch{\tInl\val{1}}{\var{1}}{\trm{1}}{\var{2}}{\trm{2}}&\beeq& \subst{\trm{1}}{\sfor{\var{1}}{\val{1}}} &
    \subst{\trm{1}}{\sfor{\var{3}}{\val{1}}}&\freeeq{\var{1},\var{2}}&\vMatch{\val{1}}
    { \begin{array}{l}\;\,\,\tInl\var{1}\To\subst{\trm{1}}{\sfor{\var{3}}{\tInl\var{1}}}\\ \gor \tInr\var{2}\To\subst{\trm{1}}{\sfor{\var{3}}{\tInr\var{2}}}\end{array}
    }
    \\
    \bvMatch{\tInr\val{1}}{\var{1}}{\trm{1}}{\var{2}}{\trm{2}}&\beeq& \subst{\trm{2}}{\sfor{\var{2}}{\val{1}}}&\\
    \pMatch{\tPair{\val{1}}{\val{2}}}{\var{1}}{\var{2}}{\trm{1}}&\beeq& \subst{\trm{1}}{\sfor{\var{1}}{\val{1}},\sfor{\var{2}}{\val{2}}}\quad &
    \subst{\trm{1}}{\sfor{\var{3}}{\val{1}}}&\freeeq{\var{1},\var{2}}&\pMatch{\val{1}}{\var{1}}{\var{2}}{\subst{\trm{1}}{\sfor{\var{3}}{\tPair{\var{1}}{\var{2}}}}} \\
    (\fun{\var{1}}{\trm{1}})\ \val{1} &\beeq&  \subst{\trm{1}}{\sfor{\var{1}}{\val{1}}} &
    \val{1} &\freeeq{\var{1}\phantom{,\var{2}}}& \fun{\var{1}}{\val{1}\, \var{1}}
\end{array}    
\]

%% file: TEX/types-values-computations-reverse.tex
\begin{syntax}
    \ty{1}, \ty{2}, \ty{3} & \gdefinedby & & \syncat{types}                          \\
    &\gor& \ldots                      & \synname{as before}\\
    \\
    \val{1},\val{2},\val{3} & \gdefinedby && \syncat{values}\\
       &\gor &\cncanoni{i}&\synname{$i$-th canonical element} \\
    &\gor&\ldots&\synname{as before}\\
    \\
    \\
    \\
    \trm{1},\trm{2},\trm{3} & \gdefinedby && \syncat{computations}\\
        &\gor&\ldots&\synname{as before}\\
     &\gor&\cncanoni{i}&\synname{canonical element}
  \end{syntax}%
~
 \begin{syntax}
  &\gor\quad\, & \tangentreals  & \synname{(co)tangent}\\
  \\
  \\
  &\gor  & \cnzero & \synname{zero}\\
    &\gor & \trm{1} + \trm{2}&\synname{addition of vectors}\\
  &\gor & \trm{1}\ast\trm{2} & \synname{scalar multiplication} \\ 
  & \gor & \tangentprojection{i}{\trm{1}} & \synname{proj. handler}
  \\
   \\
  &\gor & \cnzero & \synname{zero }\\
  &\gor & \trm{1} + \trm{2}&\synname{addition of vectors}\\
  &\gor & \trm{1}\ast\trm{2} & \synname{scalar multiplication} \\ 
    & \gor & \tangentprojection{i}{\trm{1}} & \synname{proj. handler}
\end{syntax}

%% file: TEX/type-system-reverse.tex
\[
  \begin{array}{c} 
\inferrule{\left( i\in\NN ^\ast\right)   }
{\DGinf{\cncanoni{i} }{\tangentreals }}
\qquad
\inferrule{~}
{\DGinf{\cnzero}{\tangentreals }}
\qquad
\inferrule{\DGinf{\trm{1}}{\tangentreals }\quad \DGinf{\trm{2}}{\tangentreals }}
{\DGinf{\trm{1} + \trm{2}}{\tangentreals }}
\\
\\
\inferrule{\DGinf{\trm{1}}{\tangentreals }\quad 
\DGinf{\trm{2}}{\reals}}
{\DGinf{\trm{1}\ast\trm{2} }{\tangentreals }}
\qquad
\inferrule{\left( i\in\NN ^\ast \right)\quad \DGinf{\trm{1}}{\tangentreals }}
{\DGinf{\tangentprojection{i}{\trm{1}} }{ \targetreals ^i }}
\end{array}
\]

%% file: TEX/assignment-functor-universalproperty-syntax.tex
\[
\begin{array}{c}
\mbox{For each primitive operation  $\op\in\Op_n$ ($n\in\NN$) and each constant  $c\in \RRsyntax$:}	\\
H(\reals) \in  \obb{\catV};\quad H (\tSign)  \in  \ehomC{H(\reals)} { \terminall\sqcup \terminall }  = \ehomV{H(\reals)} {T\left( \terminall\sqcup \terminall\right)}; \\
H (\cnst c) \in   \ehomV{\terminall}{H(\reals)}; \quad  
H (\op)\in  \ehomC{H(\reals)^n }{ H(\reals) } =  \ehomV{H(\reals)^n }{ TH(\reals) }.
\end{array}
\]

%% file: TEX/assignment-universalproperty-syntax-target.tex
\small
\[
\begin{array}{c}
	\extendedH{H}( \tangentreals ) \in   \obb{\catV}; \, \extendedH{H} (\cnzero ) \in   \ehomV{\terminall}{\extendedH{H} (\tangentreals )}; \,\extendedH{H}\left( \tangentprojection{i}{}\right) \in \ehomV{\extendedH{H}\left( \tangentreals\right)  }{\extendedH{H}\left( \targetreals\right)  ^i }    \mbox{ (for each  $i\in\NN ^\ast$)} ;
\\ 
	\extendedH{H} (+)\in   \ehomV{\extendedH{H}(\tangentreals )^2 }{ T\extendedH{H} (\tangentreals ) }; \quad 
\extendedH{H} (\ast )\in   \ehomV{\extendedH{H}(\tangentreals )\times \extendedH{H}  \left(\reals \right) }{ T\extendedH{H} \left( \tangentreals \right) }.
\end{array}
\] \normalsize

%% file: TEX/dualnumbers_macro.tex
\subsection{Dual numbers AD transformation for term recursion and iteration}\label{sec:fwdmode-AD-transformation}
Let us fix, for all $n\in\NN$, for all $\op\in\Op_n$, for all $1\leq i \leq n$, 
computations $\var{1}_1:\reals,\ldots,\var{1}_n:\reals\vdash \partial_i\op(\var{1}_1,\ldots,\var{1}_n):\reals$,
which represent the partial derivatives of $\op$.
Using these terms for representing partial derivatives,
we define, in Figure \ref{fig:ad1}, a structure preserving macro $\Dsynsymbol$ on the types and computations of our language for performing AD.

\begin{figure}[!ht]
	\fbox{\parbox{0.98\linewidth}{\begin{minipage}{\linewidth}\noindent
				\input{TEX/d-types1}
				\hrulefill
				\input{TEX/d-terms1}
	\end{minipage}}}
	\caption{AD macro  $\Dsyn{-}$ defined on types and computations.
		All newly introduced variables are chosen to be fresh.
		We provide a more efficient way of differentiating $\tSign$ in Appendix \ref{sec:efficient-sign-derivative}. \label{fig:ad1}}
\end{figure}

We extend $\Dsynsymbol$ to contexts: $\Dsyn{\{\var{1}_1{:}\ty{1}_1,{.}{.}{.},\var{1}_n{:}\ty{1}_n\}}\defeq
\{\var{1}_1{:}\Dsyn{\ty{1}_1},{.}{.}{.},\var{1}_n{:}\Dsyn{\ty{1}_n}\}$.
This turns $\Dsynsymbol$ into a well-typed, functorial macro in the following sense.
\begin{lemma}[Functorial macro]\label{lem:functorial-macro}
	Our macro respects typing, substitution, and $\beta\eta$-equality:
	\begin{itemize}
		\item If $\Ginf{\trm{1}}\ty{1}$, then $\Dsyn{\Gamma}\vdash\Dsyn{\trm{1}}:\Dsyn{\ty{1}}$.
		\item $\Dsyn{\letin{\var{1}}{\trm{1}}{\trm{2}}}=\letin{\var{1}}{\Dsyn{\trm{1}}}{\Dsyn{\trm{2}}}$.
		\item If $\trm{1}\beeq\trm{2}$, then $\Dsyn{\trm{1}}\beeq \Dsyn{\trm{2}}$.
	\end{itemize}
\end{lemma}

Our macro $\Dsynsymbol$ can be seen as a class of macros, since it 
depends on the target language. More precisely, it depends on what $\tangentreals $ implements (see Section \ref{sub:target-language-syntax}).

As an example, for the program of Equation \ref{eqn:taylor-series-underflow-example},
$\Dsynsymbol$ computes, modulo some $\beta\eta$-equality to aid legibility, the following derivative (where we also define $\Dsyn{\mathbf{int}}\defeq\mathbf{int}$,
$\Dsyn{\trm{1}<\trm{2}<\trm{3}}\defeq \tFst(\Dsyn{\trm{1}})<\tFst(\Dsyn{\trm{2}})<\tFst(\Dsyn{\trm{3}})$, and $\partial_i(+)(x,y)\defeq 1$):
\begin{align*}
\tItFrom{\Big(\begin{array}{l}
	\pMatch{ z}{i}{\tPair{y'_1}{y'_2}}{}
	\letin{\tPair{y_1}{y_2}}{
	\Dsyn{t}(i, x)}{}\\
	\vMatch{-\epsilon < y_1 < \epsilon}{\tInl\_ \To \tInr x\mid \tInr\_\To \tInl \tPair{i + 1}{\tPair{y_1 +y'_1}{y_2+y'_2}}}\end{array}\Big)\\}
	{z}{\tPair{0}{0}}.
\end{align*}

\subsection{AD transformation as a $CBV$ model morphism}\label{subsect:AD-transofmration-as-a-CBV-morphism}
By the universal property of $\left( \SynV, \SynT,\Synfix , \Synit   \right) $ established in Proposition \ref{prop:section-universal-property-syntax}, the assignment defined in Figure~\ref{fig:assignment-AD-functor} induces a unique $CBV$ model morphism
\begin{equation}\label{eq:macro-as-a-functor}
	\DSyn : \left( \SynV, \SynT,\Synfix , \Synit   \right)\to \left( \SynVt, \SynTt,\Synfixt , \Synitt   \right).
\end{equation} 
\begin{figure}[!ht]
	\fbox{\parbox{0.98\linewidth}{\begin{minipage}{\linewidth}\noindent\input{TEX/Assignment-AD-transformation}\end{minipage}}}
	\caption{AD assignment.
		\label{fig:assignment-AD-functor}}
\end{figure} 

\textit{The macro $\Dsynsymbol $ defined in Figure \ref{fig:ad1} is encompassed by \eqref{eq:macro-as-a-functor}.
}

%% file: TEX/d-types1.tex
\[
\begin{array}{l@{\ }l@{\ }l@{\qquad}l@{\ }l@{\ }l@{\qquad}l@{\ }l@{\ }l}
\Dsyn{\reals} &\defeq& {\reals}\t*{\tangentreals} & \Dsyn{\Init} &\defeq& \Init
& \Dsyn{\ty{1}\t+\ty{2}}&\defeq& \Dsyn{\ty{1}}\t+ \Dsyn{\ty{2}}\\
\Dsyn{\Unit}&\defeq& \Unit &
\Dsyn{\ty{1}\To\ty{2}} &\defeq& \Dsyn{\ty{1}}\To\Dsyn{\ty{2}} &
\Dsyn{{\ty{1}}\t*{\ty{2}}} &\defeq& {\Dsyn{\ty{1}}}\t*{\Dsyn{\ty{2}}}
\end{array}    
\]

%% file: TEX/d-terms1.tex
\[
    \begin{array}{ll}
    \Dsyn{\var{1}} \defeq \var{1} & \Dsyn{\letin{\var{1}}{ \trm{1}}{\trm{2}}}\defeq \letin{\var{1}}{\Dsyn{ \trm{1}}}{\Dsyn{\trm{2}}} 
    \\
    \Dsyn{\nvMatch{\trm{3}}}\defeq \nvMatch{\Dsyn{\trm{3}}} \\
    \Dsyn{\tInl \trm{1}} \defeq \tInl\Dsyn{\trm{1}}
     & \Dsyn{\vMatch{\trm{3}}{\begin{array}{l}\;\;\tInl\var{1}\To\trm{1} \\
        \gor \tInr \var{2}\To \trm{2}\end{array}}}\defeq\\
     \Dsyn{\tInr \trm{1}} \defeq \tInr\Dsyn{\trm{1}} &
     \quad \vMatch{\Dsyn{\trm{3}}}{\begin{array}{l}\;\;\tInl\var{1}\To\Dsyn{\trm{1}}\\
        \gor \tInr \var{2}\To \Dsyn{\trm{2}}\end{array}} \\
    \Dsyn{\tUnit}\defeq\tUnit 
     \\ 
     \Dsyn{\tPair{\trm{1}}{\trm{2}}}\defeq \tPair{\Dsyn{\trm{1}}}{\Dsyn{\trm{2}}} & \Dsyn{\pMatch{\trm{3}}{\var{1}}{ \var{2}}{ \trm{1}}}\defeq \pMatch{\Dsyn{\trm{3}}}{\var{1}}{ \var{2}}{\Dsyn{ \trm{1}}}\\
     \Dsyn{\fun{\var{1}}{\trm{1}}}\defeq \fun{\var{1}}{\Dsyn{\trm{1}}} & 
     \Dsyn{\trm{1}\,\trm{3}}\defeq \Dsyn{ \trm{1}}\,\Dsyn{\trm{3}}\\
    \Dsyn{\tItFrom{ \trm{1}}{\var{1}}{\trm{3}}}\defeq 
 & 
    \Dsyn{\rec{\var{1}}{\trm{1}}} \defeq \rec{\var{1}}{\Dsyn{\trm{1}}}\\
    \quad \tItFrom{\Dsyn{ \trm{1}}}{\var{1}}{\Dsyn{\trm{3}}} 
    \end{array}
    \]
    \hrulefill
    \[
        \begin{array}{lll}
            \Dsyn{\cnst{c}} &\defeq & \tPair{\cnst{c}}{\cnzero   }\\
            \Dsyn{\op(\trm{3}_1,\ldots,\trm{3}_n)}&\defeq~
                               &\pMatch{\Dsyn{\trm{3}_1}}{\var{1}_1}{\var{1}_1'}
                               { \ldots \to\pMatch{\Dsyn{\trm{3}_n}}{\var{1}_n}{\var{1}_n'}
                               {\\
                               &&\letin{ \var{2}}{\op(\var{1}_1,\ldots,\var{1}_n)}{} 
                               \\
                               &&\letin{\var{3}_1}{\partial_1\op(\var{1}_1,\ldots,\var{1}_n)}{\ldots}\letin{\var{3}_n}{\partial_n\op(\var{1}_1,\ldots,\var{1}_n)}{}\\
                               && {\tPair{ \var{2}}{\var{1}_1' *\var{3}_1+\ldots
                               +\var{1}_n' *\var{3}_n}}}}\\
            \Dsyn{\tSign{\trm{3}}}&\defeq & \tSign{(\tFst\Dsyn{\trm{3}})}
        \end{array} 
    \]

%% file: TEX/Assignment-AD-transformation.tex
\[
\begin{array}{c}
\DSyn (\reals)  \defeqq  {\reals}\times {\tangentreals} \in  \obb{\SynVt }, \qquad \DSyn ( \cnst{c} ) \defeqq  \pairL\cnst{c}, {\cnst{0}}\pairR \in  \ehom{\SynVt }{\terminall }{{\reals}\times {\tangentreals} },
\\ 
\DSyn (\op ) \defeqq   \fun{\var{2}_1}{\fun{\ldots }{ \fun{\var{2}_n}{\dDSyn{\op }\left( \var{2}_1, \ldots , \var{2}_n\right)}  } }\in \ehom{\SynVt }{(\reals\times \tangentreals)^n  }{ \SynT \left( \reals\times \tangentreals \right)  }, \\ 
 \DSyn  (\tSign)  \defeqq \left( \tSign\circ \proj{1}\right) \in \ehom{\SynVt }{\reals\times\tangentreals} {\SynT \left( \terminall\sqcup \terminall\right)  },  
\end{array}
\] \\
for each primitive operation  $\op\in\Op_n$ ($n\in\NN$) and each constant  $c\in \RRsyntax$, where
\[
\begin{array}{ll}
	\dDSyn{\op }\left( \var{2}_1, \ldots , \var{2}_n\right) \defeq~
	&\pMatch{\var{2}_1}{\var{1}_1}{\var{1}_1'}
	{ \ldots \to\pMatch{\var{2}_n}{\var{1}_n}{\var{1}_n'}
		{\\
			&\letin{ \var{2}'}{\op(\var{1}_1,\ldots,\var{1}_n)}{} 
			\\
			&\letin{\var{3}_1}{\partial_1\op(\var{1}_1,\ldots,\var{1}_n)}{\ldots}\letin{\var{3}_n}{\partial_n\op(\var{1}_1,\ldots,\var{1}_n)}{}\\
			& {\tPair{ \var{2}'}{\var{1}_1' *\var{3}_1+\ldots
					+\var{1}_n' *\var{3}_n}}}.}\\
\end{array} 
\]

%% file: TEX/Semantics_Language.tex
\section{Concrete semantics for the AD transformation}\label{sec:semantics-AD-transformation}
We give a concrete denotational semantics for our source and target languages
in terms of $\omega$-cpos.
In fact, our semantics for the target language will be parameterised by $k\in \NN\cup \{\infty\}$.
This parameter allows us to give a uniform treatment of various variants of AD.
For basic forward mode AD, we will use $k=1$.
Other $k\in \NN$ correspond to vectorised forms of forward mode AD, and $k=\infty$ is primarily of interest for dual numbers reverse AD which can be viewed as an optimised version of a vectorised forward AD with dynamically sized tangent vectors. 

We will use these semantics to phrase and prove correctness of AD in the rest of this paper.
We also recall some facts about and fix notation for derivatives, in order to phrase sufficient and necessary conditions on the semantics of primitive operations and their AD transformations.

\subsection{Basic concrete model}\label{subsec:wcPO-Basic-Model}
The most fundamental example of a $CBV$ $\wCpo$-pair is given by $\left( \wCpo , \monadwP{-}\right) $ where $\monadwP{-} $ is the 
(lax idempotent) monad that freely adds a least element $\leastelement $ to each \wcpo{}. Indeed, of course, $\ehom{\wCpo}{W}{\monadwP{Y}} $ is pointed for any pair 
$\left( W,Y\right)\in \obb{\wCpo}\times \obb{\wCpo }$.

We consider the product $\left( \wCpo , \monadwP{-}\right)\times \left( \wCpo , \monadwP{-}\right) = \left( \wCpo\times \wCpo , \monadwP{-}  \right) $, where, by abuse of language, $\monadwP{(C, C')} =\left( \monadwP{C}, \monadwP{C'} \right)  $.
By Lemma \ref{UnderlyingCBVmodel}, we obtain $CBV$ models \begin{align*}&\CBVU \left( \wCpo , \monadwP{-}\right)\qquad\text{and}\\ 
 &\CBVU \left( \wCpo\times \wCpo  , \monadwP{-}\right) = \CBVU \left( \wCpo , \monadwP{-}\right)\times \CBVU \left( \wCpo , \monadwP{-}\right). \end{align*}

 For example, the program from Equation \ref{eqn:taylor-series-underflow-example}
 is interpreted as the function 
 \begin{align}\label{eqn:example-taylor-semantics}
\RR &\to \Lift{\RR}\\
x&\mapsto \left\{
\begin{array}{ll}
	\bot & \text{if }N_{\sem{\trm{1}},x}=\infty\\
	\bot & \text{if }\sem{\trm{1}}(i,x)=\epsilon\text{ for some }i<N_{\sem{\trm{1}},x}\\
\sum_{i=0}^{N_{\sem{\trm{1}},x}-1}\sem{\trm{1}}(i,x) &
\text{otherwise}
\end{array}\right.
 \end{align}
 where
 \begin{itemize}
	\item $\lceil z\rceil^\epsilon = z $ if $|z|>\epsilon$, $\lceil z\rceil^\epsilon = 0 $ if $|z|<\epsilon$ and $\lceil z\rceil^\epsilon = \bot $ otherwise;
	\item $N_{\sem{\trm{1}},x}$ is the smallest natural number $i$ such that 
	$\lceil\sem{\trm{1}}(i,x)\rceil^\epsilon=0$.
 \end{itemize}

\subsection{Differentiable functions and interleaved derivatives}\label{subsec:differentiable-functions-semantics}
\textit{Henceforth, unless stated otherwise, the cartesian spaces $\RR ^n$ and its subspaces are endowed with the respective discrete $\wCpo $-structures, in which $r\leq r'$ if and only if $r=r'$.} 

\begin{definition}[Interleaving function]
For each $(n,k)\in\NN \times \left( \NN\cup\left\{ \infty \right\}\right)  $, denoting by $\NNN{n}$ the set $\left\{ 1,\ldots , n\right\}$, we define the isomorphism (in $\wCpo$ with the respective discrete $\wCpo$-structures)
\begin{eqnarray}
	\intle{n, k} : & \RR ^n\times \left( \RR ^k\right) ^n & \to \left( \RR \times \RR ^k \right) ^{n}\\
	&\left( (x_j)_{j\in\NNN{n}}, (y_j)_{j\in\NNN{n}} \right)  & \mapsto  \left( x_j, y_j \right) _{j\in\NNN{n}}.\nonumber 
\end{eqnarray}
For each open subset $U\subset \RR^n$, we denote by $\intle{n, k}^U : U\times \left( \RR ^k\right) ^n \to  \intle{n,k}\left( U\times\left( \RR ^k\right) ^n \right) $
the isomorphism obtained from restricting $\intle{n,k} $. 
\end{definition} 

In Definition~\ref{def:total-derivative}, Remark~\ref{eq:the-empty-nowhere-defined-case} and Lemma~\ref{lem:obvious-gluing}, let $\displaystyle g : U\to \coprod _{j\in L} V_j  $ be a map where $U$ is an open subset of $\RR ^n $, and, for each $i\in L$, $V_i $ is an open subset of $\RR ^{m_i}$. 	

\begin{definition}[Derivative]\label{def:total-derivative}
The map $g$ is \textit{differentiable} if, for any $i\in L$, $g^{-1}\left( V_i \right) =W _i $ is open in $\RR ^n$ and the restriction $g|_{W_i} : W_ i \to V_i $ is differentiable w.r.t the submanifold  structures $W_i\subset \RR ^{n}$ and $ V_i \subset \RR ^{m_i} $.		 
In this case, for each $k\in  \left( \NN\cup\left\{ \infty \right\}\right) $, we define the function: 
\begin{eqnarray}
\dDSemtotaltra{k}{g} : & \intle{n, k}\left(  U\times\left( \RR ^k\right)  ^n\right)  &\to \coprod _{j\in L}\left( \intle{{m_j}, k}\left( V_i\times\left( \RR ^k\right)  ^{m_i}\right)\right)\label{eq:definition-of-total-derivative-k-interleaved} \\
& z &\mapsto \coproje{m_j}\circ\intle{{m_j},k}^{V_j}\left( g(x), \vectoraslineartransformation{w}\cdot g'(x)^t\right) ,  \text{if } \intle{n, k} ^{-1} \left( z \right) =\left( x, w \right)\in W_i\times \left(\RR ^k\right) ^n 		\nonumber
\end{eqnarray}
	in which $\vectoraslineartransformation{w}$ is the linear transformation $ \RR ^n\to \RR ^k $ corresponding to the vector $w$, $\cdot $ is the composition of linear transformations, $\coproje{m_i}$ is the obvious $ith$-coprojection of the coproduct (in the category $\wCpo$),    and $g'(x)^t:\RR^{m_i}\to\RR^n $ is the transpose of the derivative $g'(x) :\RR ^n\to\RR ^{m_i} $ of $g|_{W_i} : W_ i \to V_i $ at $x\in U$.
\end{definition}
\begin{remark}\label{eq:the-empty-nowhere-defined-case}
	It should be noted that, in Definition~\ref{def:total-derivative}, $W_i$ might be empty for some $i\in L $. In this case, $g|_{W_i} : W_ i \to V_i $ is trivially differentiable. Analogously,	
	$U$ might be empty. In this case, the function \textit{$g$ is differentiable and $\dDSemtotaltra{k}{g}$ is the unique morphism with domain $\emptyset$ and codomain as in \eqref{eq:definition-of-total-derivative-k-interleaved}}. 
\end{remark}


 
\begin{lemma}\label{lem:obvious-gluing}
Let $\dot{g}$ be a function with domain as in \eqref{eq:definition-of-total-derivative-k-interleaved}. The map $g$ is differentiable and $\dot{g} =\dDSemtotaltra{k}{g}   $ if, and only if, $g\circ \alpha $ is differentiable and $\dot{g}\circ \dDSemtotaltra{k}{\alpha } = \dDSemtotaltra{k}{\left(g\circ \alpha\right)}$  for any differentiable map $\alpha : \RR ^n \to U $.
\end{lemma} 	
	

\begin{definition} [Differentiable partial maps]\label{def:partial-maps-differentiable-and-derivative}
Let $ \displaystyle h : \coprod_{r\in K } \RR ^{n_r} \to\monadwP{ \coprod_{j\in L } \RR ^{m_j}   }  $ 
be a morphism in $\wCpo $. We say that $h$ is differentiable if, for each $i\in K$, the component 
$\displaystyle h_i := h\circ \coproje{i} : \RR ^{n_i }\to \monadwP{ \coprod_{j\in L } \RR ^{m_j}   } $ satisfies the following two conditions:
\begin{itemize} 
\item $\displaystyle h_i ^{-1}\left( \coprod_{j\in L } \RR ^{m_j}  \right)  = U_i $ is open in $\RR ^{n_i }$; 
\item the corresponding total function \eqref{eq:notation-total-function-associated} is differentiable.
\end{itemize} 
\noindent\begin{minipage}{.4\linewidth}
\begin{equation}\label{eq:notation-total-function-associated} 
	\unlift{h_i} = h|_{{U_i}} : U_i\to \coprod_{j\in L } \RR ^{m_j}  
\end{equation} 	
\end{minipage}%
\begin{minipage}{.6\linewidth}
\begin{equation}\label{eq:derivative-for-partial-functions} 
	\displaystyle\dDSemtra{k}{h} :\coprod_{r\in K } \left( \RR \times  \RR ^k  \right)^{n_r} \to\monadwP{\coprod_{j\in L } \left( \RR \times  \RR ^k  \right)  ^{m_j} } 
\end{equation}  
\end{minipage}\\
In this case, for each $k\in \NN\cup\left\{\infty \right\} $, we define \eqref{eq:derivative-for-partial-functions}  
to be the morphism induced by $ \cpairL \dDSemtra{k}{h_r } \cpairR _{r\in K}$ where, for each $i\in K$, $\dDSemtra{k}{h_i }$ is defined by \eqref{eq:derivative-parially-defined-functions}, which is just the corresponding canonical extension of the map $\dDSemtotaltra{k}{h_i}$. 
\begin{eqnarray}
	\dDSemtra{k}{h_i} : & \left( \RR \times \RR  ^k \right) ^{n_i} &\to\monadwP{\coprod_{j\in L } \left( \RR \times  \RR ^k  \right) ^{m_j} } \label{eq:derivative-parially-defined-functions}\\
	&z &\mapsto 
	\begin{cases}
		\dDSemtotaltra{k}{h_i}\left( z\right) , & \text{if } z\in \intle{{n_i}, k}\left(  U_i\times\left( \RR ^k\right)  ^{n_i}\right)\subset \left( \RR \times \RR  ^k \right) ^{n_i};
		\\
		\leastelement  , & \text{otherwise.} 
	\end{cases}\nonumber
\end{eqnarray}
\end{definition}

For example, the partial function $h$ of Equation \ref{eqn:example-taylor-semantics},
has the following derivative $\dDSemtotaltra{k}(h)$:
\begin{align}
	\RR\times \RR^k &\to \Lift{(\RR\times \RR^k)}\\
	(x,v)&\mapsto \left\{
	\begin{array}{ll}
		\bot & \text{if }N_{\sem{\trm{1}},x}=\infty\\
		\bot & \text{if }\sem{\trm{1}}(i,x)=\epsilon\text{ for some }i<N_{\sem{\trm{1}},x}\\
	\sum_{i=0}^{N_{\sem{\trm{1}},x}-1}\dDSemtotaltra{k}(\sem{\trm{1}}(i,-))(x,v) &
	\text{otherwise}
	\end{array}\right.
\end{align}

 \subsection{The semantics for the source language}\label{subsect:semantics-source-language}
 We give a concrete semantics for our language, interpreting it in the $CBV$ $\wCpo$-pair $\left( \wCpo , \monadwP{-}\right) $.
 
 We denote by $\RR $ the discrete \wcpo{} of real numbers, in which $r\leq r'$ if and only if $r=r'$, and we define  
 $ \signR : \RR \to \monadwP{\terminall  \sqcup \terminall } $
 by \eqref{eq:Definition-SIGN-semantics},  where $\coproje{1}, \coproje{2} : \terminall \to \terminall\sqcup\terminall $ are the two coprojections of the coproduct.\\ 
\noindent\begin{minipage}{.5\linewidth}	
   \begin{equation}\label{eq:functor-semantics} 
	\sem{-} : \left( \SynV, \SynT,\Synfix , \Synit   \right) \to \CBVU\left( \wCpo , \monadwP{-}\right)  
\end{equation} 
	\normalsize
\end{minipage}%
\begin{minipage}{.5\linewidth}
 \begin{equation}\label{eq:Definition-SIGN-semantics}
	\signR (x) =
	\begin{cases}
		\leastelement , & \text{if } x = 0
		\\
		\coproje{1} (\ast )  , & \text{if } x<0 \\
		\coproje{2} (\ast )  , & \text{if } x>0 
	\end{cases}
\end{equation}
\end{minipage}  \\
 By the universal property of $\left( \SynV, \SynT,\Synfix , \Synit   \right) $, there is only one $CBV $ model morphism \eqref{eq:functor-semantics}  
consistent with the assignment of Figure~\ref{fig:assignment-semantics}  where $\semanc{c}$ is the constant that $\cnst c$ intends to implement, and, for each $\op\in\Op _n $,  $\seman{f}$ is the partial map that $\op$ intends to implement.
 \begin{figure}[!ht]
 	\fbox{\parbox{0.98\linewidth}{\begin{minipage}{\linewidth}\noindent\input{TEX/Assignment-semantics}\end{minipage}}}
 	\caption{Semantics' assignment for each primitive operation  $\op\in\Op_n$ ($n\in\NN$) and each constant  $c\in \RRsyntax$.
 		\label{fig:assignment-semantics}}
 \end{figure} 
 
The $CBV$ model morphism \eqref{eq:functor-semantics}  (or, more precisely, the underlying functor of the $CBV$ morphism $\sem{-}$) gives the semantics for the source language.
 Although our work 
 holds for more general contexts, we consider the following assumption over the semantics of our language.
 
 \begin{assum}\label{ass:differentiable-functions}
 	For  each $n\in\NN$ and  $\op\in\Op _n $,  $\sem{\op } = \seman{f}: \RR ^n \to \monadwP{\RR } $
 	is differentiable.
 \end{assum}

 \subsection{The $k$-semantics for the target language}\label{subsect:k-semantics-target-language}
For each $k\in\NN \cup\left\{ \infty \right\} $, we define the $k$-semantics for the target language by interpreting 
$\tangentreals$ as the vector space $\RR^k$.
Namely, we extend the semantics $\sem{-} $ of the source language into a $k$-semantics
of the target language. More precisely, by Proposition \ref{prop:section-universal-property-syntax},
there is a unique $CBV$ model morphism \eqref{eq:semantics-of-the-target-as-a-functor} that extends $\sem{-} $ and is consistent with the assignment given by the vector structure \eqref{eq:k-semantics-for-the-target-language} together with the projection (coprojection) $\semt{k}{\tangentprojection{i}{}   } : \RR ^k \to \RR ^i $ if $i\leq k $ ($i\geq k$), for each $i\in\NN ^\ast $. 
\begin{eqnarray}
	\semt{k}{-}: \left( \SynVt, \SynTt,\Synfixt , \Synitt   \right) &\to &\CBVU\left( \wCpo , \monadwP{-}\right)\label{eq:semantics-of-the-target-as-a-functor}\\ 
\left( \semt{k}{\tangentreals }, \semt{k}{ + }, \semt{k}{ \ast }, \semt{k}{\cnzero } \right) &:= & \left( \RR ^k,  + ,  \ast , 0 \right) \label{eq:k-semantics-for-the-target-language} 
\end{eqnarray} 

\subsection{Soundness of $\Dsynsymbol$ for primitive operations}\label{sub:sound-for-primitives}

\begin{definition}[Sound for primitives]\label{def:sound-for-primitives}
	A macro $\Dsynsymbol $ as defined in Figure \ref{fig:ad1} and its corresponding $CBV$ model morphism $\DSyn $ as defined in \eqref{eq:macro-as-a-functor} are \textit{sound for primitives} if, for any primitive $\op\in\Op $,
	$ \semt{k}{\Dsyn{\op }} = \dDSemtra{k}{\sem{\op}}	$ for any $k$.  
\end{definition}

For each $j\in\NNN{n}$, given a differentiable function $f : \RR ^n\to \monadwP{\RR}$,
we denote by $\dDSemj{ f }{j} : \RR ^n\to \monadwP{\RR\times\RR } $ the function defined by $\dDSemj{f }{j} \left( x_1,\ldots,x_n\right) = \dDSemtra{1}{f }\circ \intle{n,1} \left( (x_1,\ldots,x_n), \canonicalbasise{n}{j}\right) $, where $\canonicalbasise{n}{j} $ the $j$-th vector of the canonical basis of $\mathbb{R}^n $.

\begin{lemma}\label{lem:sound-for-primitives}
The macro $\Dsynsymbol $ defined in Figure \ref{fig:ad1} is sound for primitives provided that
	\begin{equation}
	\sem{\tPair{\op(\var{2}_1,\ldots,\var{2}_n)}{\partial_j\op(\var{2}_1,\ldots,\var{2}_n) }  } 
	= \dDSemj{\sem{\op } }{j}  ,
\end{equation}
for any primitive operation $\op\in\Op _ n $ of the source language.
\end{lemma}

%% file: TEX/Assignment-semantics.tex
\[
\begin{array}{c}
	\sem{\reals} \defeqq  \RR   \in  \obb{\wCpo } ; \qquad
\sem{\cnst c}\defeqq \semanc{c} \in  \ehom{\wCpo}{\terminall}{\RR }; \\ 
\sem{\op }\defeqq \seman{f} \in  \ehom{\wCpo}{ \RR^n }{ \monadwP{\RR}  };\quad
\sem{\tSign}\defeqq \signR  \in \ehom{\wCpo}{\RR } { \monadwP{\terminall\sqcup \terminall} } . 
\end{array}
\] 

%% file: TEX/sconing.tex
\section{Enriched scone and subscone}\label{sec:LogicalRelations-basic-proof}
Here, we present general, reusable results about logical relations proofs for languages 
with recursive features.
We phrase these in terms of category theory.
Concretely, we discuss two categorical perspectives on logical relations, both of which are constructions to build a new categorical semantics out of two existing semantics $\catB$ and $\catD$.
The first perspective, called the \emph{scone}, is as simple as a plain comma category  $\catD\downarrow G$ of the identity along a suitable functor $G:\catB\to\catD$ between the two existing semantics.
It gives a proof-relevant perspective in which we may distinguish different witnesses demonstrating the truth of a predicate.
The second perspective, called the \emph{subscone}, arises as a suitable reflective subcategory of the scone.
It crucial property is that its objects are chosen such that they represent only proof-irrelevant predicates, meaning that we can think of its morphisms simply as $\catB$-morphisms that respect the predicates.

Here, we focus, in particular, on characterising when the scone and subscone are $\wCpo$-bicartesian closed categories, getting us most of the way to a \emph{CBV} $\wCpo$-pair.
We discuss the remaining ingredient of lifting the (pointed) monad to the subscone in Section \ref{subsect:the-definition-of-the-monad-for-the-logical-relations}\footnote{
\cite{goubault2002logical} develops some general methods for obtaining monad liftings to the scone and subscone.
However, there are many choices for such monad liftings and they need to be chosen depending on the specific application.
}.

\subsection{Scone: proof-relevant categorical logical relations}
Given an $\wCpo$-functor $G:\catB\to\catD$, the comma $\wCpo$-category $\catD\downarrow G$ of the identity along $G$ in 
$\ecat{\wCpo}$ is defined as follows.
\begin{itemize} 
\renewcommand\labelitemi{--}	
	\item  The objects of $\catD\downarrow G$ are triples $(D\in\catD  , C\in\catB  , j:D\to G(C) )$ in which  $j$ is a morphism of $\catD$; 
	we think of these as pairs of a $\catB$-object $C$ and a proof-relevant predicate $(D, j)$ on $G(C)$;
	\item  a morphism $(D,C, j)\to (D', C', h)$ between objects of $\catD\downarrow G$ is a pair \eqref{eq:pair-alpha-scone} making \eqref{eq:comma-morphism-scone} commutative in $\catD $;
	we think of these as $\catB$-morphisms $\alpha_1$ that respect the predicates, as evidenced by $\alpha_0$;
\end{itemize}	
\begin{minipage}{.5\linewidth}
\begin{equation}\label{eq:pair-alpha-scone}
		\alpha = \left( \alpha_0 : D\to D'  , \alpha _1 : C\to C'\right)
	\end{equation} 		
\end{minipage}%
\begin{minipage}{.5\linewidth}
\begin{equation}\label{eq:comma-morphism-scone}
	\diag{morphism-comma-category}
\end{equation}   
\end{minipage}
\begin{itemize} 
	\renewcommand\labelitemi{--}	
\item if $\alpha = \left( \alpha _0 : D\to D'  , \alpha _1 : C\to C'\right), \beta = \left( \beta _0 : D\to D'  , \beta _1 : C\to C'\right) : \left( D, C , j\right)\to \left( D', C', h\right) $,
	are two morphisms of $\catD\downarrow G$, 
	we have that $\alpha\leq \beta $ if $\alpha _0\leq  \beta _0 $ in $\catD $ and $\alpha _1\leq \beta _1$ in $\catB $.
\end{itemize}

Following the approach of  \cite[Section~9]{VAKAR-LUCATELLI2021}, we have:

\begin{theorem}[Monadic-comonadic scone]\label{theo:fundamental-scone-monadic-comonadic}
Let  $G: \catB\to\catD $ be a right $\wCpo $-adjoint functor. Assuming that $\catD $ has finite $\wCpo$-products and $\catB $ has finite $\wCpo$-coproducts, the $\wCpo $-functor
	\begin{equation}\label{eq:forgetful-scone-notation}
		\forgetfulS : \catD \downarrow G\to \catD\times\catB  ,
	\end{equation}
	defined by  $\left( D\in\catD  , C\in\catB  , j:D\to G(C) \right)\mapsto \left( D, C \right)$,
	is $\wCpo$-comonadic and $\wCpo$-monadic\footnote{That is, this $\wCpo$-functor is up to equivalence, the forgetful $\wCpo$-functor from the $\wCpo$-Eilenberg-Moore category of a (co)monad on $\catD\times\catB$.
	\cite{MR0280560} gives a good introduction to the theory of enriched (co)monads.
	}.
		This implies, in particular, that $\forgetfulS$ creates (and strictly preserves) $\wCpo$-limits and colimits\footnote{In this work,
		we are only concerned with \emph{conical} $\wCpo$-limits and colimits of functors $J:\catE\to \catC$ in the sense of $\catC$-objects such that we have natural $\wCpo$-isomorphisms
		$$
		\catC(C, \lim J)\cong \catCat(\catE,\catC)(\Delta_C, J)\qquad\text{ and }\qquad 
		\catC(\colim J, C)\cong \catCat(\catE,\catC)(J,\Delta_C),
		$$
		where we write $\Delta_C$ for the constantly $C$ functor.
		For the more general theory of (weighted) $\catV$-(co)limits, we refer the reader to \cite{kelly1982basic}.
		}. 
\end{theorem}	

By Theorem \ref{theo:fundamental-scone-monadic-comonadic} and the enriched adjoint triangle theorem\footnote{See \cite{MR0233864} for the original adjoint triangle theorem, and \cite[Section~1]{Lucatelli2016} for the enriched version.}, we have:

\begin{corollary}\label{coro:fundamental-scone-monadic-comonadic}
Let  $G: \catB\to\catD $ be a right $\wCpo $-adjoint functor between $\wCpo $-bicartesian closed categories. In this case, $\catD \downarrow G$ is an $\wCpo $-bicartesian closed category. Moreover, if $\catD\times\catB$ is $\wCpo$-cocomplete, so is  $\catD \downarrow G$.
\end{corollary}

Theorem \ref{theo:fundamental-scone-monadic-comonadic} and Corollary \ref{coro:fundamental-scone-monadic-comonadic} are $\wCpo $-enriched versions of the fundamental results of \cite[Section~9]{VAKAR-LUCATELLI2021}. The details and proofs are presented in Appendix~\ref{app-wCPO-enriched-scone}.


\subsection{Subscone: proof-irrelevant categorical logical relations}
\textit{Henceforth, we assume that $\SUBscone{\catD}{G}$ is a full\footnote{That is, all $\SUBscone{\catD}{G}((D,C,j),(D',C',h))\to (\catD\downarrow G)((D,C,j),(D',C',h))$ are isomorphisms of $\omega$-cpos.} reflective\footnote{That is, $\SUBscone{\catD}{G}\to \catD\downarrow G$ has an $\wCpo$-left adjoint.} and replete\footnote{That is, $\SUBscone{\catD}{G}$ is closed under isomorphisms in $\catD\downarrow G$.} $\wCpo $-subcategory of $\catD\downarrow G$. We denote, herein, by $\monadSub$ the idempotent $\wCpo $-monad induced by the $\wCpo$-adjuntion.}


Recall that \textit{a morphism $q$ in $\wCpo$ is full} if its underlying functor is full. In this case, the underlying functor is also faithful and injective on objects.
Moreover, a morphism \textit{$j$ in an $\wCpo$-category $\catB$ is full} if $\ehom{\catB}{B}{j}$ is full in $\wCpo$ for any $B\in\catB$. 

Furthermore, recall that \textit{an $\wCpo$-functor $H:\catW\to\catZ $ is locally full if, for any $\left( X,W \right)\in \obb{\catW}\times\obb{\catW} $, the morphism $H : \ehom{\catW}{X}{W} \to \ehom{\catZ}{HX}{HW} $ is a full $\wCpo$-morphism.}
It should be noted that the $2$-functor underlying a locally full $\wCpo$-functor is \textit{locally fully faithful}. Moreover, since every full morphism in $\wCpo $ is injective on objects, \textit{every locally full $\wCpo$-functor is faithful (locally injective on objects). } 

\begin{assum}\label{assum:subscone-assumptions}
We require that:
\begin{enumerate}[(Sub.1)]
	\item whenever $\left( D\in\catD , C\in\catB, j\right)\in\SUBscone{\catD}{G} $, $j$ is a full morphism  in $\catB $;\label{assumption:condition:proof-irrelevant}
    \item $G: \catB\to\catD $ is a right $\wCpo$-adjoint functor between $\wCpo$-bicartesian closed categories;
	\item $\monadSub$ strictly preserves $\wCpo$-products;\label{assumption:cartesian-idempotent-monad}
	\item Diag.~\eqref{eq:forgetful-indentity-on-C-monad-subscone} commutes.
\end{enumerate}
\end{assum}
\begin{minipage}{.5\linewidth}
	\small	
	\begin{equation}\label{eq:definition-forgetful-subscone}
	\diag{forgetful-subscone}
\end{equation} 		
\end{minipage}%
\begin{minipage}{.4\linewidth}
	\small
\begin{equation}\label{eq:forgetful-indentity-on-C-monad-subscone}
	\diag{forgetful-indentity-on-C-monad-subscone}
\end{equation} 
\end{minipage}\\
\normalsize
We denote by $\forgetfulSub : \SUBscone{\catD}{G}\to \catB $ the  $\wCpo$-functor given by the composition \eqref{eq:definition-forgetful-subscone} where 
the unlabeled arrow is the full inclusion.

\begin{proposition}\label{prop:limits-and-colimits-subscone}
	The full inclusion 
	$\SUBscone{\catD}{G}\to \catD\downarrow G $
	creates (and strictly preserves) $\wCpo $-limits and $\wCpo$-exponentials. Moreover, if $\catD\downarrow G  $ is $\wCpo$-cocomplete, so is  $\SUBscone{\catD}{G}$.
\end{proposition} 
\begin{proof}
	$\SUBscone{\catD}{G}\to \catD\downarrow G $ is $\wCpo$-monadic and, hence, it creates $\wCpo$-limits. 
	
	By Assumption \ref{assumption:cartesian-idempotent-monad} of Assumption \ref{assum:subscone-assumptions}, $\monadSub$ is commutative and, hence, $\SUBscone{\catD}{G}\to \catD\downarrow G $ creates $\wCpo$-exponentials. 
	
	Since $\monadSub$ is idempotent,  $\SUBscone{\catD}{G}$ is $\wCpo$-cocomplete whenever $\catD\downarrow G $ is. 
\end{proof}

\begin{corollary}\label{coro:limits-and-colimits-subscone}
$\SUBscone{\catD}{G}$ is an $\wCpo $-bicartesian closed category. Moreover, if $\catD\times\catB$ is $\wCpo$-cocomplete, so is  $\SUBscone{\catD}{G}$.
\end{corollary}
\begin{proof}
	It follows from Proposition \ref{prop:limits-and-colimits-subscone} and Corollary \ref{coro:fundamental-scone-monadic-comonadic}.
\end{proof}


%

\begin{theorem}\label{theo:properties-forgetful-subscone}
	The $\wCpo $-functor $\forgetfulSub : \SUBscone{\catD}{G}\to \catB $ is strictly (bi)cartesian closed and locally full (hence, faithful). Moreover, $\forgetfulSub$ strictly preserves $\wCpo$-colimits.
\end{theorem} 
\begin{proof} 
The $\wCpo$-functors $\forgetfulS :\catD\downarrow G\to \catD\times \catB $ and $\pi_{\catB} : \catD \times \catB\to \catB $ strictly preserve $\wCpo$-weighted limits and colimits. Since $\monadSub$ is idempotent and \eqref{eq:forgetful-indentity-on-C-monad-subscone} commutes, this implies that $\forgetfulSub $ strictly preserves $\wCpo$-limits and colimits.

The composition $\pi_{\catB} \circ \forgetfulS $ has a left $\wCpo$-adjoint given by $ C\mapsto \left( \initiall , C, \coproje{\initiall } \right) $. Since the counit of this $\wCpo $-adjunction is the identity
and $\pi_{\catB} \circ \forgetfulS $  strictly preserves $\wCpo$-products, we get that this $\wCpo$-adjunction strictly satisfies the \textit{Frobenius reciprocity condition}. This implies that $\pi_{\catB} \circ \forgetfulS $
strictly preserves $\wCpo$-exponentials. 

Since $\monadSub $ strictly preserves $\wCpo$-products, we get that $\SUBscone{\catD}{G}\to \catD\downarrow G $ strictly preserves $\wCpo $-exponentials as well. Therefore, $\forgetfulSub $ strictly preserves $\wCpo$-exponentials.

The locally fully faithfulness (and, hence, faithfulness) of $\forgetfulSub $ follows from Condition \ref{assumption:condition:proof-irrelevant} of  Assumption \ref{assum:subscone-assumptions}. 
\end{proof} 

\begin{remark}[Proof-irrelevance]\label{rem:defining-morphism-subscone} 
	Condition \ref{assumption:condition:proof-irrelevant} of Assumption \ref{assum:subscone-assumptions} ensures that our subscone indeed gives us a 
	proof-irrelevant approach to logical relations: in particular, as stressed above, it implies that $\forgetfulSub $ is faithful. Given objects $(D, C, j), (D', C', j')$ and  a morphism $f: C\to C' $  in $\catB$, if there is $\alpha : (D, C, j)\to (D', C', j')$ satisfying 
$\forgetfulSub (\alpha ) = f $, then $\alpha$ is unique with this property. In this case, \textit{we say that $f$ defines a morphism $(D, C, j)\to (D', C', j')$ in $\SUBscone{\catD}{G}$}.

	Generally, we see a trade-off between using proof-relevant logical relations proofs via an interpretation in the scone or  proof-irrelevant ones via an interpretation in the subscone.
	The scone is generally better behaved as a category, as it tends to be both monadic and comonadic by Theorem \ref{theo:fundamental-scone-monadic-comonadic}, while the subscone tends to only be monadic.
	The objects and morphisms of the subscone, however, can be simpler to work with, as we do not need to track witnesses thanks to their uniqueness.
	In the rest of this paper, we work with the (proof-irrelevant) subscone, mostly to conform to conventions in the literature. 
\end{remark}

%% file: TEX/CorrectnessAD.tex
\section{Correctness of Dual Numbers AD}\label{sec:Basic-orrectness-for-forward-AD}
In this section, we show that, as long as the macro $\Dsynsymbol $ defined in Figure \ref{fig:ad1} is sound for primitives  and $\tangentreals $ implements $\RR ^k $, 
$\Dsynsymbol $ is correct according to the $k$-specification below. More precisely, we prove that:
\begin{theorem}\label{theo:main-theorem-section-proof}
		Assume that $\tangentreals $ implements the vector space $\RR ^k$, for some $k\in\NN\cup\left\{ \infty \right\}$. For any program $\var{1}:\ty{1}\vdash \trm{1}:\ty{2}$  where	$\ty{1},\ty{2}$ are data types (i.e., types not containing function types), we have that $\sem{\trm{1}} $ is differentiable and, moreover, 
		\begin{equation}
			 \semt{k}{\Dsyn{\trm{1} }} = \dDSemtra{k}{\sem{\trm{1} }}
		\end{equation}
	provided that $\Dsynsymbol $ is sound for primitives.
\end{theorem} 	 

We take the following steps to achieve this result:
\begin{itemize}
\item In Section \ref{ssec:basic-setting}, we fix a particular functor $G:\catB\to\catD$ for which to consider the scone, as well as a particular reflective subscone of the scone.
This sets the concrete stage in which our logical relations proof will take place.
\item In Section \ref{subsect:the-definition-of-the-monad-for-the-logical-relations}, we choose a particular lifting of the partiality monad to this subscone, to establish a reasoning principle for derivatives of partial functions.
\item In Section \ref{subsec:LR-assignment}, we fix a lifting of the interpretation of the primitive type $\reals$ to the subscone, establishing a reasoning principle for derivatives of real-valued functions.
We further show that, for a macro $\Dsynsymbol$ that is sound for primitives, $\sem{\Dsyn{-}}_k$ respects the logical relation, hence yields an interpretation of our full language in the subscone.
\item In Section \ref{ssec:ad-log-rel-data-types}, we show that logical relations at data types (composite types not containing function types) also capture correct differentiation.
\item In Section \ref{subsec:proof-basic-correctness-theorem}, we derive our fundamental AD correctness theorem from the interpretation of our language in the subscone and in Sections \ref{sub:forward-mode-types-correctness} and \ref{sub:reverse-mode-types-correctness} we spell out in more detail what this correctness theorem entails for the choice of semantics $\sem{\tangentreals}_k=\RR^k$.
\end{itemize}

\subsection{Fixing a particular subscone $ \SUBscone{\wCpo }{\sconeFUNCTOR{n, k} }$}\label{ssec:basic-setting}
\textit{Henceforth, we follow the notation and definitions established in Section \ref{sec:semantics-AD-transformation}.} In particular, \textit{ unless stated otherwise, the cartesian spaces $\RR ^n$ and its subspaces are endowed with the discrete $\wCpo $-structure, in which $r\leq r'$ if and only if $r=r'$.} 

For each $\left( n, k\right) \in \NN\times\left( \NN\cup\left\{ \infty \right\} \right)  $, we define the  $\wCpo $-functor \eqref{eq:Gnk-wCpo-functor}. We consider the full reflective $\wCpo$-subcategory $ \SUBscone{\wCpo }{\sconeFUNCTOR{n, k} }$ of $\wCpo\downarrow\sconeFUNCTOR{n, k} $ whose objects are triples 
\eqref{eq:triples-that-define-subscone}  such that $j$ is full (and, hence, injective on objects). 
\begin{equation} \label{eq:Gnk-wCpo-functor} 
\sconeFUNCTOR{n, k} \defeqq \ehom{\wCpo\times\wCpo}{\left( \RR^n , \left( \RR\times \RR^k\right) ^n  \right)}{  \left( - , - \right)  } : \wCpo\times \wCpo \to \wCpo 
\end{equation} 
\begin{equation}\label{eq:triples-that-define-subscone} 
	\left( D\in \wCpo ,\, \left( C, C'\right)\in \wCpo\times \wCpo,\, \left( j: D\to   \sconeFUNCTOR{n , k}\left( C, C' \right)\right)\in \wCpo \right)
\end{equation}
That is, we are considering what \cite{bcdg-open-logical-relations}  calls \emph{open logical relations} (where closed logical relations would correspond to the case of $\sconeFUNCTOR{n, k}=\ehom{\wCpo\times\wCpo}{(1,1)}{-}$).

The $\wCpo$-functor $\sconeFUNCTOR{n, k} $ together with $ \SUBscone{\wCpo }{\sconeFUNCTOR{n, k} }$ satisfies Assumption \ref{assum:subscone-assumptions}. Therefore:

\begin{proposition} \label{prop:subscone-wcpo-cartesian-Lnk}
 $\SUBscone{\wCpo }{\sconeFUNCTOR{n, k} }$ is a cocomplete $\wCpo $-cartesian closed category. Moreover, the forgetful $\wCpo$-functor
$\forgetfulSub _{n,k} : \SUBscone{\wCpo }{\sconeFUNCTOR{n, k} }\to \wCpo\times \wCpo $ 
is locally full and strictly cartesian closed. Furthermore, it strictly preserves $\wCpo$-colimits. 
\end{proposition} 
\begin{proof}
	It follows from Corollary \ref{coro:limits-and-colimits-subscone} and Theorem \ref{theo:properties-forgetful-subscone}. 
\end{proof}

\subsection{Lifting the partiality monad to the subscone}\label{subsect:the-definition-of-the-monad-for-the-logical-relations}
Let $(n,k)\in \NN\times \left( \NN\cup\left\{\infty \right\}\right) $. In order to get a categorical model of our language, we need to define a partiality monad for $\SUBscone{\wCpo }{\sconeFUNCTOR{n, k}}$.

We denote by  
$\topologyO{n}$ the set of proper open non-empty subsets of the cartesian space $\RR^n $. For each $U\in \topologyO{n} $,
we define
\begin{eqnarray*} 
\openLift{U, n, k} &\defeqq &\left( \left\{ \left( g: \RR ^n\to U, \dDSemtotaltra{k}{g} \right) : g \mbox{ is differentiable} \right\} , \left( U,  \intle{n,k}\left( U\times \left( \RR ^k\right)  ^n\right)  \right), \mathrm{incl.}\right) \\ 
&\in  &\SUBscone{\wCpo }{ \sconeFUNCTOR{n,k} } .
\end{eqnarray*}

We define the $\SUBscone{\wCpo }{\sconeFUNCTOR{n,k}}$-monad $\monadLR{n, k}{-}$ on $\SUBscone{\wCpo }{\sconeFUNCTOR{n,k}}$ by
\begin{equation}
	\monadLR{n, k}{D, \left( C, C'\right), j} \defeqq  \left( \unlift{\monadLR{n, k}{D, \left( C, C'\right) , j}} , \left( \monadwP{C} , \monadwP{C'}\right) , \morLRmonad{X}	
	\right) 
\end{equation}
where $\unlift{\monadLR{n, k}{D, \left( C, C'\right) , j}} $ is the union
\begin{equation} 
\left\{  \leastelement  \right\}
\sqcup 	
D
 \sqcup 
 \left( \coprod _ {U\in\topologyO{n} } 
 \ehom{\SUBscone{\wCpo }{\sconeFUNCTOR{n, k}}}{\openLift{U, n, k} }{\left( D, \left( C, C'\right) , j\right)}\right)
\end{equation} 
with the full $\wCpo$-substructure of 
$ \sconeFUNCTOR{n, k}  \left( \monadwP{C} , \monadwP{C'} \right)  $ 
induced by the inclusion $\morLRmonad{X} $ which is defined by the following components:
\begin{itemize} 
\item the inclusion $\left\{ \leastelement \right\} \to\sconeFUNCTOR{n, k}  \left( \monadwP{C} , \monadwP{C'} \right)  $ of the least morphism $\leastelement : \left( \RR ^n , \left( \RR  \times \RR ^k\right) ^n \right)\to \left( \monadwP{C}, \monadwP{C'} \right) $ in $\ehom{\wCpo\times\wCpo}{\left( \RR ^n , \left( \RR \times \RR ^k \right) ^n  \right) }{\left( \monadwP{C} , \monadwP{C'} \right) }$;
\item the inclusion of the total functions $\sconeFUNCTOR{n, k}\left( \ee _{C} , \ee _{C'} \right)\circ j : D\to \sconeFUNCTOR{n, k}  \left( C , C' \right)\to \sconeFUNCTOR{n, k}  \left( \monadwP{C} , \monadwP{C'} \right) $;
\item the injections $\displaystyle\ehom{\SUBscone{\wCpo }{\sconeFUNCTOR{n, k}}}{\openLift{U, n, k} }{\left( D, \left( C, C'\right) , j\right)}\to \sconeFUNCTOR{n, k}  \left( \monadwP{C} , \monadwP{C'} \right) $,\newline for  $U\in \topologyO{n}$, defined by
\small 
$$ \left( \alpha _0,  \alpha_1 = \left( \beta _0 : U\to C, \beta _1 : \intle{n, k}\left( U\times \left( \RR ^k \right)  ^n\right) \to C' \right)    \right)  \mapsto    \left( \lift{\beta _0} : \RR ^n\to \monadwP{C} , \lift{\beta _1} :\left( \RR \times \RR ^k\right)  ^n\to \monadwP{C'}  \right) , 
$$
\normalsize 
where 
$\lift{\beta _0} $ and  $\lift{\beta _1} $ are the respective corresponding canonical extensions.
The image of  $\morLRmonad{X} $ forms a sub-$\omega$-cpo because 
the union $\bigcup_{n\in \NN} U_n$  of open sets $U_n$ is open 
and because $D$ is an $\omega$-cpo.
\end{itemize}

For each $(C,C')\in \wCpo\times \wCpo $, the component $\left( \mm _{C} , \mm _ {C'} \right) $ and $\left( \ee _{C} , \ee _ {C'} \right) $  of the multiplication  and the unit  of the monad $\monadwP{-} $ on $\wCpo\times \wCpo $ define morphisms 
\begin{eqnarray}
\lift{\mm } _{\left( D, \left( C, C'\right), j \right) } : &\monadLR{n, k}{ \monadLR{n,k}{D, \left( C, C'\right), j}}&\to \monadLR{n, k}{D, \left( C, C'\right), j}\\
\lift{\ee }_{\left( D, \left( C, C'\right), j \right)} : &\left( D, \left( C, C'\right), j \right) &\to\monadLR{n, k}{D, \left( C, C'\right), j} . \label{eq:unit-of-the-monad-for-the-logical-relations} 
\end{eqnarray}
in $\SUBscone{\wCpo }{\sconeFUNCTOR{n, k}}$.
Therefore,  $\lift{\mm }$ and $\lift{\ee }$ define the  multiplication and the  unit for $\monadLR{n, k}{-}$, completing the definition of our monad.
Analogously, we lift, as morphisms of $\SUBscone{\wCpo}{\sconeFUNCTOR{n, k}} $, the strength of 
$\monadwP{-} $, making $\monadLR{n, k}{-}$ into a strong monad (\textit{i.e.} $\SUBscone{\wCpo}{\sconeFUNCTOR{n, k}} $-enriched monad).

In order to finish the proof that $\left( \SUBscone{\wCpo }{\sconeFUNCTOR{n, k}}, \monadLR{n, k}{-} \right) $ is a $CBV$ $\wCpo $-pair, it is enough to see that, for any pair of objects $\left( D_0, \left( C_0, C_0'\right), j_0 \right) $, $ \left( D_1, \left( C_1, C_1'\right), j_1 \right)$ of $\SUBscone{\wCpo }{\sconeFUNCTOR{n, k}}$,   the least morphism 
$\leastelement : \left( C_0 , C_0 '\right) \to \left( \monadwP{C_1}, \monadwP{C_1'} \right) ,
$
of $\ehom{\wCpo}{C_0}{\monadwP{C_1}}\times\ehom{\wCpo}{C_0'}{\monadwP{C_1'}}  $ 
defines the least morphism
$\left( D_0, \left( C_0, C_0'\right), j_0 \right)\to \monadLR{n, k}{ D_1, \left( C_1, C_1'\right), j_1 } $
in $\SUBscone{\wCpo }{\sconeFUNCTOR{n, k}}$.

Finally, since the underlying endofunctor of the monad $\monadLR{n, k}{-}$, the multiplication and the identity are clearly lifted from $\monadwP{-} $ through $\forgetfulSub _{n,k}  $ as defined above, we have:
\begin{proposition}\label{prop:the-subscone-with-the-monads-yields-a-CBV-wCPO-pair}
	For each $(n,k)\in\NN\times \left( \NN \cup \left\{ \infty \right\}\right) $, $\left(\SUBscone{\wCpo }{\sconeFUNCTOR{n, k}}, \monadLR{n, k}{-} \right) $ is a $CBV$ $\wCpo $-pair. Moreover,
	$\forgetfulSub _{n,k} : \SUBscone{\wCpo }{\sconeFUNCTOR{n,k} }\to \wCpo\times \wCpo$
	is a $CBV $ $\wCpo $-pair morphism between $\left(\SUBscone{\wCpo }{\sconeFUNCTOR{n, k}}, \monadLR{n, k}{-} \right) $ and $\left(\wCpo\times \wCpo, \monadwP{-}\right) $.
\end{proposition}

Therefore, by Lemma \ref{UnderlyingCBVmodel}, $\CBVU\left( \forgetfulSub _{n, k} \right)  $ is a $CBV $ model morphism between the underlying $CBV$ models of 
$\left(\SUBscone{\wCpo }{\sconeFUNCTOR{n, k}}, \monadLR{n, k}{-} \right) $ and $\left(\wCpo\times \wCpo, \monadwP{-}\right) $.

\subsection{Logical relations for $\reals$
and deriving a $CBV$ model morphism}\label{subsec:LR-assignment}

\textit{Henceforth, we assume that the macro $\Dsynsymbol $ is sound for primitives (see Definition~\ref{sub:sound-for-primitives}).}
We establish the $CBV$ model morphism \eqref{eq:CBV-model-morphism-that-gives-the-logical-relations}. 
We start by establishing the logical relations' assignment.

\textit{Let $(n, k)\in\NN\times\left( \NN \cup\left\{ \infty \right\}\right)  $.} We define the object \eqref{assig-object-LR} in $\SUBscone{\wCpo }{\sconeFUNCTOR{n, k }}$.
\begin{equation}\label{assig-object-LR}
	\semLR{n,k}{\reals }\defeqq   \left(\left\{ \left( f : \RR ^n\to \RR , \secondM{f} \right) : f\mbox{ is differentiable, } \secondM{f} =  \dDSemtotaltra{k}{f} \right\}, \left( \RR , \RR \times \RR ^k \right), \mathrm{incl.}   \right)  
\end{equation}
For each $m\in\NN $, $\op\in\Op _m $ and $c\in \RRsyntax $,  we define the morphisms \eqref{defeq:LR-sign}, \eqref{defeq:LR-cnst} and \eqref{defeq:LR-op} in $\wCpo\times\wCpo $, in which 
$\DSyn$, $\sem{-}$ and $\semt{k}{-} $ are the functors underlying the $CBV$ model morphisms respectively defined in \eqref{eq:macro-as-a-functor}, \eqref{eq:functor-semantics} and \eqref{eq:semantics-of-the-target-as-a-functor}. 
\small 
\begin{eqnarray}    
\semLRG{k}{\tSign }&\defeqq & \left( \signR ,   \dDSemtra{k}{\signR } \right) = \left( \signR , \semt{k}{\DSyn\left( \tSign  \right)  }\right) : \left( \RR , \RR\times\RR ^k  \right) \to \left( \monadwP{\terminall\sqcup\terminall }, \monadwP{\terminall\sqcup \terminall } \right)  \label{defeq:LR-sign} \\
\semLRG{k}{\cnst c }&\defeqq & \left( \semanc{c} , \dDSemtra{k}{\semanc{c} }  \right)  : \left( \terminall , \terminall \right) \to \left( \RR , \RR\times\RR ^k \right)  \label{defeq:LR-cnst}\\
\semLRG{k}{\op}&\defeqq & \left( \sem{\op}, \dDSemtra{k}{\sem{\op} }  \right) : \left( \RR ^m, \left(\RR\times \RR^k\right)^m \right) \to \left( \monadwP{\RR } , \monadwP{\RR\times\RR ^k } \right)\label{defeq:LR-op}
\end{eqnarray} \normalsize
By Proposition~\ref{prop:limits-and-colimits-subscone}, 
we have that the product $\semLR{n,k}{\reals  }^m$ in $\SUBscone{\wCpo }{\sconeFUNCTOR{n, k}}$ is given by \eqref{eq:product-LR}. Therefore, \textit{by the chain rule for derivatives}, we have that \eqref{defeq:LR-sign}, \eqref{defeq:LR-cnst} and \eqref{defeq:LR-op}   respectively define the morphisms \eqref{LR:assig-tSign}, \eqref{LR:assig-cnst}, and \eqref{LR:assig-op} in $\SUBscone{\wCpo }{\sconeFUNCTOR{n,k}}$, where $\lift{\terminall}\sqcup\lift{\terminall} $ denotes the coproduct of the terminal $\lift{\terminall} = \left( \terminall , \left( \terminall , \terminall  \right) , \ID \right) $ with itself.
\begin{eqnarray}
	&& \left(\left\{ \left( f_j : \RR ^n\to \RR , \secondM{f}_j\right)_{j\in \NNN{m} } : \secondM{f}_j\mbox{ is differentiable and } \secondM{f}_j =  \dDSemtotaltra{k}{f_j} , \forall j\in \NNN{m} \right\}, \left( \RR , \RR\times \RR^k \right) ^m , \mathrm{incl.}   \right) \nonumber \\
	&& \cong \left(\left\{ \left( f : \RR ^n\to \RR ^m , \secondM{f}\right) : f\mbox{ is differentiable, } \secondM{f} =  \dDSemtotaltra{k}{f} \right\}, \left( \RR ^m , \left(\RR\times \RR ^k\right) ^m \right), \mathrm{incl.}   \right) .\label{eq:product-LR}
\end{eqnarray}
 \\
\noindent\begin{minipage}{.5\linewidth}	
	\begin{equation}\label{LR:assig-tSign}
		\semLR{n, k}{\tSign } : \semLR{n, k}{\reals  } \to  \monadLR{n, k}{\lift{\terminall }\sqcup\lift{\terminall}}
	\end{equation} 
	\normalsize
\end{minipage}%
\begin{minipage}{.5\linewidth}
	\begin{equation}\label{LR:assig-cnst}
		\semLR{n,k}{\cnst c } : \lift{\terminall}  \to  \semLR{n,k}{\reals  }
	\end{equation} 
\end{minipage} 
\begin{equation}\label{LR:assig-op}
	\semLR{n,k}{\op } : \semLR{n, k}{\reals  }^m \to \monadLR{n,k}{\semLR{n,k}{\reals  }} 
\end{equation} 
By the universal property of the $CBV$ model $\left( \SynV, \SynT,\Synfix , \Synit   \right) $, we get:
\begin{proposition}\label{prop:semantics-of-the-LR}
For each $(n, k)\in\NN\times\left( \NN \cup\left\{ \infty \right\}\right)  $, there is only one $CBV$ model morphism 	
\begin{equation}\label{eq:CBV-model-morphism-that-gives-the-logical-relations}
	\semLR{n,k}{-}:\left( \SynVt, \SynTt,\Synfixt , \Synitt   \right)\to \CBVU\left( \SUBscone{\wCpo }{\sconeFUNCTOR{n, k}},  \monadLR{n, k}{-}    \right) 
\end{equation}
that is consistent with the assignment given by \eqref{assig-object-LR}, \eqref{LR:assig-tSign},
\eqref{LR:assig-op},  and \eqref{LR:assig-cnst}.
Moreover, Diag.~\eqref{eq:basic-diagram-commutativity-logical-relations} commutes. 
\begin{equation}\label{eq:basic-diagram-commutativity-logical-relations}
	\diag{basic-logicalrelations}
\end{equation} 
\end{proposition} 	
\begin{proof}
	Both $\left(\sem{-}\times\semt{k}{-}\right)\circ\left( \ID\times\DSyn\right) $ and $\CBVU\left({\forgetfulSub}_{n,k}\right)\circ\semLR{n,k}{-} $
	yield  $CBV$ model morphisms that are consistent with the assignment given by the object $\left( \RR , \RR\times\RR ^k  \right) $ together with
	\eqref{defeq:LR-sign}, \eqref{defeq:LR-cnst} and \eqref{defeq:LR-op}. 
\end{proof}

\subsection{AD Logical Relations for Data Types}\label{ssec:ad-log-rel-data-types}
As a consequence of Proposition \ref{prop:semantics-of-the-LR}, we establish a fundamental result on the logical relations $\semLR{n,k}{-}$ for data types (i.e., types not containing function types) in our setting: namely, Proposition \ref{prop:the-LR-general-case}.
Observe that, by distributivity of products over coproducts, any such data type is isomorphic to $\bigsqcup_{j\in L}\reals^{l_j}$
for some finite set $L$ and $l_j\in\NN$.
Therefore, we start by establishing Lemma \ref{lem:the-LR-case-of-total-functions} about our  logical relations and the coproducts in $\SUBscone{\wCpo }{\sconeFUNCTOR{{n}, k}}$.

\begin{lemma}\label{lem:the-LR-case-of-total-functions}
	Let $\left( n, k\right) \in \NN \times \left( \NN \cup\left\{\infty \right\}\right)  $. If $\displaystyle \left( g, \dot{g}\right) \in\coprod _{j\in L} \semLR{n,k}{\reals  } ^{l_j} $, then $\displaystyle g : \RR ^n\to\coprod _{j\in L} \RR ^{l_j}  $ is differentiable and $\dot{g} = \dDSemtotaltra{k}{g} $.
\end{lemma} 	
\begin{proof}
	By Proposition \ref{prop:subscone-wcpo-cartesian-Lnk}, $\SUBscone{\wCpo }{\sconeFUNCTOR{{n}, k}}$ has coproducts. Moreover, we can conclude that $\displaystyle \left( g, \dot{g}\right) \in\coprod _{j\in L} \semLR{n,k}{\reals  } ^{l_j} $  implies that, for some $r\in L $, we have a pair 
	\begin{equation} 
		\displaystyle\left(\unlift{g} : \RR^{n} \to\RR ^{l_r}, \dDSemtotaltra{k}{g} : \left(\RR\times \RR ^k\right)  ^{n}\to \left(\RR\times \RR ^k\right)  ^{l_r} \right) 
	\end{equation} 
	such that  $\left( g, \dot{g}\right) = \left( \coproje{\RR ^{l_r} }\circ\unlift{g}, \coproje{\left(\RR\times \RR ^k\right)  ^{l_r} }\circ\dDSemtotaltra{k}{g} \right) $.
	Following Definition~\ref{def:total-derivative}, this completes our proof.
\end{proof}
\begin{proposition}\label{prop:the-LR-general-case}
	Let $\left( n, k\right) \in \NN \times \left( \NN \cup\left\{\infty \right\}\right)  $. If $\displaystyle \left( g, \dot{g}\right) \in\unlift{\monadLR{n,k}{ \coprod _{j\in L} \semLR{n,k}{\reals  } ^{l_j} }} $, then $\displaystyle g : \RR ^n\to\monadwP{\coprod _{j\in L} \RR ^{l_j}}  $ is differentiable and $\dot{g} = \dDSemtra{k}{g} $.
\end{proposition} 	
\begin{proof}
	Indeed, by the definition of $\unlift{\monadLR{{n},k}{ - }}$, we have one of the following situations.
	\begin{enumerate}[\textbf{s}1.]
		\item \label{proofLR:case1-nowhere-defined}	$g$ and $\dot{g}  $ are the least morphisms, that is to say, they are constantly equal to $\leastelement $;
		\item \label{proofLR:case1-totalfunctions} the pair $\left( g, \dot{g}\right) $ come from a pair of total functions  $\left( \unlift{g}, \unlift{\dot{g}} \right) \in \displaystyle\coprod _{j\in L} \semLR{n,k}{\reals  } ^{l_j}$;
		\item \label{proofLR:case1-nontotalfunctions} $\displaystyle g^{-1}\left( \coprod _{j\in L} \RR ^{l _j} \right) = W$ is open. Moreover, denoting by \eqref{eq:notation-for-the-corresponding-total-gs} 	the pair consisting of the corresponding total functions, we have that \eqref{eq:alphas-differentiable} holds for any differentiable map $ \alpha : \RR ^n \to W  $.
	\end{enumerate}	
	\noindent\begin{minipage}{.4\linewidth}
		\begin{equation} \label{eq:notation-for-the-corresponding-total-gs}
			\displaystyle\left( \unlift{g} : W\to\left( \coprod _{j\in L} \RR ^{l _j} \right)  , \, \unlift{\dot{g}}  \right)
		\end{equation}   	
	\end{minipage}
	\begin{minipage}{.6\linewidth}
		\begin{equation} \label{eq:alphas-differentiable}
			\left( \unlift{g} \circ \alpha , \,  \unlift{\dot{g} } \circ \dDSemtotaltra{k}{\alpha} \right)\in\coprod _{j\in L} \semLR{n,k}{\reals  } ^{l_j}  .  
		\end{equation} 
	\end{minipage} 	\\
	If \eqref{proofLR:case1-nowhere-defined} holds,  following Definition~\ref{def:partial-maps-differentiable-and-derivative}, we get that $g$ is differentiable and  $\dot{g} = \dDSemtra{k}{g} $ by Remark \ref{eq:the-empty-nowhere-defined-case}. 
	
	In case of \eqref{proofLR:case1-totalfunctions},
	we get $\unlift{g}$ is differentiable and  $\unlift{\dot{g}}  = \dDSemtotaltra{k}{\unlift{g}} $ by Lemma \ref{lem:the-LR-case-of-total-functions}. Hence $g$ is differentiable and $\dot{g} = \dDSemtra{k}{g} $.
	
	Finally, in case of   \eqref{proofLR:case1-nontotalfunctions}, by Lemma \ref{lem:the-LR-case-of-total-functions}, we get that, for any differentiable $ \alpha : \RR ^n \to W  $,  $\unlift{g} \circ \alpha$ is differentiable and
	$\unlift{\dot{g} } \circ \dDSemtotaltra{k}{\alpha} 	$ is well defined and equal to  $\dDSemtotaltra{k}{\left( \unlift{ {g} } \circ \alpha\right) } 	$. By Lemma \ref{lem:obvious-gluing}, this implies that $\unlift{g}$ is differentiable and $\dDSemtotaltra{k}{ \unlift{g }  } = \unlift{\dot{g} } $. Following Definition~\ref{def:partial-maps-differentiable-and-derivative}, this completes the proof that  $g$ is differentiable and  $\dot{g} = \dDSemtra{k}{g} $.
\end{proof}
\begin{corollary}\label{coro:fundamental-LR-conclusion-about-morphisms-subscone}
	Let $k\in\NN\cup\left\{\infty\right\}$. If, for each $i\in \KI$, the morphism   $	\left(g, \dot{g} \right) $ in $\wCpo\times\wCpo $ defines the morphism \eqref{eq:morphism-CBV-wCPO-pair-Subscone-defined} in $\SUBscone{\wCpo }{\sconeFUNCTOR{{s_i}, k}}$, then $\displaystyle g : \coprod _{r\in \KI} \RR ^{s_r}\to\monadwP{\coprod _{j\in L} \RR ^{l_j}}  $ is differentiable and $\dot{g} = \dDSemtra{k}{g} $.\\
\noindent\begin{minipage}{.55\linewidth}
	\begin{equation} \label{eq:morphism-CBV-wCPO-pair-Subscone-defined}
	\mathtt{g} : \coprod _{r\in \KI } \semLR{{s_i},k}{\reals } ^{s _r}  \to \monadLR{{s_i}, k}{ \coprod _{j\in L}\semLR{{s_i},k}{\reals }^{l _j} }
\end{equation} 
\end{minipage}%
\noindent\begin{minipage}{.45\linewidth}
	\begin{equation} \label{eq:LR-the-coprojection-LR}
	\coproje{i} : \semLR{{s_i},k}{\reals } ^{s _i}\to \coprod _{r\in K} \semLR{{s_i},k}{\reals } ^{s _r}
\end{equation} 
\end{minipage}	
\end{corollary} 
\begin{proof}
From the hypothesis, for each $i\in\KI$, we conclude that the pair \eqref{eq:component-of-g-differentiable-LR} defines the morphism \eqref{eq:component-of-g-differentiable-LR-in-subscone}, since $\left( \coproje{\RR^{s_i}} , \coproje{\left( \RR\times\RR ^k\right) ^{s_i}} \right)$ defines the coprojection \eqref{eq:LR-the-coprojection-LR} in $\SUBscone{\wCpo }{\sconeFUNCTOR{{s_i}, k}}$.\\
\noindent\begin{minipage}{.45\linewidth}
\begin{equation}\label{eq:component-of-g-differentiable-LR}
	\left( g_i\defeqq g\circ \coproje{\RR^{s_i}} ,\, \dot{g}_i\defeqq \dot{g} \circ \coproje{\left( \RR\times\RR ^k\right) ^{s_i}}  \right) 
\end{equation} 
\end{minipage}%
\noindent\begin{minipage}{.55\linewidth}
\begin{equation}\label{eq:component-of-g-differentiable-LR-in-subscone}
	\mathtt{g}_i \defeqq \mathtt{g}\circ\coproje{i} : \semLR{{s_i},k}{\reals } ^{s _i}  \to \monadLR{{s_i}, k}{ \coprod _{j\in L}\semLR{{s_i},k}{\reals }^{l _j} } 	
\end{equation}
\end{minipage}	\\
Since 
$\ID _{\RR ^{s_i}} : \RR ^{s_i}\to \RR^{s_i} $ is differentiable, and $\dDSemtotaltra{k}{\left( \ID _{\RR ^{s_i}} \right) }$
is given by the identity $\left( \RR\times\RR ^k\right) ^{s_i}\to\left( \RR\times\RR ^k\right) ^{s_i} $, we conclude that
\begin{equation}\label{eq:fundamental-belonging-relation-for-the-LR-argument-coro}
 	\left(g _i,\dot{g}_i \right)\in\unlift{\monadLR{{s_i},k}{ \coprod _{j\in L} \semLR{{s_i},k}{\reals  } ^{l_j} }}.
\end{equation}
By Proposition \ref{prop:the-LR-general-case}, \eqref{eq:fundamental-belonging-relation-for-the-LR-argument-coro} proves that  $g_i$ is differentiable and $\dot{g}_i = \dDSemtra{k}{g_i}$. Since this result holds for any $i\in \KI$, we conclude that $g$ is differentiable and $\dot{g} = \dDSemtra{k}{g } $.
\end{proof}

\subsection{Fundamental AD correctness theorem}\label{subsec:proof-basic-correctness-theorem}
We prove Theorem \ref{theo:basal-version-of the-correctness-theorem}, which completes the proof of Theorem \ref{theo:main-theorem-section-proof}. 

\begin{theorem}\label{theo:basal-version-of the-correctness-theorem}
 	Let $\displaystyle t: \coprod _{r\in \KI} \reals ^{s _r} \to \SynT\left( \coprod _{j\in L} \reals ^{l _j}  \right)  $  be a morphism in $\SynV$. We have that 
 	$\displaystyle\sem{ t } : \coprod _{r\in \KI} \RR  ^{s _r} \to \monadwP{\coprod _{j\in L} \RR ^{l _j} } $ is differentiable and, for any $k\in \left( \NN\cup\left\{ \infty \right\}\right) $, 
 	$	\semt{k}{\DSyn\left( t \right) } = \dDSemtra{k}{\sem{ t }}  $.
\end{theorem} 	 
\begin{proof} 
We assume that we have $t$ as above. For each $i\in \KI $, the pair \eqref{eq:pair-morphisms-for-LR} is in the image of $\left(\sem{-}\times\semt{k}{-}\right)\circ\left( \ID\times\DSyn\right) =\CBVU\left({\forgetfulSub}_{{s_i},k}\right)\circ\semLR{{s_i},k}{-}  $. This implies that   \eqref{eq:pair-morphisms-for-LR} defines the morphism \eqref{eq:semLR{s,k}{t}}
in $\SUBscone{\wCpo }{\sconeFUNCTOR{{s_i}, k}}$. Therefore, by Corollary \ref{coro:fundamental-LR-conclusion-about-morphisms-subscone}, we conclude that  $\sem{t} $
is differentiable and $\semt{k}{\DSyn{\left( t\right) } }  = \dDSemtra{k}{\sem{t} }$.
\\
\noindent\begin{minipage}{.3\linewidth}
\begin{equation}\label{eq:pair-morphisms-for-LR}
	\left(\sem{t} ,\semt{k}{\DSyn{\left( t\right) } }  \right) 
\end{equation} 
\end{minipage}%
\begin{minipage}{.7\linewidth}
\begin{equation} \label{eq:semLR{s,k}{t}}
	\semLR{{s_i},k}{t} : \coprod _{r\in K} \semLR{{s_i},k}{\reals } ^{s _r}  \to \monadLR{{s_i}, k}{ \coprod _{j\in L}\semLR{{s_i},k}{\reals }^{l _j} }
\end{equation} 
\end{minipage}
\quad\\
\end{proof}

\subsection{Correctness of the dual numbers forward AD}\label{sub:forward-mode-types-correctness}
We assume that $\tangentreals$ implements the vector space $\RR$.
It is straightforward to see that we get forward mode AD out of our macro $\Dsynsymbol $: namely, for a program $\var{1}:\ty{1} \vdash \trm{1}:\ty{2} $ (where $\ty{1} $ and $\ty{2}$ are data types) in the source language, we get a program $\var{1}:\Dsyn{\ty{1}} \vdash \Dsyn{\trm{1}}:\Dsyn{\ty{2}}  $ in the target language, which, by Theorem \ref{theo:main-theorem-section-proof}, satisfies the following properties.
\begin{itemize} 
	\item $\sem{\trm{1} } :  \coprod_{r\in K } \RR ^{n_r} \to\monadwP{ \coprod_{j\in L } \RR ^{m_j}   }   $ is differentiable as in  Definition~\ref{def:partial-maps-differentiable-and-derivative}; 
    \item if $y\in\RR ^{n_i}\cap\sem{\trm{1} }^{-1}\left( \RR ^{m_j} \right) = W_ j $ for some $i\in K $ and $j\in L $,  we have that, for any $w\in \RR ^{n_i} $, denoting  $z: =  \intle{{n_i},1}\left(y,w\right) $, 
    \begin{eqnarray}
    	\semt{1}{ \Dsyn{\trm{1}} } \left( \intle{{n_i},1}\left(y,w\right) \right) &=& \dDSemtra{1}{\sem{\trm{1}} }\left( z\right) = \dDSemtotaltra{1}{\sem{\trm{1}}|_{W_j} }\left( z\right)       	
    	 =  \intle{{m_j},1}\left( \sem{\trm{1} }\left( y\right) ,  \vectoraslineartransformation{w}\cdot\sem{\trm{1} }'(y) ^{t}  \right) \label{eq:computing-forward-AD-out-of-the-main-result}\nonumber\\
    	  &=& \intle{l,1}\left( \sem{\trm{1} }\left( y\right) ,  \sem{\trm{1} }'(y)(w) \right) ,  
    \end{eqnarray}
where $\sem{\trm{1} }'(y) : \RR ^{n_i}\to\RR ^{m_j} $ is the derivative of $\sem{\trm{1} }|_{W_j}: W_j \to \RR ^{m_j} $ at $y$.
\end{itemize} 

\subsection{Correctness of the dual numbers reverse AD}\label{sub:reverse-mode-types-correctness}
We assume that $\tangentreals $ implements the vector space $\RR ^k $, for some fixed $k\in\NN\cup\left\{\infty \right\} $. We consider the respective (co)projections 
$\semanticshandler{k}{s} $ for each $s\in\NN\cup\left\{\infty \right\} $, as defined in \eqref{eq:respective-coprojections} . 
The following shows how our macro encompasses reverse mode AD. 

For each $s\in\NN  ^\ast $ with $s \leq k $, we can define the morphism $\wrapSyncat{s}\defeqq \pairL \proj{j}, \cncanoni{j} \pairR _{j\in \NNN{s}} : \reals ^s\to \left( \reals\times \tangentreals \right) ^s $
in $\SynVt$, which corresponds to the wrapper defined in \eqref{eq:wrap-k-x-Dk} in the target language. 
We denote $\wrapSemcat{s}\defeqq\semt{k}{\wrapSyncat{s}} $. By the definition of the $k$-semantics, it is clear that $\wrapSemcat{s} \left( y \right) = \intle{s,k}\left( y, \canonicalbasise{k}{1}, \ldots , \canonicalbasise{k}{s} \right)  $.

For a program $\var{1}:\reals^{s} \vdash \trm{1}:\reals ^l$ (where $s, l\in \NN^\ast  $), we have that, for any $y\in\sem{\trm{1} }^{-1}\left( \RR ^l \right)\subset \RR ^s  $,
\begin{eqnarray*}
	\semt{k}{\Dsyn{\trm{1}}\circ \wrapSyncat{s}  } \left( y \right) & = &  \dDSemtra{k}{\sem{\trm{1}} }\circ \wrapSemcat{s}\left( y \right) = \dDSemtotaltra{k}{\sem{\trm{1}} }\circ \wrapSemcat{s}\left( y \right)\\
	 & = & 	\dDSemtotaltra{k}{\sem{\trm{1}} }\circ \intle{s,k}\left( y, \canonicalbasise{k}{1}, \ldots , \canonicalbasise{k}{s} \right) \\
	 &=& \intle{l,k}\left( \sem{\trm{1} }\left( y\right) ,  \semanticshandler{s}{k} \sem{\trm{1} }'(y) ^t  \right)
\end{eqnarray*}	 
by Theorem \ref{theo:main-theorem-section-proof}. This gives the transpose derivative $\semanticshandler{s}{k} \sem{\trm{1} }'(y) ^t$ as something of the type $\tangentreals ^l $. This should be
good enough whenever $k = s $, since, in this case,  $\semt{k}{\tangentreals ^l} = \left( \RR^s\right) ^l  $  and  $\semanticshandler{s}{k} =\semanticshandler{k}{k} = \ID $.

In case of $s < k $, if needed, the type can be fixed
by using the handler $\tangentprojection{s}{}$. More precisely, we can define the morphism $$\tangentprojection{l,s}\defeqq \pairL \ID, \tangentprojection{s} \pairR _{i\in\NNN{l}} : \left( \reals \times \tangentreals \right) ^l\to \left( \reals \times \reals ^s \right) ^l   $$ 
and, by the definition of $k$-semantics, we conclude that 
\begin{eqnarray*}
	\semt{k}{\tangentprojection{l,s}{}\circ \Dsyn{\trm{1}}\circ \wrapSyncat{s}  } \left( y \right) & = & \semt{k}{\tangentprojection{l,s}}\circ \intle{l,k}\left( \sem{\trm{1} }\left( y\right) ,  \semanticshandler{s}{k} \sem{\trm{1} }'(y) ^t  \right) \\ 
	&=& \intle{l,k}\left( \sem{\trm{1} }\left( y\right) ,  \semanticshandler{k}{s}\circ\semanticshandler{s}{k} \sem{\trm{1} }'(y) ^t  \right) \\
	& = & \intle{l,k}\left( \sem{\trm{1} }\left( y\right) ,  \sem{\trm{1} }'(y) ^t  \right) ,
\end{eqnarray*}	
since $ \semanticshandler{k}{s}\circ\semanticshandler{s}{k} = \ID$ whenever $s\leq k $.

Again, by Theorem \ref{theo:main-theorem-section-proof}, it is straightforward to generalize the correctness statements above to more general data types $\ty{2} $. Furthermore, it should be noted that, for $k=\infty $ (representing the case of a type of dynamically sized array of cotangents), the above
shows that our macro gives the reverse mode AD for any program $\var{1}:\ty{1} \vdash \trm{1}:\ty{2}$ for data types $\ty{1}$ and $\ty{2} $.
This choice of $k=\infty$ is the easiest route to take for a practical implementation of this form of dual-numbers reverse AD,
as it leads to a single type of cotangent vectors that works for any program.

%% file: TEX/recursive-types.tex
\section{AD for recursive types and ML-polymorphism}\label{sec:recursive-types}
\subsection{Syntax for recursive types}\label{ssec:rec-types-syntax}
We extend both our source and target languages of Section \ref{subsect:source-language} and \ref{sub:target-language-syntax} with ML-style polymorphism and type recursion in the sense of FPC \cite{fiore1994axiomatisation}.
That is, we extend types, values and computations for each of the two languages as \\
\input{TEX/types-values-computations2}
\\
The new values and computations according to the rules in Figure \ref{fig:types2}.
\begin{figure}[!ht]
	\fbox{\parbox{0.98\linewidth}{\begin{minipage}{\linewidth}\noindent\input{TEX/type-system2}\end{minipage}}}
	\caption{Typing rules for the recursive types extension.
		\label{fig:types2}}
\end{figure}\\
Here, kinding contexts $\Delta$ are lists of type variables $\tvar{1}_1,\ldots,\tvar{1}_n$.
We consider judgements $\DGinf{\trm{1}}{\ty{1}}$, where the types in $\Gamma$ and $\ty{1}$ may contain free type variables from $\Delta$.
They should be read as specifying that $\trm{1}$ is a program of type $\ty{1}$, with free variables typed according to $\Gamma$, that is polymorphic in the type variables of $\Delta$.

We use the $\beta\eta$-rules of Figure \ref{fig:beta-eta2}.
\begin{figure}[!ht]
	\fbox{\parbox{0.98\linewidth}{\scalebox{0.92}{\begin{minipage}{\linewidth}\noindent\input{TEX/beta-eta2}\end{minipage}}}
	}\caption{\label{fig:beta-eta2} The standard $\beta\eta$-equational theory for recursive types in CBV.
	}
\end{figure}

 Once a language has recursive types, it is already expressive enough to get term recursion and, hence, iteration. Namely, we can now consider term recursion at type $\ty{1}=\ty{2}\To\ty{3}$ as syntactic sugar. Namely, we first define $\chi\defeqq \trec{\tvar{1}}{\left( \tvar{1}\To\ty{1}\right) }$ and then:
\begin{align}
&\tUnroll\trm{1}\defeq \rMatch{\trm{1}}{\var{1}}{\var{1}}\nonumber\\
&\rec{\var{1}:\ty{1}}{\trm{1}}\defeq \letin{body:\chi\To\ty{1}}{(\fun{\var{2}:\chi}\fun{\var{3}:\ty{2}}{\letin{\var{1}:\ty{1}}{\tUnroll\var{2}\,\var{2}}{\trm{1}\,\var{3}}})}{body (\tRoll\,body)}.\label{eq:expressing-term-recursion-in-terms-of-type-recursion}
\end{align} 

The semantics of the language is, of course, expected to be consistent -- meaning that the interpretations of term recursion and recursive types should be compatible according to the definition above.
Alternatively, we can consider that the source language is given by the basic language with the typing rules given by Figure~\ref{fig:types1} with the corresponding grammar plus the  recursive types established above, while the target language is the source language plus the extension given by the grammar and typing rules defined in Section \ref{sub:target-language-syntax}.

\subsection{Categorical models for recursive types: $rCBV$ models}\label{ssec:cat-models-recursive-types}

Here, we establish the basic categorical model for the syntax of  call-by-value languages with recursive types. \textit{Let $\left(\catV , \monadT\right) $ be a $CBV$ pair and $J:\catV\to\catC $ the corresponding universal Kleisli functor}. Moreover, let  $\ehom{\catCat}{\mathsf{2}}{\ecat{\catV}}$ be the category of morphisms of $\ecat{\catV}$.
 
For each $n\in\NN  $, an \textit{$n$-variable  $\left(\catV , \monadT \right) $-parametric type} (or a $\left(\catV , \monadT \right) $-parametric type of degree $n$) is a
morphism $\pEE{E}  : \left( J^\op \times J\right) ^n \rightarrow  J  $ in $\ehom{\catCat}{\mathsf{2}}{\ecat{\catV}}$. In other words, it
consists of a pair $\pEE{E} = \left( \pE{E}{\catV}, \pE{E}{\catC} \right) $ of $\catV$-enriched functors such that \eqref{eq:parametric-types-in-terms-of-functors} commutes. \textit{A $\left(\catV , \monadT \right) $-parametric type of degree $0$ \eqref{eq:parametric-type-of-degree-0} can be identified with the corresponding object $\catV $.}
\begin{equation}\label{eq:parametric-types-in-terms-of-functors} 
	\diag{obvious-diagram-parametric-types} 
\end{equation}

We denote by $\ParamAll{\catV}{\monadT}$ the collection of all  $\left(\catV , \monadT \right) $-parametric types $\pEE{E} = \left( \pE{E}{\catV}, \pE{E}{\catC}\right) $ of any degree $n\in\mathbb{N} $. As the terminology indicates, the objects of $\ParamAll{\catV}{\monadT}$  play the role of the semantics of parametric types in our language. However, the parametric types in the actual language could be a bit more restrictive. They usually are those constructed out of the primitive type formers. Namely, in our case, tupling (finite products), cotupling (finite coproducts), exponentiation (Kleisli exponential) and type recursion.

 

\begin{definition}[Free type recursion]\label{def:free-type-recursion}
A \textit{free decreasing degree type operator} (\fddt operator) for $\left(\catV , \monadT\right) $  is a function \eqref{eq:fddt-operator} identity on parametric types of degree $0$ which takes each  $(n+1)$-variable $\left(\catV , \monadT \right) $-parametric type $\pEE{E} = \left( \pE{E}{\catV}, \pE{E}{\catC}\right) $ to a  $\left(\catV , \monadT \right) $-parametric type  $\pEE{\fixpointRecT E} = \left( \pE{\fixpointRecT E}{\catV}, \pE{\fixpointRecT E}{\catC} \right) $ of degree $n $, provided that $n\in\NN$.
\begin{eqnarray} 
\fixpointRecT :   \ParamAll{\catV}{\monadT} & \to &  \ParamAll{\catV}{\monadT}\label{eq:fddt-operator} \\
\diag{obvious-diagram-parametric-types-n} & \mapsto &  \diag{obvious-diagram-parametric-types-(n-1)-recursive} \nonumber
\end{eqnarray} 
A \textit{rolling} for \eqref{eq:fddt-operator} is a collection \eqref{eq:rolling} of natural transformations such that \eqref{eq:diag-roll} is invertible for any $\pEE{E} = \left( \pE{E}{\catV},  \pE{E}{\catC} \right) $, that is to say,  $J\left(\rollReT ^{\pEE{E}} \right) $ is a natural isomorphism.\\
\noindent\begin{minipage}{.5\linewidth}
\begin{equation}\label{eq:parametric-type-of-degree-0} 
	\left( \left( \catV ^\op \times \catV\right) ^0 \to \catV, \left( \catC ^\op \times \catC\right) ^0 \to \catC\right) 
\end{equation}		
	\begin{equation}\label{eq:rolling}
		\rolling = \left( \rollReT ^{ \pEE{E} } \right) _ {\pEE{E}= \left( \pE{E}{\catV}, \pE{E}{\catC} \right)\in \ParamAll{\catV}{\monadT}  }  
	\end{equation} 	
\end{minipage}%
\noindent\begin{minipage}{.5\linewidth}
	\begin{equation}\label{eq:diag-roll}
		\diag{roll-basic-diagram}
	\end{equation} 
\end{minipage}\\
A \textit{free type recursion} for $\left(\catV , \monadT\right) $ is a pair 
$\ffixpointRecT = \left( \fixpointRecT , \rolling \right)  $
where $\fixpointRecT$ is an \fddt operator and $\rolling $ is a rolling for $\fixpointRecT$.

\end{definition} 
\begin{definition}[$H$-compatible]
	Let $H$ be a $CBV$ pair morphism between $CBV$ pairs $\left( \catV , \monadT \right) $ and  $\left( \catV ' , \monadT ' \right) $. A pair 
	$\left( \pEE{E}, \pEE{E'}\right) \in \ParamAll{\catV}{\monadT}\times \ParamAll{\catV '}{\monadT '} $ of parametric types
	is \textit{$H$-compatible} if they have the same degree $n$ and the diagram \eqref{eq:H-compatibility-of-paramatric-types} commutes.
	In particular, if $n = 0$, the pair $\left( \pEE{E}, \pEE{E'}\right) $ is $H$-compatible if $H\left( \pE{E}{\catV}\right) = \pE{E'}{\catV} $.
\begin{equation}\label{eq:H-compatibility-of-paramatric-types}
	\diag{diag-H-compatibility-of-paramatric-types}
\end{equation}		
\end{definition} 

\begin{definition}[$rCBV$ models]\label{def:rCBV-models}
An \textit{$rCBV$ model} is a triple $\left(\catV , \monadT, \ffixpointRecT\right) $
where $\left(\catV , \monadT\right) $ is a $CBV$ pair and $\ffixpointRecT$ is a free type recursion for $\left(\catV , \monadT\right) $. 

An \textit{$rCBV$ model morphism} between the $rCBV$ models $\left(\catV , \monadT, \ffixpointRecT \right) $ and $\left(\catV  ', \monadT ', \ffixpointRecT '\right) $ consists of a $CBV$ pair morphism between $\left(\catV  , \monadT \right) $ and 
$\left(\catV  ', \monadT '\right) $ such that, for every $H$-compatible pair
$\left( \pEE{E}, \pEE{E'}\right) \in \ParamAll{\catV}{\monadT}\times \ParamAll{\catV '}{\monadT '} $ of $n$-variable parametric types, $\left( \pEE{\fixpointRecT E},  \pEE{\fixpointRecT E'}\right) $ is  $H$-compatible and, if $n>0$, \eqref{eq:roll-eq-diagram} holds, that is to say, $H\left(\rollReT ^{\pEE{E}} \right) =\rollReT ^{\pEE{E}} _{\left( H^\op\times H \right) ^{n-1} } $. The $rCBV$ models and $rCBV$ model morphisms define a category, \textit{denoted herein by $\RCBVcat$}.
\begin{equation}\label{eq:roll-eq-diagram}
	\diag{roll-eq1-diagram} = \diag{roll-eq2-diagram}
\end{equation}

\textit{There is, then, an obvious forgetful functor $\CBVUrp : \RCBVcat\to \CBVcatund $.}
\end{definition}

\begin{remark} \label{rem:rCBVtoCBV}
We do not use this fact in our work, but every $rCBV$ model has an underlying 
$CBV$ model. More precisely, free term iteration can be defined out of the free term recursion, while the latter can be defined out of the free type recursion (see \eqref{eq:expressing-term-recursion-in-terms-of-type-recursion}). This defines a forgetful functor
\begin{equation}
	\rCBVtoCBV : \RCBVcat\to \CBVcat .
\end{equation}
\end{remark}

\subsection{The $rCBV$ models $\left( \SynVr, \SynTr, \SynffixpointRecT   \right) $ and $\left( \SynVrt , \SynTrt ,\SyntffixpointRecT  \right) $} \label{subsec:Syntax-as-a-categoricalstructure-recursive-types}


We consider the  $rCBV$ model generated by each syntax, that is to say, the free $rCBV$ models coming from the fine-grain CBV translations of the source and target languages. This provides us with the $rCBV$ models 
\begin{equation}  \label{eq:rCBV-models-syntax}
	\left( \SynVr, \SynTr, \SynffixpointRecT   \right) \qquad\mbox{and}\qquad 
\left( \SynVrt , \SynTrt ,\SyntffixpointRecT  \right)
\end{equation} 
with the universal property described in Proposition \ref{prop:universal-property-type-recursive-language}.

\begin{proposition}[Universal Property of the $rCBV$ models \eqref{eq:rCBV-models-syntax}]\label{prop:universal-property-type-recursive-language}
	Let $\left( \catV , \monadT, \ffixpointRecT \right) $  be an $rCBV$ model. Assume that Figure~\ref{fig:assignment-functor-universalproperty-syntax} and Figure~\ref{fig:assignment-functor-universalproperty-syntax-target} are given consistent assignments.
	\begin{enumerate}	
		\item There is a unique $rCBV$ model morphism $H: \left( \SynVr, \SynTr, \SynffixpointRecT   \right) \to \left( \catV , \monadT, \ffixpointRecT \right) $ respecting the assignment of Figure~\ref{fig:assignment-functor-universalproperty-syntax}.	 	
		\item There is a unique $rCBV$ model morphism $\extendedH{H}: \left( \SynVrt , \SynTrt ,\SyntffixpointRecT  \right) \to \left( \catV , \monadT, \ffixpointRecT \right) $ that extends $H$ and respects the assignment of Figure~\ref{fig:assignment-functor-universalproperty-syntax-target}.
	\end{enumerate}	
\end{proposition}

\begin{remark}\label{rem:extending-functors} 
By Proposition \ref{prop:section-universal-property-syntax}, we have (unique) $CBV$ model morphisms $$\incCBVtorCBV :  \left( \SynV, \SynT,\Synfix , \Synit   \right)\to 	\rCBVtoCBV  \left( \SynVr, \SynTr, \SynffixpointRecT   \right) $$
and $$\incTCBVtorCBV : \left( \SynVt , \SynTt ,\Synfixt  , \Synitt    \right)\to 	\rCBVtoCBV \left( \SynVrt , \SynTrt ,\SyntffixpointRecT  \right) $$ that are identity on the primitive operations and types. 

Proposition \ref{prop:universal-property-type-recursive-language} states that $H\mapsto \rCBVtoCBV\left( H \right)\circ \incCBVtorCBV  $ and $\extendedH{H}\mapsto \rCBVtoCBV\left( \extendedH{H} \right)\circ \incTCBVtorCBV  $
give the bijections \eqref{eq:universal-property-bijections} and \eqref{eq:universal-property-bijections2}, respectively, showing that our syntax extension for recursive types give a free \emph{rCBV} model on the syntax without recursive types. 
\footnotesize
\begin{eqnarray}
	\RCBVcat\left(\left(\SynVr{,}\SynTr{,}\SynffixpointRecT\right){,}\left(\catV{,}\monadT{,}\ffixpointRecT\right)\right) &\cong & \CBVcat\left(\left(\SynV{,}\SynT{,}\Synfix{,}\Synit\right){,}\rCBVtoCBV\left(\catV{,}\monadT{,}\ffixpointRecT\right)\right) \label{eq:universal-property-bijections}\\
	  \RCBVcat\left(\left( \SynVrt , \SynTrt ,\SyntffixpointRecT  \right) {,}\left(\catV{,}\monadT{,}\ffixpointRecT\right)\right) & \cong  & \CBVcat\left(\left( \SynVt , \SynTt ,\Synfixt  , \Synitt    \right){,}\rCBVtoCBV\left(\catV{,}\monadT{,}\ffixpointRecT\right)\right)\label{eq:universal-property-bijections2}
\end{eqnarray}
\normalsize

\end{remark} 	

\subsection{Automatic differentiation for languages with recursive types}\label{sub:extended-macro}
We extend our definition of AD to recursive types in Figure \ref{fig:ad2}.
We note that our extension is compatible with our previous definitions if we view term recursion (and iteration) as syntactic sugar.
\begin{lemma}[Type preservation]
If $\DGinf{\trm{1}}{\ty{1}}$, then $\Delta\mid \Dsynplain{\Gamma}\vdash \Dsynplain{\trm{1}}:\Dsynplain{\ty{1}}$.
\end{lemma}

\begin{figure}[!ht]
	\fbox{\parbox{0.98\linewidth}{\begin{minipage}{\linewidth}\noindent
				\input{TEX/d-types2}
				\hrulefill
				\input{TEX/d-terms2}
	\end{minipage}}}
	\caption{The definitions of AD on recursive types. \label{fig:ad2}}
\end{figure}
\subsection{AD transformation as an $rCBV$ model morphism}
By Proposition \ref{prop:universal-property-type-recursive-language}, the assignment defined in Figure~\ref{fig:assignment-AD-functor} induces a unique $rCBV$ model morphism \eqref{eq:macro-as-a-functor-recursive-types}, which \textit{encompasses  the macro $\Dsynsymbol $ defined  by Figure \ref{fig:ad1} and extended in Figure~\ref{fig:ad2}.}
\begin{equation}\label{eq:macro-as-a-functor-recursive-types}
	\DSynrec : \left( \SynVr, \SynTr, \SynffixpointRecT   \right) \to \left(\SynVrt{,}\SynTrt{,}\SyntffixpointRecT\right).
\end{equation}

\subsection{$\wCpo$-enriched categorical models for recursive types: $rCBV$ $\wCpo$-pairs}
\label{ssec:wcpo-enriched-cat-models-rec-types}
Although the setting of \textit{bilimit compact expansions}
 is the usual reasonable basic framework for solving recursive domain equations, we do not need this level of generality. Instead, we consider a subclass of $\wCpo$-enriched models, the  $rCBV$ $\wCpo$-pairs established in  Definition~\ref{def:concrete-rCBV-models}.\footnote{See \cite[4.2.2]{levy2012call} or \cite[Sect.~8]{vakar2020denotational} for the general setting of bilimit compact expansions.}

 We are back again to the setting of $\wCpo$-enriched categories. 
 Recall that
an \textit{embedding-projection-pair (ep-pair)}  $u : A\epto B$
in an $\wCpo$-category $\catC$ 
 is a pair $u = \eppair u$ consisting of a $\catC $-morphism
 $\emb u:A\to B$, the \emph{embedding}, and a
$\catC$-morphism $\prj u:B\to A$, the \emph{projection}, such that
$\emb u \circ \prj u \leq \id$ and $\prj u\circ \emb u= \id$.

It should be noted that, when considering the underlying $2$-category of the $\wCpo $-category,  an ep-pair consists of an adjunction\footnote{See, for instance, \cite[Sect.~2]{MR0357542} or \cite[3.10]{2019arXiv190201225L} for adjunctions in $2$-categories.} whose unit is the identity. In this context, it is also called a  lari adjunction (\textit{left adjoint right-inverse}), see \cite[Sect.~1]{2020arXiv200203132C}. In particular, as in the case of any adjunction, an embedding $\emb u : A \to B$ uniquely determines the associated projection $\prj u: B\to A$ and vice-versa.

A zero object\footnote{Recall that a \emph{zero object} is an object that is both initial and
terminal.} $\epzeroo$ in an $\wCpo$-category $\catC $ is an \emph{ep-zero object}  if, for any object $A$, the pair $\uniqMor _A = \left( \emb \uniqMor : \epzeroo\to A , \prj \uniqMor : A\to\epzeroo \right)  $ consisting of the unique morphisms is an ep-pair.

\begin{definition}[$rCBV$ $\wCpo$-pair]\label{def:concrete-rCBV-models}
An 	$rCBV$ $\wCpo$-pair is a $CBV$ pair $\left( \catV , \monadT \right) $ such that, denoting by $J : \catV\to\catC  $ the corresponding universal Kleisli $\catV$-functor,
\begin{enumerate}[r$\omega$.1]
	\item $\catV$ is a cocomplete $ \wCpo $-cartesian closed category\footnote{Because $\catV$ is cartesian closed, any colimit in $\catV$ is a conical $\catV$-colimit \cite{kelly1982basic}.
	Because $\catV$ is $\wCpo$-cartesian closed, any conical $\catV$-colimit in $\catV$ is, in particular, a conical $\wCpo$-colimit.} 
	;\label{condition-for-rCBVwCpo-cocomplete}
	\item the unit of $\monadT$ is pointwise a full morphism (hence, $J$ is a locally full $\wCpo$-functor); \label{condition-for-rCBVwCpo-unit}
	\item  $\catC $ has an ep-zero object $\epzeroo = J\left( \initiall\right) $, where 
	$\initiall$ is initial in $\catV$;	\label{condition-for-rCBVwCpo-ep-zero}
	\item whenever $u : J(A)\epto J(B)$ is an ep-pair in $\catC $, there is one morphism 
	$\unlifteppair{u} : A\to B $ in $\catV$ such that $J\left( \unlifteppair{u}\right) =\emb u  $. \label{eq:pulling-embeddings-values}
\end{enumerate}

An  \textit{$rCBV$ $\wCpo$-pair morphism} from $\left(\catV , \monadT \right) $ into $\left(\catV  ', \monadT ' \right) $ is an $\wCpo$-functor 
$H : \catV\to\catV ' $ that strictly preserves $\wCpo$-colimits, and whose underlying functor is a morphism between the $CBV$ pairs.
This defines a category of $rCBV$ $\wCpo$-pairs, denoted herein by $\wrCBVcat$.
\end{definition}

Every $rCBV$ $\wCpo$-pair $\left( \catV , \monadT \right) $ has an underlying $\wCpo $-pair,
and this extends to a forgetful functor $\wrCBVcat\to\wCBVcat $. More importantly to our work, we have the following.

\subsubsection{$rCBV$ $\wCpo$-pairs are $rCBV$ models} 
Let $\left( \catV , \monadT \right) $ be an $rCBV$ $\wCpo$-pair. 
It is clear that we have an underlying $CBV$ pair which, by abuse of language, we denote by $\left( \catV , \monadT \right) $  as well. Hence, 
we can consider $\left(\catV , \monadT \right) $-parametric types.

Let $n\in\NN^\ast $ and \eqref{eq:parametric-types-in-terms-of-functors} be an $n$-variable  $\left(\catV , \monadT \right) $-parametric type. For each $A\in \left( \catV ^\op \times \catV\right) ^{n-1} $, we get an $1$-variable $\left(\catV , \monadT \right) $-parametric type $\pEE{E}^A = \left( \paE{E}{A}{\catV}, \paE{E}{A}{\catC} \right) $ where $\paE{E}{A}{\catV}\left( W,Y \right) \defeqq \pE{E}{\catV}\left( A, W,Y \right) $ and $\paE{E}{A}{\catC}\left( W',Y' \right) \defeqq \pE{E}{\catC}\left( J(A), W',Y' \right) $. Let $\pEDchain{E}{A} $ be the diagram  \eqref{eq:diagram-for-the-colimit-that-defines-the-fddt} in $\catC$ given by the chain of morphisms $\left( \mochain{n} : \objectD{A}{n}\to\objectD{A}{n+1}  \right) _{n\in\NN} $,   where $\left( \epchain{n}\right) _{n\in\NN} $ is the chain of ep-pairs inductively defined by \eqref{eq:inductive-definition-of-ep-pairs}. \\
\noindent\begin{minipage}{.5\linewidth}
\small	
	\begin{eqnarray}
	\epchain{0} &\defeqq & \left( \emb \uniqMor : \epzeroo\to  \paE{E}{A}{\catC}\left( \epzeroo, \epzeroo \right)  , \prj \uniqMor : \paE{E}{A}{\catC}\left( \epzeroo, \epzeroo \right) \to\epzeroo \right)  \nonumber\\
	\epchain{n+1} & \defeqq & \left( \paE{E}{A}{\catC}\left( \comochain{n} , \mochain{n} \right), \paE{E}{A}{\catC}\left( \mochain{n} , \comochain{n} \right)   \right)  \label{eq:inductive-definition-of-ep-pairs}
\end{eqnarray} 
\normalsize
\end{minipage}\noindent
\begin{minipage}{.5\linewidth}
	\small
		\begin{eqnarray}
		\diag{diag-chain-of-ep-pairs} \label{eq:diagram-for-the-colimit-that-defines-the-fddt}\\
		\diag{diag-chain-of-ep-pairs-projections}\label{eq:projections-diagram-for-the-colimit-that-defines-the-fddt}
	\end{eqnarray} 	
\normalsize
\end{minipage}\\
\normalsize
There is a unique diagram \textit{ $\unlifteppair{\pEDchain{E}{A}}$ such that $J\circ\unlifteppair{\pEDchain{E}{A}} = \pEDchain{E}{A}$} by \eqref{eq:pulling-embeddings-values} of Definition~\ref{def:concrete-rCBV-models}. Since $\catV $ has $\wCpo $-colimits, we conclude that the conical $\wCpo$-colimit of $\unlifteppair{\pEDchain{E}{A}}$ exists and is preserved by $J$ (being an $\wCpo$-left adjoint) -- hence, $\pEDchain{E}{A}$ has a conical $\wCpo $-colimit in $\catC $ as well. 

We recall the following variation on \cite{smyth-plotkin:rde}'s celebrated \emph{limit-colimit	coincidence} result.
\begin{lemma}[Limit-colimit coincidence, \`a la \cite{smyth-plotkin:rde}]
For any $\omega$-chain $(a^e_n\dashv a^p_n)_{n\in \NN}$ of ep-pairs  in an $\wCpo$-category $\catC$, any $\wCpo$-colimiting cocone on $(a^e_n)_{n\in \NN}$
consists of embeddings and the corresponding projections form an 
$\wCpo$-limiting cone on $(a^p_n)_{n\in \NN}$.
\end{lemma}
Since \eqref{eq:diagram-for-the-colimit-that-defines-the-fddt} is the chain of embeddings  of a chain of ep-pairs, the $\wCpo$-colimit of these embeddings coincides with the $\wCpo$-limit of the associated chain $\left( \comochain{n} \right) _{n\in\NN } $ of projections \eqref{eq:projections-diagram-for-the-colimit-that-defines-the-fddt}, denoted herein by $\prjEDchain{E}{A}$. Such a \textit{bilimit} of ep-pairs is absolute in the sense that any $\wCpo$-functor $H : \catC \to \catC '$ preserves the conical $\wCpo$-colimit (and $\wCpo$-limit) of $\pEDchain{E}{A} $ (respectively, $\prjEDchain{E}{A}$).

Since the conical $\wCpo$-colimit of $\pEDchain{E}{A} $ is absolute, the diagram \eqref{eq:parametric-types-in-terms-of-functors} commutes, and $J$ strictly preserves $\wCpo$-colimits, we have the invertible morphism \eqref{eq:isomorphism-giving-the-roll-for-rCBV-wCPO-pairs} given by the composition of the respective canonical comparison morphisms. 

\footnotesize
\begin{equation}\label{eq:isomorphism-giving-the-roll-for-rCBV-wCPO-pairs}
\diag{rolling-wCPO-pairs} 
\end{equation} 
\normalsize 
 
It should be noted that, for each $f: \left( J^\op\times J\right) ^{n-1} (A)\to \left( J^\op\times J\right)^{n-1} (B) $ in $\left( \catC ^\op\times \catC\right) ^{n-1}  $, we have an induced $\catV $-natural transformation $\pEDchain{E}{f} : \pEDchain{E}{A}\to \pEDchain{E}{B} $. This association extends to a \textit{$\catV$-functor  $\pEDchainn{E}$ from $\left( \catC ^\op\times \catC\right) ^{n-1}   $ into the $\catV$-category of chains in $\catC $}. The association $A\mapsto\unlifteppair{\pEDchain{E}{A}}  $ also extends to a \textit{$\catV $-functor $\unlifteppair{\pEDchainn{E}}$  from $\left( \catV ^\op\times \catV\right) ^{n-1}  $ into the $\catV $-category of chains} by the $\catV$-faithfulness of $J$, .

\textit{We define the $\fddt $ operator $\wfixpointRecT $ as follows}. For each $n\in\NN ^\ast $, given a $\left(\catV , \monadT \right) $-parametric type $\pEE{E} = \left( \pE{E}{\catV},  \pE{E}{\catC}\right) $, we define:
\begin{equation}
	\wfixpointRecT \pEE{E} = \left(  \wfixpointRecT \pE{E}{\catV} , \wfixpointRecT \pE{E}{\catC} \right) \defeqq  \left(  \colim\circ\unlifteppair{\pEDchainn{E}} , \colim\circ\pEDchainn{E} \right)
\end{equation}
where, by abuse of language, $\colim $ is the $\catV $-functor from the $\catV$-category of chains in $\catV$ (respectively, in $\catC $) into the $\catV$-category $\catV$ (respectively, $\catC$).

Since every isomorphism is an embedding, there is only one \textit{$\pwrollReT{E}_A $ in $\catV $ such that $J\left( \pwrollReT{E}_A\right) $ is equal to \eqref{eq:isomorphism-giving-the-roll-for-rCBV-wCPO-pairs}}.
The morphisms $\wrollReT ^E = \left( \pwrollReT{E} _A\right) _{A\in \left( \catV ^\op \times\catV\right) ^{n-1} } $ gives a $\catV$-natural transformation $\pE{E}{\catV}\left( \id, \wfixpointRecT\pE{E}{\catV} ^\op , \wfixpointRecT\pE{E}{\catV}  \right)\rightarrow \wfixpointRecT\pE{E}{\catV}  $  such that $J\left(\wrollReT ^E  \right) $ is invertible. Therefore $\wrolling \defeqq \left( \wrollReT ^E  \right) _{E\in\ParamAll{\catV}{\monadT} } $ is a rolling for $\wfixpointRecT $ and we can define the (free) type recursion $\wffixpointRecT\defeqq \left( \wfixpointRecT , \wrolling\right) $.

\begin{lemma}[Underlying $rCBV$ model]\label{lem:UnderlyingrCBVmodel}
	There is a forgetful functor $\rCBVU : \wrCBVcat\to \RCBVcat $ defined by $\rCBVU\left( \catV , \monadT \right) = \left( \catV , \monadT , \wffixpointRecT\right)  $, that takes every morphism $H$ to its underlying morphism of $CBV$ models.
\end{lemma}	
\begin{proof}
	 From the definition of $\wffixpointRecT$ and the fact that $H$ strictly preserves $\catV$-colimits, we conclude that, indeed, $H$ respects the conditions of a $rCBV$ model morphism described in Definition~\ref{def:rCBV-models}.
\end{proof}

\begin{remark}\label{rem:products-of-rCBV-wCPO-pairs} 
The product of $rCBV$ $\wCpo$-pairs is computed as expected: $\left( \catV _0 , \monadT _0 \right) \times \left( \catV _1 , \monadT _1 \right) \cong \linebreak\left( \catV _0 \times\catV _1 , \monadT _0\times \monadT _1 \right) $. Moreover, it is clear that $\rCBVU$ preserves finite products.
\end{remark}

\subsection{Concrete semantics}\label{subsect:concrete-semantics-for-the-recursive-types}
The $CBV$ pair $\left( \wCpo , \monadwP{-}\right) $ as in Section \ref{subsec:wcPO-Basic-Model} clearly satisfies the conditions of Definition~\ref{def:concrete-rCBV-models} and, hence, it is also an $rCBV$ $\wCpo $-pair. By Proposition \ref{prop:universal-property-type-recursive-language}, for each $k\in\NN \cup\left\{ \infty \right\} $, we have unique $rCBV$ model morphisms \eqref{eq:semantics-for-the-source-recursive-types} and 
\eqref{eq:semantics-for-the-target-recursive-types} respecting the assignments of Figure~\ref{fig:assignment-semantics} and \eqref{eq:k-semantics-for-the-target-language}. In other words, following Remark \ref{rem:extending-functors}, we have only one extension of the semantics \eqref{eq:functor-semantics}  and \eqref{eq:semantics-of-the-target-as-a-functor} to the respective languages with recursive types.
\\ 
\footnotesize
\noindent\begin{minipage}{.5\linewidth}	
	\begin{equation}\label{eq:semantics-for-the-source-recursive-types}
	\sem{-} : \left( \SynVr, \SynTr, \SynffixpointRecT   \right)  \to \rCBVU\left( \wCpo , \monadwP{-}\right)  
	\end{equation} 
	\normalsize
\end{minipage}%
\begin{minipage}{.5\linewidth}
	\begin{equation}\label{eq:semantics-for-the-target-recursive-types}
	\semt{k}{-}: \left( \SynVrt , \SynTrt ,\SyntffixpointRecT  \right)  \to \rCBVU\left( \wCpo , \monadwP{-}\right) .
	\end{equation} 	
\end{minipage}\\
\normalsize

Moreover, by Remark \ref{rem:products-of-rCBV-wCPO-pairs}, we have that the product
$ \left( \wCpo\times \wCpo , \monadwP{-}  \right) $ as in Section \ref{subsec:wcPO-Basic-Model} is an $rCBV$ $\wCpo$-pair.

\subsection{Subscone for $rCBV$ $\wCpo$-pairs}\label{ssec:rec-types-sscone}
The first step for our logical relations proof is to verify that, for each $(n,k)\in\NN\times \left( \NN \cup \left\{ \infty \right\}\right) $, the $CBV$  $\wCpo$-pair $\left(\SUBscone{\wCpo }{\sconeFUNCTOR{n, k}}, \monadLR{n, k}{-} \right) $ as in Proposition \ref{prop:the-subscone-with-the-monads-yields-a-CBV-wCPO-pair} yields an $rCBV$ $\wCpo$-pair. In order to do that, we rely on Theorem \ref{theo:theorem-descent-rCBV-wCpo-structure-structure} about lifting the $rCBV$ $\wCpo$-pair structure.
\begin{definition}[Impurity preserving/purity reflecting]
	Let $\left( \catV , \monadT \right) $ and $\left( \catV ', \monadT ' \right) $ be $CBV$ pairs. A $CBV$ pair morphism $H: \catV \to \catV '$ is \textit{impurity preserving} (or, \emph{purity reflecting}) if, whenever $H(f) = \ee ' _Y \circ g $, there is $\unlifteppair{f}$ in $\catV$ such that $\ee  _Y \circ \unlifteppair{f}  = f$. 
\end{definition} 	

\begin{theorem}\label{theo:theorem-descent-rCBV-wCpo-structure-structure}
	Let $\left( \catV ', \monadT ' \right) $ be an $rCBV$ $\wCpo$-pair, and	
	 $\left( \catV , \monadT \right) $ a $CBV$ pair such that $\catV $ is a cocomplete $\wCpo$-cartesian closed category and $T(\initiall) $ is terminal.
	
If $H : \catV \to \catV '$ is a locally full $\wCpo$-functor that yields an impurity preserving $CBV$ pair morphism $\left( \catV , \monadT \right)\to \CBVUrp\left( \catV ', \monadT ' \right)$,
then $\left( \catV , \monadT  \right) $ is an $rCBV$ $\wCpo$-pair. If, furthermore, 
$H$ strictly preserves $\wCpo$-colimits, then $H$ yields an $rCBV$ $\wCpo$-pair morphism. 
\end{theorem} 	
\begin{proof}
	We prove that $\left( \catV , \monadT \right) $ yields an
	$rCBV$ $\wCpo$-pair. By hypothesis, $\left( \catV , \monadT \right) $  satisfies \eqref{condition-for-rCBVwCpo-cocomplete}. We prove the remaining conditions of Definition~\ref{def:concrete-rCBV-models} below.
\begin{enumerate}
	\item[\eqref{condition-for-rCBVwCpo-unit}] Let $\ee $ and $\ee '$ be respectively the unit of $\monadT $ and $\monadT '$. 
	 Since $H$ is locally full, it reflects full morphisms. This implies that, for any $C\in\catV$,  $\ee _ C $ is full since $\ee '_{H(C)}= H\left( \ee _ C \right) $ is full.  
	 \item[\eqref{condition-for-rCBVwCpo-ep-zero}] Since $T(\initiall)$ is terminal, $J\left(\initiall \right) $ is a zero object. Thus, for each $A\in\catC $,
	 we have the pair  \eqref{eq:unique-morphisms-computations}  of unique morphisms in $\catC$.\\
	 Since $\lift{H}$ preserves initial objects and $\left( \catV ', \monadT ' \right)$ is an $rCBV$ $\wCpo$-pair, we have that \eqref{eq:unique-morphisms-computations-image-by-H} is the ep-pair of the unique morphisms. Finally, since $\lift{H}$ is a locally full $\wCpo$-functor, it reflects ep-pairs and, hence, \eqref{eq:unique-morphisms-computations}  is an ep-pair.
\\ 
\noindent\begin{minipage}{.5\linewidth}	
	\begin{equation}\label{eq:unique-morphisms-computations} 
		\left( \uniqMor _A : J\left(\initiall \right) \to  A  ,  {\uniqMor ^A}  :  A \to J\left(\initiall \right)  \right) 
	\end{equation} 
	\normalsize
\end{minipage}%
\begin{minipage}{.5\linewidth}
	\begin{equation}\label{eq:unique-morphisms-computations-image-by-H} 
		\left( \lift{H}\left( \uniqMor _A\right)   ,  \lift{H}\left({\uniqMor ^A}\right)   :  \lift{H}\left( A\right) \to \epzeroo \right) 
	\end{equation} 
\end{minipage} 
 \item[\eqref{eq:pulling-embeddings-values}]  
  Given an ep-pair $u : J(A)\epto J(B)$ in $\catC $, the image $H(u) :\lift{H}J(A)\epto \lift{H}J(B) $  by $H$ is an ep-pair. Since $\left( \catV ', \monadT '\right) $ is an $rCBV$ $\wCpo$-pair, there is one morphism 
 $\unlifteppair{\lift{H}\left( u\right)} : H(A)\to H(B)  $ in $\catV '$ such that $J'\left( \unlifteppair{\lift{H}\left( u\right)}  \right) = \lift{H}\left( \emb u \right)   $.
 Since the $CBV$ pair morphism $H: \left( \catV , \monadT \right)\to \CBVUrp\left( \catV ', \monadT ' \right)$ is impurity preserving, we conclude that there is $\unlifteppair{ u } : A\to B$ such that $J\left( \unlifteppair{ u } \right) = \emb u $.  
\end{enumerate}
\end{proof}

As a consequence, in the setting of subscones satisfying Assumption~\ref{assum:subscone-assumptions}, we get:

\begin{theorem}\label{the:maybe-the-main-result-on-LR-recursive}
	Let $\left( \catV , \monadT \right) $ be an $rCBV$ $\wCpo$-pair, and  \eqref{eq:forgetful-subscone-rCBV} the forgetful $\wCpo$-functor coming 
	from a pair $\left( G: \catV  \to\catD , \monadSub\right) $ satisfying Assumption~\ref{assum:subscone-assumptions}. 
	
	 If $\catD$ is cocomplete and $\LRMONAD = \left( \lift{T}, \lift{\mm}, \lift{\ee} \right)   $ is a strong monad that is a lifting of the monad $\monadT $ along \eqref{eq:forgetful-subscone-rCBV}  such that \eqref{eq:trivial-condition-monad} and \eqref{eq:impurity-preserving-in-terms-of-the-pullback-unit} hold, then $\left( \SUBscone{\catD}{G}, \LRMONAD  \right) $ is an $rCBV$ $\wCpo$-pair and $\forgetfulSub$ yields an $rCBV$ $\wCpo$-pair morphism \eqref{eq:the-CBV-pair-morphicm-sscone}.
\begin{enumerate}[$\mathfrak{c}$.1] 
	 \item  $\lift{T}$ takes the initial to the terminal object;\label{eq:trivial-condition-monad}
	 \item  for any $\left( D, C, j\right)\in \SUBscone{\catD}{G} $, 
	 denoting  $\lift{T}\left( D, C, j\right) = \left( \unlift{\LRMONAD\left(D,C,j\right)}, T(D), \unlift{\LRMONAD{j}}\right) $, Diag~\eqref{eq:diagram-induced-unit} induced by the unit $\lift{\ee}$ is a pullback in $\catD $.\label{eq:impurity-preserving-in-terms-of-the-pullback-unit}
\end{enumerate}	 
\end{theorem} 
\noindent\begin{minipage}{.5\linewidth}	
\begin{equation}\label{eq:forgetful-subscone-rCBV}
	\forgetfulSub : \SUBscone{\catD}{G}\to \catV 
\end{equation}	
\begin{equation}\label{eq:the-CBV-pair-morphicm-sscone}
	\left( \SUBscone{\catD}{G}, \LRMONAD  \right)\to \left( \catV , \monadT \right) 
\end{equation}
	\normalsize
\end{minipage}%
\begin{minipage}{.5\linewidth}
\begin{equation}\label{eq:diagram-induced-unit}
	\diag{unit-pullback}
\end{equation}
\end{minipage}

\begin{proof}
	By Corollary \ref{coro:limits-and-colimits-subscone}, $\SUBscone{\catD}{G}$ is cocomplete $\wCpo$-cartesian closed. Moreover,  $\forgetfulSub$ is locally full, strict $\wCpo$-cartesian closed, and $\wCpo$-colimit preserving by Theorem \ref{theo:properties-forgetful-subscone}.
	Therefore, the fact that $\LRMONAD$ is a lifting of $\monadT $ through $\forgetfulSub$ implies
	that it yields a $CBV$ pair morphism \eqref{eq:the-CBV-pair-morphicm-sscone}.
	
\eqref{eq:impurity-preserving-in-terms-of-the-pullback-unit} implies that the $CBV$ pair morphism \eqref{eq:the-CBV-pair-morphicm-sscone} is purity reflecting. Assuming  \eqref{eq:trivial-condition-monad}, this implies that $\left( \SUBscone{\catD}{G}, \LRMONAD  \right)$ is indeed an $rCBV$ $\wCpo$-pair morphism and $\forgetfulSub$ yields an \eqref{eq:the-CBV-pair-morphicm-sscone} is an $rCBV$ $\wCpo$-pair morphism by Theorem \ref{theo:theorem-descent-rCBV-wCpo-structure-structure}.
\end{proof}

In the particular case of interest, we conclude:

\begin{proposition}\label{prop:wCpo-pair-of-Artin-Gluing}
	For each $(n,k)\in\NN\times \left( \NN \cup \left\{ \infty \right\}\right) $, $\left(\SUBscone{\wCpo }{\sconeFUNCTOR{n, k}}, \monadLR{n, k}{-} \right) $ is an $rCBV$ $\wCpo $-pair. Moreover,
	$\forgetfulSub _{n,k} : \SUBscone{\wCpo }{\sconeFUNCTOR{n,k} }\to \wCpo\times \wCpo$
	yields an $rCBV $ $\wCpo $-pair morphism 
	\begin{equation} \label{eq:rCBV-pair-morphism-Lnk}
	 \left(\SUBscone{\wCpo }{\sconeFUNCTOR{n, k}}, \monadLR{n, k}{-} \right) \to \left(\wCpo\times \wCpo, \monadwP{-}\right) .
	 \end{equation}  
\end{proposition} 
\begin{proof}
In fact, we already know that $\forgetfulSub _{n,k} $ comes from a pair that satisfies Assumption~\ref{assum:subscone-assumptions}. Moreover, $\left(\wCpo\times \wCpo, \monadwP{-}\right)$ is an $rCBV$ $\wCpo$-pair and $\monadLR{n, k}{-}$ is a lifting of $\monadwP{-}$ along $\forgetfulSub _{n,k} $ satisfying the conditions of Theorem \ref{the:maybe-the-main-result-on-LR-recursive}.

\end{proof}

By Proposition \ref{prop:wCpo-pair-of-Artin-Gluing} and Lemma \ref{lem:UnderlyingrCBVmodel}, we get:
\begin{corollary}
	$\forgetfulSub _{n,k}$ yields an $rCBV$ model morphism  $$\rCBVU \left(\SUBscone{\wCpo }{\sconeFUNCTOR{n, k}}, \monadLR{n, k}{-} \right)\to\rCBVU\left(\wCpo\times \wCpo, \monadwP{-}\right) . $$
\end{corollary}

\subsection{Logical relations as an $rCBV$ model morphism}
Let $(n, k)\in\NN\times\left( \NN \cup\left\{ \infty \right\}\right)  $,
and let's assume that $\Dsynsymbol $ is sound for primitives (see Definition \ref{def:sound-for-primitives}). By the universal property of the $rCBV$ model $\left( \SynVr, \SynTr, \SynffixpointRecT   \right) $ and the chain rule for derivatives, 
 there is only one $rCBV$ model morphism 	
\begin{equation}\label{eq:extending-logical-relations}
	\semLR{n,k}{-}:\left( \SynVr, \SynTr, \SynffixpointRecT   \right) \to \rCBVU\left( \SUBscone{\wCpo }{\sconeFUNCTOR{n, k}},  \monadLR{n, k}{-}    \right) 
\end{equation}
that is consistent with the assignment given by \eqref{assig-object-LR}, \eqref{LR:assig-tSign},
\eqref{LR:assig-op},  and \eqref{LR:assig-cnst}.

\begin{lemma}\label{lem:fundamental-commutativity-diagram-for-LR-argument-recursive-types}
	For any $(n, k)\in\NN\times \left( \NN \cup\left\{ \infty\right\}\right)$, Diag.~\eqref{eq:basic-commutative-diagram-for-the-recursive-types-LR-argument} commutes. 
\begin{equation}\label{eq:basic-commutative-diagram-for-the-recursive-types-LR-argument}
	\diag{basic-logicalrelations-recursive-types}
\end{equation} 
\end{lemma} 
\begin{proof}
	Both $\left(\sem{-}\times\semt{k}{-}\right)\circ\left( \ID\times\DSynrec\right) $ and $\rCBVU\left({\forgetfulSub}_{n,k}\right)\circ\semLR{n,k}{-} $
	yield  $rCBV$ model morphisms that are consistent with the assignment given by the object $\left( \RR , \RR\times\RR ^k  \right) $ and the morphisms \eqref{defeq:LR-sign}, \eqref{defeq:LR-cnst} and \eqref{defeq:LR-op}. Therefore, by the universal property of $\left(\SynVr,\SynTr{,}\SynffixpointRecT\right) $, we conclude that  Diag.~\eqref{eq:basic-commutative-diagram-for-the-recursive-types-LR-argument} indeed
	commutes.
\end{proof}


\subsection{AD correctness theorem for non-recursive data types}
The correctness theorem for non-recursive data types (i.e., types formed from $\reals$, products, and coproducts) follows from Lemma \ref{lem:fundamental-commutativity-diagram-for-LR-argument-recursive-types} and Corollary \ref{coro:fundamental-LR-conclusion-about-morphisms-subscone}. That is to say, we have:

\begin{theorem}\label{theo:basal-version-of the-correctness-theorem-recursive-types}
	 	Let $\displaystyle t: \coprod _{r\in \KI} \reals ^{s _r} \to \SynTr\left( \coprod _{j\in L} \reals ^{l _j}  \right)  $  be a morphism in $\SynVr$. We have that 
	 	$\displaystyle\sem{ t } : \coprod _{r\in \KI} \RR  ^{s _r} \to \monadwP{\coprod _{j\in L} \RR ^{l _j} } $ is differentiable and, for any $k\in \left( \NN\cup\left\{ \infty \right\}\right) $, $	\semt{k}{\DSynrec\left( t \right) } = \dDSemtra{k}{\sem{ t }}  $.
\end{theorem}

\subsection{AD on recursive data types}
The logical relations argument we presented provides us with an easy way to compute the logical relations of general recursive types: namely, since $\left(\SUBscone{\wCpo }{\sconeFUNCTOR{n, k}}, \monadLR{n, k}{-} \right)$ is an $rCBV$ $\wCpo $-pair,
the recursive types will be computed out of suitable colimits. This gives us useful information about the semantics of $\Dsyn{\trm{1} }$  for a program $\var{1}:\ty{1}\vdash \trm{1}:\ty{2}$  where $\ty{1}$ and $\ty{2}$ are recursive types. In particular, we can extend the correctness result of Theorem \ref{theo:basal-version-of the-correctness-theorem-recursive-types} to any 
\emph{recursive data type}. 
By that, we mean any type $\ty{1}$ built from the grammar $$\ty{1},\ty{2}::= \tvar{1}\mid \reals\mid \Init \mid \Unit\mid \ty{1}\times \ty{2}\mid \ty{1}\sqcup \ty{2}\mid \mu\tvar{1}.\ty{1},$$ i.e., any type not involving function types.

We can define these (recursive) data type more formally as follows.
We denote by $\SynCr$ the Kleisli $\SynVr$-category associated with 
$\left( \SynVr, \SynTr  \right)$. Moreover, we respectively denote by \eqref{eq:basic-coproduct-functor-for-the-definition-of-data-type} and \eqref{eq:step-diagonal-inductive} the coproduct, product and $n$-diagonal functors.\\
\noindent
\begin{minipage}{.4\linewidth}	
	\begin{equation}
	\sqcup, \times : \SynVr\times  \SynVr \to\SynVr\label{eq:basic-coproduct-functor-for-the-definition-of-data-type}
\end{equation} 
\end{minipage}\noindent
\noindent
\begin{minipage}{.6\linewidth}	
\begin{equation}\label{eq:step-diagonal-inductive}
	\diagk{n} :  \left(\SynVr\right) ^\op\times \SynVr  \to \left(\left(\SynVr\right) ^\op\times \SynVr\right) ^n
\end{equation} 
\end{minipage}\normalsize
\begin{definition} 
Let $\pEE{R}, \pEE{I}, \pEE{O} : \left(\SynVr\right) ^\op\times \SynVr \to\SynVr $ be the constant functors which are, respectively, equal to $\reals $, $\terminall$ and $\initiall$.
\textit{We define the set $\ParamVd{\SynVr , \SynTr, \ffixpointRecT_\Syn}$ inductively by \ref{data-type-first}, \ref{data-type-second} and \ref{data-type-third}. }
\begin{enumerate}[(D1)]
	\item \label{data-type-first} The functors $\pEE{R}, \pEE{I}, \pEE{O}$ are in $\ParamVd{\SynVr , \SynTr, \ffixpointRecT_\Syn}$. Moreover, the projection $\proj{2}:\left(\SynVr\right) ^\op\times \SynVr\to \SynVr  $ belongs to $\ParamVd{\SynVr , \SynTr, \ffixpointRecT_\Syn}$.
	\item \label{data-type-second} For each $n\in \NN ^\ast $, if the functors \eqref{eq:basic-functor-for-assumption-inductive-definition-data-types} belong to $\ParamVd{\SynVr , \SynTr, \ffixpointRecT_\Syn}$, then the functors \eqref{eq:composition-of-Ev-diagonal} and \eqref{eq:step-doproduct-product-inductive}  are in $\ParamVd{\SynVr , \SynTr, \ffixpointRecT_\Syn}$.
	\item \label{data-type-third} If $\pEE{E} = \left(\pE{E}{\SynVr}, \pE{E}{\SynCr} \right)\in\ParamAll{\SynVr}{\SynTr}$ is such that $\pE{E}{\SynVr}\in\ParamVd{\SynVr , \SynTr, \ffixpointRecT_\Syn} $, then $ \left( \pE{\fddtpointRecT{E}}{\SynVr } \right)$ is in
	$\ParamVd{\SynVr , \SynTr, \ffixpointRecT_\Syn}$. 
\end{enumerate}	
\textit{	We define the set  $\Paramd{\SynVr}{\SynTr ,\ffixpointRecT_\Syn }$ of parametric data types by \eqref{eq:definition-of-parametric-data-types}.}
\end{definition} 
\noindent
\begin{minipage}{.5\linewidth}	
	\begin{equation}
	\pEE{G}, \pEE{G'} : \left(\left(\SynVr\right) ^\op\times \SynVr\right) ^n  \to\SynVr\label{eq:basic-functor-for-assumption-inductive-definition-data-types}
\end{equation} 
\end{minipage}\noindent
\begin{minipage}{.5\linewidth}	
	\begin{equation}\label{eq:composition-of-Ev-diagonal}
	\pEE{G}\circ \diagk{n} : \left(\SynVr\right) ^\op\times  \SynVr \to\SynVr
\end{equation} 	
\end{minipage}\\
\normalsize
\begin{equation}\label{eq:step-doproduct-product-inductive}
	\times \circ \left( G\times G'\right),  \sqcup \circ \left( G\times G'\right) : \left(\left(\SynVr\right) ^\op\times \SynVr\right) ^{2n}  \to\SynVr
\end{equation}
\begin{equation}\label{eq:definition-of-parametric-data-types}
\Paramd{\SynVr}{\SynTr ,\ffixpointRecT_\Syn }\coloneq \left\{ \pEE{E}  \in \ParamAll{\SynVr}{\SynTr} : \pE{E}{\SynVr}\in\ParamVd{\SynVr , \SynTr, \ffixpointRecT_\Syn}  \right\}
\end{equation}

All such (recursive) data types are, up to isomorphism, of a particularly simple form: a sum of products.
\begin{proposition}\label{prop:image-of-recursive-data-types}
	Let $E$ be an $n$-variable $\left( \SynVr, \SynTr, \SynffixpointRecT   \right)$-parametric data type, where $n\in\NN ^\ast $. There is a countable family of natural numbers $\displaystyle\left( \mind{j, \mathtt{T}} \right) _{\left( j, \mathtt{T}\right)\in \left( \NNN{n}\cup \left\{0\right\} \right)\times {\TreeIndex}} $ such that, for any $rCBV$ model morphism $H : \left( \SynVr, \SynTr, \SynffixpointRecT   \right) \to \rCBVU\left( \catV , \monadT\right)   $  and any 
	$H$-compatible pair $\left( E,F\right)$, we have that  \eqref{eq:standard-format-parametric-type} holds, where the isomorphism $\cong$ is induced by coprojections and projections\footnote{That is to say, it is just a reorganization of the involved coproducts and products.}. 
\noindent
\begin{minipage}{.35\linewidth}	
	\begin{equation}\label{eq:image-of-R-by-an-arbitrary-rCBV-morphism}
		H\left(\uptau\right)  =  \coprod _{j\in L}  H\left( \reals\right)  ^{l_j}
	\end{equation} 
	\normalsize
\end{minipage}
\noindent\begin{minipage}{.65\linewidth}
	\begin{equation}\label{eq:standard-format-parametric-type} 
		\pE{F}{\catV }\left( W_j, Y_j\right) _{j\in\NNN{n}} \cong\coprod _{\mathtt{T}\in {\TreeIndex} }\left(  {H\left( \reals\right) }^{\mind{0,\mathtt{T}}}\times \prod _{j=1}^{n} Y_j ^{\mind{j,\mathtt{T}} }\right) 
	\end{equation} 
\end{minipage} 
 As a consequence, if $\uptau\in\SynVr $ corresponds to a data type $\ty{1}$, then there is a countable family $\left( l_j \right) _{j\in L}\in \NN ^L $ of natural numbers  such that
		 \eqref{eq:image-of-R-by-an-arbitrary-rCBV-morphism} holds for any $rCBV$ model morphism $H : \left( \SynVr, \SynTr, \SynffixpointRecT   \right) \to \rCBVU\left( \catV , \monadT\right)   $.
\end{proposition} 	

\begin{proof}

The result follows from induction.
The non-trivial part is a consequence of the following.

Let 
$\left( \pEE{\uE}, \pEE{\uF}\right) \in \Paramd{\SynVr}{\SynTr}\times \ParamAlll{ \left(\rCBVU\left(  \catV , \monadT \right) \right) }$
be an $H$-compatible pair 
of $\left( n+1\right)$-variable parametric types where $\pE{\uF}{ \catV }   $ is given by \eqref{eq:standard-format-parametric-type-proof} for some countable family $\displaystyle\left( \sind{i,r} \right) _{{\left(i,r\right)\in \left( \NNN{n+1}\cup \left\{0\right\} \right)\times {\KI}} }$ of natural numbers. We prove below that $\left( \pEE{\fddtpointRecT{\uE}},\pEE{F}\right)  $ is $H$-compatible for some $\pEE{F}$ such that $\pE{F}{\catV}$ satisfies Eq.~\eqref{eq:standard-format-parametric-type}. 
By the definition $rCBV$ model morphism, we have that  $\left( \pEE{\fddtpointRecT{\uE} }, \pEE{\wfixpointRecT{\uF} }\right) $ is $H$-compatible. Hence, we only need to prove that $\pE{\wfixpointRecT{\uF} }{\catV}$ is given by
 \eqref{eq:standard-format-parametric-type}.  
\begin{enumerate}[(I)] 
\item We inductively define the set $\TreeIndex $ by the following. Let $r\in \KI$: (a)  if $\sind{n+1,r} = 0$, then $r\in\TreeIndex $; (b) if $\sind{n+1,r} \neq 0$, then, for any $\mathtt{T}\in \TreeIndex ^{\sind{n+1,r}} $, the pair $\left( \mathtt{T}, r\right) $ is in $\TreeIndex $.
\item We inductively define the family  $\left( \treeindexing{m}{ j, \mathtt{T} }\right) _{\left( j, \mathtt{T}\right) \in\left( \NNN{n}\cup \left\{ 0\right\}\right) \times\TreeIndex   } $ of indices by the following. Let $r\in \KI $:  (a) if $\sind{n+1,r} = 0 $,  we define $\treeindexing{m}{ j, r } \coloneq \sind{j, r}$ for each $j$; (b) if $\sind{n+1,r} \neq 0 $,  given $\mathtt{T} = \left( \mathtt{T}_i\right) _{i\in\NNN{\sind{n+1,r}} }\in \TreeIndex ^{\sind{n+1,r}} $, we define $\treeindexing{m}{ j, \left( \mathtt{T} , r  \right)  } $ by \eqref{eq:definition-of-K-overline-Tr} 
for each $j$.
\end{enumerate}
\small\noindent\begin{minipage}{.6\linewidth}	
	\begin{equation}\label{eq:standard-format-parametric-type-proof} 
		\pE{\uF}{\catV }\left( W_i, Y_i\right) _{i\in\NNN{n+1}} =\coprod _{r\in {\KI } }\left(  {H\left( \reals\right) }^{\sind{0,r}}\times \prod _{i=1}^{n+1} Y_i ^{\sind{i,r} }\right) 
	\end{equation} 
	\normalsize
\end{minipage}
\noindent\begin{minipage}{.4\linewidth}
	\begin{equation}\label{eq:definition-of-K-overline-Tr} 
\treeindexing{m}{ j, \left( \mathtt{T} , r  \right)  } = \sind{j,r} + \sum _{i=1}^{\sind{n+1,r}}\treeindexing{m}{ j, \mathtt{T}_i  }
	\end{equation} 	
\end{minipage} \\
\normalsize Let $X=\left( W_i, Y_i\right) _{i\in\NNN{n}}\in \left(\catV ^\op \times \catV  \right) ^n $,  $\FforREC \coloneq \paE{\uF}{X}{\catV}\left( \initiall , - \right) $ and $\coproje{}$ the obvious unique morphism. The colimit of \eqref{eq:the-colimit-giving-the-nuF} is isomorphic to \eqref{eq:the-coproduct-of-the-tree}. Hence, 
by the definition of the \fddt operator $\wfixpointRecT$ of $ \rCBVU\left( \catV , \monadT\right) = \left( \catV , \monadT , \wffixpointRecT\right)$,
$\pE{\wfixpointRecT\uF}{\catV} $ is given by the formula given in \eqref{eq:standard-format-parametric-type}. This completes the proof.\\
\noindent\begin{minipage}{.6\linewidth}
\footnotesize\begin{equation}\label{eq:the-colimit-giving-the-nuF}
	\diag{diag-chain-of-ep-pairs-for-the-F}
\end{equation} 
\end{minipage} 
\noindent
\begin{minipage}{.4\linewidth}	
\small	\begin{equation}
		\label{eq:the-coproduct-of-the-tree}
\coprod _{\mathtt{T}\in {\TreeIndex} }\left(  {H\left( \reals\right) }^{\mind{0,\mathtt{T}}}\times \prod _{j=1}^{n} Y_j ^{\mind{j,\mathtt{T}} }\right) 
	\end{equation} 
\end{minipage}

Finally, if $\uptau \in\SynVr $ corresponds to a data type $\ty{1}$, then the constant parametric type $ \pEE{\unlift{\uptau }}  $  equal to $\uptau $ is an $\left( \SynVr, \SynTr\right)$-parametric data type of degree $1$. Hence, denoting by $\pEE{\unlift{H\uptau }}$ the constant parametric type equal to $H\left(\uptau \right) $, since $\left(\pEE{\unlift{\uptau} }  ,  \pEE{\unlift{H\uptau} } \right) $ is $H$-compatible, 
we conclude that  \eqref{eq:standard-format-parametric-type}  holds for some $\left( l_j\right) _{j\in L}$ where $L$ is countable.
\end{proof}

In particular, for any non-parametric (meaning: $0$-variable-parametric)
recursive data type $R$, we have the following: \\
\begin{minipage}{.5\linewidth}
	\begin{equation}\label{eq:logical-relations-for-data-recursive-types} 
		\semLR{n,k}{R} = \coprod _{j\in L}\semLR{n,k}{\reals}^{l_j}
	\end{equation} 
\end{minipage} 
\begin{minipage}{.5\linewidth}	
	\begin{equation}\label{eq:semantics-for-data-types-involving-recursion} 
		\sem{R} = \coprod _{j\in L} \RR ^{l_j}.
	\end{equation} 
	\normalsize
\end{minipage}\\

This lets us strengthen our correctness theorem to apply also to programs between recursive data types:
\begin{proposition}\label{prop:basal-version-of the-correctness-theorem-recursive-data-types}
	Let $\displaystyle t: \uptau\to\upsigma  $  be a morphism in $\SynVr$. If $\uptau $ and $\upsigma $ correspond to data types,  
	$\displaystyle\sem{ t } : \coprod _{r\in \KI} \RR  ^{s _r} \to \monadwP{\coprod _{j\in L} \RR ^{l _j} } $ is differentiable and, for any $k\in \left( \NN\cup\left\{ \infty \right\}\right) $, $	\semt{k}{\DSynrec\left( t \right) } = \dDSemtra{k}{\sem{ t }}  $.
\end{proposition}	
\begin{proof}
	First of all, indeed, by Proposition \ref{prop:image-of-recursive-data-types}, we have that there are countable families $\left( s _r\right) _{r\in\KI}$ and  $\left( l_j\right) _{j\in L}$ such that
	\begin{equation}\label{eq:fundamental-morphism-LR-argument-for-recursive-types-correctness}
		\semLR{{s_i},k}{t} : \coprod _{r\in \KI} \semLR{{s_i},k}{\reals } ^{s _r}  \to \monadLR{{s_i}, k}{ \coprod _{j\in L}\semLR{{s_i},k}{\reals }^{l _j} }
	\end{equation}
is a morphism in $\SUBscone{\wCpo }{\sconeFUNCTOR{{s_i}, k}}$, for each $i\in\KI$ and any $k\in\NN\cup\left\{ \infty\right\}$.
	
By the commutativity of \eqref{eq:basic-commutative-diagram-for-the-recursive-types-LR-argument} for any $(s_i, k)\in\NN\times \left( \NN \cup\left\{ \infty\right\}\right)$, we get that
the pair $\left( \sem{t},  \semt{k}{\DSynrec{\left( t\right) } } \right) $
defines the morphism \eqref{eq:fundamental-morphism-LR-argument-for-recursive-types-correctness} for each $i\in\KI$.  By Corollary \ref{coro:fundamental-LR-conclusion-about-morphisms-subscone}, this implies that $\sem{t}$ is differentiable and $	\semt{k}{\DSynrec\left( t \right) } = \dDSemtra{k}{\sem{ t }}  $.
\end{proof}

Finally, as a consequence, we get:
\begin{theorem}\label{theo:main-theorem-section-proof-recursive}
	Assume that $\tangentreals $ implements the vector space $\RR ^k$, for some $k\in\NN\cup\left\{ \infty \right\}$. For any program $\var{1}:\ty{1}\vdash \trm{1}:\ty{2}$  where	$\ty{1},\ty{2}$ are data types (including recursive data types), we have that $\sem{\trm{1}} $ is differentiable and, moreover, 
	\begin{equation}
		\semt{k}{\Dsyn{\trm{1} }} = \dDSemtra{k}{\sem{\trm{1} }}
	\end{equation}
	provided that $\Dsynsymbol $ is sound for primitives.
\end{theorem} 	 

Following the considerations of Section \ref{sub:forward-mode-types-correctness} and \ref{sub:reverse-mode-types-correctness}, it follows from Theorem \ref{theo:basal-version-of the-correctness-theorem-recursive-types} that $\Dsynsymbol $ as defined in Section \ref{sub:extended-macro} \textit{correctly} provides us with forward and reverse AD transformations for data types.

\subsection{AD on arrays}
Arrays are semantically the same as lists: in our language, if $\ty{1}$ is a data type, an array of $\ty{1}$ is given by $\trec{\tvar{1}}{\Unit\t+ \ty{1}\t*\tvar{1}}$. 
It should be noted that, if $\var{1}:\trec{\tvar{1}}{\Unit\t+ \ty{1}\t*\tvar{1}}\vdash \trm{1}:\trec{\tvar{1}}{\Unit\t+ \ty{1}\t*\tvar{2}}$, we have that $$\sem{\trm{1}} : \coprod _{i=1}^{\infty } \sem{\ty{1}}\to \monadwP{ \coprod _{i=1}^{\infty } \sem{\ty{2}}}  .$$
By Theorem \ref{theo:main-theorem-section-proof-recursive}, if $\ty{1}$ and $\ty{2}$ are data types, we get that $\dDSemtra{k}{\sem{\trm{1} }}$ (as defined in \eqref{eq:derivative-parially-defined-functions}) is equal to $\semt{k}{\Dsyn{\trm{1} }}$.
Therefore, Theorem \ref{theo:main-theorem-section-proof-recursive} already encompasses the correctness for arrays (of data types).


%% file: TEX/types-values-computations2.tex
\begin{syntax}
    \ty{1}, \ty{2}, \ty{3} & \gdefinedby & & \syncat{types}                          \\
    &\gor& \ldots                      & \synname{as before}\\
    &&&\\
    \val{1}, \val{2}, \val{3} & \gdefinedby & & \syncat{values}                          \\
    &\gor& \ldots                      & \synname{as before}\\
    &&&\\
    \trm{1}, \trm{2}, \trm{3} & \gdefinedby & & \syncat{computations}                          \\
    &\gor& \ldots                      & \synname{as before}\\
  \end{syntax}%
~\qquad\quad
 \begin{syntax}
  &\gor & \tvar{1},\tvar{2},\tvar{3}   & \synname{type variables}\\
	&\gor\quad\, & \trec{\tvar{1}}\ty{1} & \synname{recursive type}\\
	& & &\\
  &\gor& \tRoll\val{1}                      & \synname{recursive type introduction}\\
	&&&\\
  &&&\\
  &\gor& \tRoll\trm{1}                      & \synname{recursive type introduction}\\
	&\gor\quad\, & \rMatch{\trm{1}}{\var{1}}{\trm{2}} & \synname{recursive type elimination}
\end{syntax}

%% file: TEX/type-system2.tex
\[
  \begin{array}{c}
\inferrule{
  \Delta\mid\Ginf{\trm{1}}{\subst{\ty{2}}{\sfor{\tvar{1}}{\trec{\tvar{1}}{\ty{2}}}}}
}{
  \Delta\mid\Ginf {\tRoll{\trm{1}}}{\trec{\tvar{1}}{\ty{2}}}    
}\qquad \qquad \qquad 
\inferrule{
  \Delta\mid\Ginf{\trm{1}}{\trec{\tvar{1}}{\ty{2}}}
\quad 
\Delta\mid\Ginf[,\var{1}:\subst{\ty{2}}{\sfor{\tvar{1}}{\trec{\tvar{1}}{\ty{2}}}}]{\trm{2}}{\ty{1}}
}
{
  \Delta\mid\Ginf{\rMatch{\trm{1}}{\var{1}}{\trm{2}}}{\ty{1}}}
\end{array}
\]

%% file: TEX/beta-eta2.tex
\[
\begin{array}{l@{\qquad\qquad}l}
\rMatch{\tRoll\val{1}}{\var{1}}{\trm{1}}\ \beeq\ \subst{\trm{1}}{\sfor{\var{1}}{\val{1}}} &
    \subst{\trm{1}}{\sfor{\var{3}}{\val{1}}}\freeeq{\var{1}}\rMatch{\val{1}}{\var{1}}{\subst{\trm{1}}{\sfor{\var{3}}{\tRoll{\var{1}}}}} 
\end{array}    
\]

%% file: TEX/d-types2.tex
\[
\begin{array}{l@{\ }l@{\ }l@{\qquad}l@{\ }l}
\Dsynplain{\tvar{1}} \defeq {\tvar{1}}\qquad & \Dsynplain{\trec{\tvar{1}}{\ty{1}}} \defeq \trec{\tvar{1}}{\Dsynplain{\ty{1}}}
& \end{array}    
\]

%% file: TEX/d-terms2.tex
\[
    \begin{array}{ll}
    \Dsynplain{\tRoll{\trm{1}}} \defeq \tRoll{\Dsyn{\trm{1}}}\qquad &\Dsynplain{\rMatch{\trm{1}}{\var{1}}{\trm{2}}}\defeq 
    \rMatch{\Dsynplain{\trm{1}}}{\var{1}}{\Dsynplain{\trm{2}}} 
    \end{array}
    \]

%% file: TEX/almost-everywhere-differentiability.tex
\section{Almost everywhere correct AD}
\label{sec:almost-everywhere-differentiability}
Here, we show how some of the arguments of \cite{huot2023omegapap}
about almost everywhere differentiability
can be  accommodated in our framework, by making use of a minor variation 
of our chosen logical relations over $\omega$-cpos.
The resulting arguments use plain logical relations over $\omega$-cpos and do not rely on sheaf-structure.
They are also a bit more general, as they apply to languages with coproduct and recursive types.

The central notion is \cite{lee2020correctness}'s concept of functions that are
 piecewise analytic under analytic partition (PAP).
We recall some of the required notions to talk about PAP functions first.
 \begin{definition}[Analytic function]
A function $f:U\to V$, for  $U\subseteq \RR^n$ and $V\subseteq \RR^m$, is \emph{analytic} if, for all $x\in U$, its Taylor series converges pointwise to $f$ on an open neighbourhood of $x$.
 \end{definition} 

 \begin{definition}[($c$)-Analytic set]
A subset $A\subseteq \RR^n$ is called \emph{analytic} if there exist analytic functions
$g_1,\ldots, g_m:U\to \RR$ defined on an open neighbourhood $U$ of $A$, such that 
$$
A=\{x\in U\mid g_i(x)\geq 0\text{ for $1\leq i\leq m$}\}.
$$
A subset $A\subseteq \RR^n$ is called \emph{$c$-analytic} if it is the countable union of analytic subsets.
 \end{definition} 
 As noted by \cite{huot2023omegapap}, we can equivalently define a $c$-analytic set as a countable \emph{disjoint} union of analytic subsets.

\begin{definition}[PAP function]
    A function $f:U\to V$, for  $U\subseteq \RR^n$ and $V\subseteq \RR^m$, 
    is called \emph{piecewise analytic under analytic partition
    (PAP)} if it has a \emph{PAP representation} in the sense of a countable family $\{(A_i,
    U_i, f_i)\}_{i\in I}$ such that:
    \begin{itemize}
    \item the sets $A_i$ are analytic and form a partition of $U$;
    \item each $f_i
    : U_i \to V$
     is an analytic function defined on
    an open neighbourhood $U_i$ of  $A_i$;
    \item $f_i|_{A_i}=f|_{A_i}$ in the sense that $f_i(x) = f(x)$ for all $x \in A_i$.
    \end{itemize}
\end{definition}
A crucial observation by \cite{lee2020correctness} is that PAP functions are closed under composition.
As noted by \cite{huot2023omegapap}, a subset $A\subseteq \RR^n$ is $c$-analytic if and only if the inclusion $A\hookrightarrow \RR^n$ is a PAP function.

We consider the following notion of partial PAP function.
\begin{definition}[Partial PAP function]
    We call a partial function $f:U\rightharpoonup V$
    a \emph{partial PAP function} if its domain of definition is $c$-analytic
    and it restricts to a (total) PAP function on its domain.
\end{definition}
As noted by \cite{huot2023omegapap}, such partial PAP functions are closed under composition.

\begin{definition}[Intensional derivative]Each particular PAP representation $\{(A_i, U_i
, f_i)\}_{i\in I}$ of a PAP function $f$ gives rise to a unique \emph{intensional derivative} 
$\{(A_i, U_i
, Df_i)\}_{i\in I}$, where we write $D f_i$ for the (standard) derivative of $f_i$,
such that $Df_i=Df$ on $A_i$.
\end{definition}
A given PAP function may therefore have several distinct intensional derivatives, arising from the different PAP representations.
However, \cite{lee2020correctness} show that such PAP functions $f$ are differentiable almost everywhere and that each intensional derivative corresponds almost everywhere with the (standard) derivative of $f$.

Next, we redefine our logical relations for $\reals$ and monadic types from Sections \ref{subsec:LR-assignment} and \ref{subsect:the-definition-of-the-monad-for-the-logical-relations}.
First, we redefine
$$
\semLR{n,k}{\reals }\defeqq   \left(\left\{ \left( f : \RR ^n\to \RR , \secondM{f} \right) : f\mbox{ is analytic, } \secondM{f} =  \dDSemtotaltra{k}{f} \right\}, \left( \RR , \RR \times \RR ^k \right), \mathrm{incl.}   \right). 
$$
Second, we denote by  
$\topologyO{n}$
the set of countable families $\{(A_i, U_i)\}_{i\in I}$ of pairs of analytic subsets $A_i\subseteq \RR^n$ and open neighbourhoods $U_i$ of $A_i$ in, such that all $A_i$ are pair-wise disjoint and $\bigsqcup_{i\in I}A_i\neq \emptyset,\RR^n$.
Then, for each $\{(A_i, U_i)\}_{i\in I}\in \topologyO{n} $,
we redefine
\begin{eqnarray*} 
\openLift{\{(A_i, U_i)\}_{i\in I}, n, k} &\defeqq &
\bigsqcup_{i\in I}
\left( \left\{ \left( g: \RR ^n\to  U_i, \dDSemtotaltra{k}{g} \right) : g \mbox{ is analytic} \right\} , \left(U_i,  \intle{n,k}\left(   U_i\times \left( \RR ^k\right)  ^n\right)  \right), \mathrm{incl.}\right) \\ 
&\in  &\SUBscone{\wCpo }{ \sconeFUNCTOR{n,k} } .
\end{eqnarray*} 
We redefine the $\SUBscone{\wCpo }{\sconeFUNCTOR{n,k}}$-monad $\monadLR{n, k}{-}$ on $\SUBscone{\wCpo }{\sconeFUNCTOR{n,k}}$ by
\begin{equation*}
	\monadLR{n, k}{D, \left( C, C'\right), j} \defeqq  \left( \unlift{\monadLR{n, k}{D, \left( C, C'\right) , j}} , \left( \monadwP{C} , \monadwP{C'}\right) , \morLRmonad{X}	
	\right) 
\end{equation*}
where $\unlift{\monadLR{n, k}{D, \left( C, C'\right) , j}}\subseteq \sconeFUNCTOR{n,k}(\Lift{C},\Lift{C'}) $ is the union
\begin{equation*} 
\left\{  \leastelement  \right\}
\sqcup 	
 D
 \sqcup 
 \left( \coprod _ {\{(A_i, U_i)\}_{i\in I}\in\topologyO{n} } 
 \ehom{\SUBscone{\wCpo }{\sconeFUNCTOR{n, k}}}{\openLift{\{(A_i, U_i)\}_{i\in I}, n, k} }{\left( D, \left( C, C'\right) , j\right)}\right)/\sim,
\end{equation*} 
where we identify $[(\gamma_{i}, \gamma'_{i})\mid i\in I]\sim [(\overline\gamma_{j}, \overline\gamma'_{j})\mid j\in J]$ if their domains of definition coincide ($\bigsqcup_{i\in I}A_i=\bigsqcup_{j\in J}\overline A_j$) and they define the same function on this domain.
To be more formal, we define the identification $[(\gamma_{i}, \gamma'_{i})\mid i\in I]\sim [(\overline\gamma_{j}, \overline\gamma'_{j})\mid j\in J]$ if
$\bigsqcup_{i\in I}A_i=\bigsqcup_{j\in J}\overline A_j$
and  $[\gamma_{i} \circ \iota_i\mid i\in I]=[\overline \gamma_{j} \circ \overline\iota_j\mid j\in J]$ and 
$[\gamma'_{i} \circ\intle{n,k}\circ (\iota_i\times \id[(\RR^k)^n])\circ \intle{n,k}^{-1}\mid i\in I]=[\overline \gamma'_{j} \circ \intle{n,k}\circ (\overline\iota_j\times \id[(\RR^k)^n])\circ \intle{n,k}^{-1}\mid j\in J]$, where we use the
inclusions  $\iota_i:A_i\hookrightarrow U_i$ and $\overline\iota_j:\overline A_j\hookrightarrow \overline U_j$.
The structure of the monad is defined entirely analogously to that in Section \ref{subsect:the-definition-of-the-monad-for-the-logical-relations}.
Closure under suprema of $\omega$-chains follows from \cite[Corollary B.9]{huot2023omegapap}.
It is easy to see that that the conditions of Theorem \ref{the:maybe-the-main-result-on-LR-recursive} are satisfied as before.

The rest of the development remains essentially unchanged, except for the minor modification that we work with (1) PAP functions rather than differentiable functions and (2) countable families of analytic subsets with open neighbourhoods rather than open subsets.

If we spell out the resulting definitions for the logical relations (focusing on the $k$-semantics for $k=1$), the result is as follows:
\begin{align*}
    T^n_{\reals} &\defeq \{(\gamma,\gamma')\mid \gamma \text{ is PAP and } \gamma'=(x,v)\mapsto  (\gamma(x),\gamma''(x,v))\text{ for an intensional derivative $\gamma''$ of $\gamma$} \}\\
    P^{n}_{\ty{1}} & \defeq \Big\{(\gamma,\gamma')\mid \gamma^{-1}(\sem{\ty{1}})\times \RR^n=\gamma'^{-1}(\sem{\Dsyn{\ty{1}}})\text{ and there exists a countable analytic partition }\\
    &\{A_i\subseteq \RR^n\}_{i\in I}\text{ of }\gamma'^{-1}(\sem{\Dsyn{\ty{1}}})\text{ and there exist open neighbourhoods $U_i$ of $A_i$ with functions }\\
    &\gamma_i:U_i\to \sem{\ty{1}}, \gamma'_i:U_i\times \RR^n\to \sem{\Dsyn{\ty{1}}}\text{ such that }\gamma|_{A_i}=\gamma_i|_{A_i}\text{ and } \gamma'|_{A_i\times \RR^n}=\gamma'_i|_{A_i\times \RR^n} \text{ and}\\
    & \text{for all analytic }\delta:\RR^n\to U_i\text{ we have that } 
    (\gamma_i\circ \delta, (x,v)\mapsto (\gamma_i(\delta(x)), \gamma'_i(D\delta(x,v))))\in T^n_{\ty{1}}\Big\}.
    \end{align*}
We see that $P^{n}_{\reals^m}$ precisely captures \cite{huot2023omegapap}'s notion of partial PAP functions and their intensional derivatives, if we note that we can use (analytic) $\delta$ to define for any point $y\in \RR^n$ an arbitrary small neighbourhood: $x\mapsto \frac{x*\epsilon}{\sqrt{1+|\!|x|\!|^2}}+y$ is an analytic isomorphism between $\RR^n$ and an $\epsilon$-ball centred at $y$.
We can show (by induction) that $P^n_{\ty{1}}$ is closed under suprema of $\omega$-chains using \cite[Corollary B.9]{huot2023omegapap}.

With these new definitions, our entire development goes through again.
As long as we ensure that all our primitive operations denote partial PAP functions, we obtain versions of Theorem III.2 and Corollary III.3. of \cite{huot2023omegapap} for a language that additionally includes recursive types, by using a plain logical relations argument over $\omega$-cpos:
\begin{theorem}[Almost everywhere differentiability] \label{theo:almost-everywhere-differentiable}
	Assume that $\tangentreals $ implements the vector space $\RR ^k$, for some $k\in\NN\cup\left\{ \infty \right\}$. For any program $\var{1}:\ty{1}\vdash \trm{1}:\ty{2}$  where	$\ty{1},\ty{2}$ are data types (including recursive data types), we have that $\sem{\trm{1}} $ is differentiable almost everywhere on its domain and, moreover, 
	\begin{equation}
		\semt{k}{\Dsyn{\trm{1} }} = \dDSemtra{k}{\sem{\trm{1} }}
	\end{equation}
	almost everywhere, provided that $\Dsynsymbol $ is sound for primitives.
\end{theorem} 	 
Consequently, we obtain the correct derivative almost everywhere for any program $\trm{1}$ that terminates almost everywhere.
Importantly, this result remains true if we change the semantics of $\tSign\trm{1}$ to be defined even for $\trm{1}=0$, as is done in \cite{huot2023omegapap,DBLP:journals/pacmpl/MazzaP21}:
$$\sem{\tSign\trm{1}}\defeq \left\{\begin{array}{ll}\iota_1()&\text{ if }\sem{\trm{1}}\leq 0\\
\iota_2()&\text{ otherwise}\end{array}\right. .$$
Indeed, this semantics is still logical relation respecting, thanks to our choice of lifting of the partiality monad to logical relations.

%% file: TEX/related-work.tex
\section{Related Work}\label{sec:related-work}
This is an improved version of the unpublished preprint \cite{vakar2020denotational}.
In particular, we have simplified the correctness argument to no longer depend on diffeological or sheaf-structure and to have it apply to arbitrary differentiable (rather than merely smooth) operations.
We have further simplified the subsconing technique for recursive types.

There has recently been a flurry of work studying AD from a programming language point of view, a lot of it focussing on functional formulations of AD and their correctness. 
Examples of such papers are 
\cite{pearlmutter2008reverse,elliott2018simple,shaikhha2019efficient,brunel2019backpropagation,abadi-plotkin2020,bcdg-open-logical-relations,lee2020correctness,hsv-fossacs2020,DBLP:journals/pacmpl/MazzaP21,vakar2021chad,VAKAR-LUCATELLI2021,DBLP:journals/corr/abs-2101-06757,DBLP:journals/toplas/VakarS22,DBLP:journals/pacmpl/KrawiecJKEEF22,smeding2022,huot2023omegapap}.
Of these papers, \cite{pearlmutter2008reverse,abadi-plotkin2020,lee2020correctness,DBLP:journals/pacmpl/MazzaP21,smeding2022,huot2023omegapap} are particularly relevant as they also consider automatic differentiation of languages with partial features.

Here, \cite{pearlmutter2008reverse} considers an implementation that differentiates recursive programs and the implementation of \cite{smeding2022} even differentiates code that uses recursive types. 
They do not give correctness proofs, however.

\cite{abadi-plotkin2020} pioneers a notion of correctness that we use for most of this paper, where points of non-differentiability are essentially ignored by making a function undefined at such points.
They use it to give a denotational correctness proof of AD on a first-order functional language with (first-order) recursion.
The first-orderness of the language allows the proof to proceed by plain induction rather than needing a logical technique.

\cite{lee2020correctness} introduces a more ambitious notion of correctness in the sense of almost everywhere correct AD.
\cite{DBLP:journals/pacmpl/MazzaP21} proves the correctness of basically the same AD algorithms that we consider in this paper when restricted to PCF with a base type of real numbers and a real conditional.
Importantly, they also take care to prove almost-everywhere correct differentiation for a language that supports conditionals on real numbers and primitives that can have points of non-differentiability. 
Their proof relies on operational semantic techniques.
\cite{huot2023omegapap} combines the ideas of \cite{lee2020correctness} with those of \cite{vakar2020denotational} to give a denotational proof of almost everywhere correct AD for PCF, by using sheaves of logical relations.
Section \ref{sec:almost-everywhere-differentiability} of the present paper shows how their arguments can be reproduced without any sheaf-theoretic machinery, essentially by choosing a different lifting of the partiality monad to logical relations.

\cite{bcdg-open-logical-relations} have previously used (open) logical relations over the syntax, rather than semantics, to prove correctness of AD on total languages.
It would be interesting to see whether and how their techniques could be adapted to languages with partial features. 
We suspect that the choice between logical relations over the syntax or semantics is mostly a matter of taste, but that the extra (co)completeness properties that the semantics has can help, particularly when proving things about recursion and recursive types.

There is an independent line of inquiry into differential $\lambda$-calculus \cite{ehrhard2003differential} and differential categories \cite{DBLP:journals/acs/BluteCLS20,DBLP:conf/csl/CockettCGLMPP20}.
A conceptual distinction with the work on automatic differentiation is that differentiation tends to be a first-class construct (part of the language) in differential $\lambda$-calculus, rather than a code transformation in a meta-language.
Further, there is a stronger emphasis on the axioms that derivatives need to satisfy and less of a focus on recipes for computing derivatives.
In this setting, differential restriction categories \cite{cockett2012differential} gives a more abstract semantic study of the interaction between (forward) differentiation and partiality. 
We found that for our purposes, a concrete semantics in terms of \wcpos{} sufficed, however.

Our contribution is to give an alternative denotational argument, which we believe is simple and systematic, and to extend it to apply to languages which, additionally, have the complex features of recursively defined data structures that we find in realistic ML-family languages.

Such AD for languages with expressive features such as recursion and user-defined data types has been called for by the machine learning community \cite{jeong2018improving, van2018automatic}.
Previously, the subtlety of the interaction of automatic differentiation and real conditionals had first been observed by \cite{beck1994if}.

Our work gives a relatively simple denotational semantics for recursive types,
which can be considered as an important special case of bilimit compact categories \cite{levy2012call}.
Bilimit compact categories are themselves, again, an important special case of the very general semantics of recursive types in terms of 
algebraically compact categories \cite{freyd1991algebraically}.
We believe that working with this special case of the semantics significantly simplifies our presentation.

In particular, this simplified semantics of recursive types allows us to give a very simple but powerful (open, semantic) logical technique for recursive types.
It is an alternative to the two existing techniques for logical relations for recursive types: relational properties of domains \cite{pitts1996relational}, which is quite general but very technical to use, in our experience, and step-indexed logical relations \cite{DBLP:conf/esop/Ahmed06}, which are restricted to logical relations arguments about syntax, hence not applicable to our situation.

Finally, we hope that our work adds to the existing body of programming languages literature on automatic differentiation and recursion (and recursive types).
In particular, we believe that it provides a simple, principled denotational explanation of how AD and expressive partial language features should interact.
We plan to use it to generalise and prove correct the more advanced AD technique CHAD \cite{elliott2018simple, vakar2021chad,DBLP:journals/toplas/VakarS22,VAKAR-LUCATELLI2021,ad-2022-kerjean:hal-03123968} when applied to languages with partial features.

%% file: TEX/fine-language.tex
\section{Fine grain call-by-value and AD}\label{appx:fine-grain-cbv}
In \S\ref{sec:our-cbv-language}, we have discussed a standard coarse-grain CBV language, also known as the $\lambda_C$-calculus,
computational $\lambda$-calculus \cite{moggi1988computational},
or, plainly, CBV.
In this appendix, we discuss an alternative presentation in terms of 
 fine-grain CBV\footnote{This is a type theory for the Freyd category given by the Kleisli functor of the partiality monad.
In the presence of the connectives we consider (in particular, function types), it is equivalent to Moggi's monadic metalanguage \cite{moggi1991notions}.} \cite{levy2003modelling,levy2012call}.
While it is slightly more verbose, this presentation clarifies the precise universal property that 
is satisfied by the syntax of our language.

\subsection{Fine grain call-by-value}\label{subsec:finegrain-CBV-language}
We consider a standard fine-grain call-by-value language (with complex values) over a ground type $\reals$ of 
real numbers, real constants $\cnst{c}\in\Op_0$ for $c\in\RR$,
and certain basic operations $\op\in\Op_n$ for each natural number $n\in\NN$.

The types $\ty{1},\ty{2},\ty{3}$, (complex) values $\val{1},\val{2},\val{3}$, and computations $\trm{1},\trm{2},\trm{3}$ of our language are as follows.\\
\input{TEX/fine-types-values-computations}\\
We will use sugar
\begin{align*}&\ifelse{\val{1}}{\trm{1}}{\trm{2}}\defeq \toin{\var{1}}{\tSign(\val{1})}
{\vMatch{\var{1}}{{
   \_\To{\trm{2}}
\vor \_\To{\trm{3}}
}}}\\
&\tFst\val{1}\defeq \pMatch{\val{1}}{\var{1}}{\_}{\var{1}}\\
&\tSnd\val{1}\defeq \pMatch{\val{1}}{\_}{\var{1}}{\var{1}}\\
&
\letrec{f}{\var{1}}{\trm{1}}{\trm{2}}\defeq \toin{f}{(\rec{f}{\tReturn(\fun{\var{1}}\trm{1})})}{\trm{2}}.
\end{align*}
We could also define iteration as syntactic sugar:\\
$\tItFrom{\trm{1}}{\var{1}}{\val{1}}\defeq \left(\rec{\var{3}}{\fun{\var{1}}{\toin{\var{2}}{\trm{1}}{\bvMatch{\var{2}}{\var{1}'}{ \var{3}\, \var{1}'}{\var{1}''}{\tReturn \var{1}''}}}}\right)\, \val{1}$.

The typing rules are in Figure~\ref{fig:fine-type-system}.

\begin{figure}[!ht]
  \fbox{\parbox{0.98\linewidth}{\begin{minipage}{\linewidth}\noindent\input{TEX/fine-type-system}\end{minipage}}}
  \caption{Typing rules for the our fine-grain CBV language with iteration and real conditionals.
  We use a typing judgement $\vdash^v$ for values and $\vdash^c$ for computations.
  \label{fig:fine-type-system}}
  \end{figure}


\subsection{Equational theory}
We consider our language up to the usual $\beta\eta$-equational theory for fine-grain CBV,
which is displayed in Figure \ref{fig:fine-beta-eta1}.

\begin{figure}[!ht]
   \fbox{\parbox{0.98\linewidth}{\begin{minipage}{\linewidth}\noindent
\input{TEX/fine-beta-eta1}
\end{minipage}}}
\caption{Standard $\beta\eta$-laws for fine-grain CBV.
We write $\freeeq{\var{1}_1,\ldots,\var{1}_n}$ to indicate that the variables are fresh in the left hand side.
In the top right rule, $\var{1}$ may not be free in $\trm{3}$.
Equations hold on pairs of terms of the same type.\label{fig:fine-beta-eta1}}
\end{figure}
Under the translation of coarse-grain CBV into fine-grain CBV, this equational theory induces precisely that of Section \ref{sec:our-cbv-language}.

\subsection{The $CBV$ model $\left( \SynV, \SynT,\Synfix , \Synit   \right) $}\label{subsec:the category syntax}
Our fine grain call-by-value language corresponds with a $CBV$ model (see
Definition~\ref{def:CBVmodel}).

\textit{We define the category $\SynV$ of values}, which has types as objects. 
$\SynV(\ty{1},\ty{2})$ consists of $(\alpha)\beta\eta$-equivalence classes of 
values $\var{1}:\ty{1}\vdash^v \val{1}:\ty{2}$, where identities are $\var{1}:\ty{1}\vdash^v \var{1}:\ty{2}$ 
and composition of $\var{1}:\ty{1}\vdash^v \val{1}:\ty{2}$ and $\var{2}:\ty{2}\vdash^v \val{2}:\ty{3}$
is given by $\var{1}:\ty{1}\vdash^v \subst{\val{2}}{\sfor{\var{2}}{\val{1}}}:\ty{3}$. 
\begin{lemma} 
$\SynV $ is bicartesian closed.
\end{lemma}

Similarly, \textit{we define the category $\SynC$ of computations}, which also has types as objects.
$\SynC(\ty{1},\ty{2})$ consists of $(\alpha)\beta\eta$-equivalence classes of 
computations $\var{1}:\ty{1}\vdash^c \trm{1}:\ty{2}$, where identities are $\var{1}:\ty{1}\vdash^c \tReturn\var{1}:\ty{2}$ 
and composition of $\var{1}:\ty{1}\vdash^c \trm{1}:\ty{2}$ and $\var{2}:\ty{2}\vdash^c \trm{2}:\ty{3}$
is given by $\var{1}:\ty{1}\vdash^c \toin{\var{2}}{\trm{1}}{\trm{2}}:\ty{3}$.
\begin{lemma} 
	$\SynC $ is a $\SynV$-category.
\end{lemma} 

We define the $\SynV $-functors  \\
\noindent  \begin{minipage}{.5\linewidth}
	\begin{eqnarray*}
		\SynG: &\SynC &\embed  \SynV\\
		&\ty{1} & \mapsto \left( {\Unit\To\ty{1}}\right)   \\ 
		& \trm{1} & \mapsto \fun{\tUnit}{\trm{1}} 
	\end{eqnarray*}   
\end{minipage}
\begin{minipage}{.5\linewidth}
 \begin{eqnarray*}
	\SynJ: &\SynV &\embed  \SynC\\
	&\ty{1} & \mapsto \ty{1}\\ 
	& \val{1} & \mapsto \tReturn\val{1} .
\end{eqnarray*}
\end{minipage}
\\ \\
We have that $\SynJ \dashv \SynG $ is a (Kleisli) $\SynV $-adjunction $\SynJ \dashv \SynG $ and, hence, denoting by $\SynT$ the 
induced $\SynV$-monad, $\left( \SynV , \SynT \right) $  is a $CBV$ pair, as defined in Definition \ref{def:CBVpair}. Moreover, considering the free recursion and free iteration
\begin{eqnarray*}
\Synit : \qquad &\left(\var{1}:\ty{2}\vdash^c \trm{1}:\ty{2}\t+\ty{1}\right) &\mapsto  \fun{ \var{2} }  { \left( {\tItFrom{\trm{1}}{\var{1}}{\var{2}}} \right) } \\
\Synfix :\qquad  &\left(\var{1}:\ty{1}\vdash^v \val{1}:\ty{1}\right) &\mapsto   {\rec{\var{1}}{\val{1}}} \quad   (\ty{1}=\ty{2}\To\ty{3}),
\end{eqnarray*} 
we get the $CBV$ model $\left( \SynV, \SynT,\Synfix , \Synit   \right) $ which has the following universal property.

\begin{proposition}[Universal Property of the Syntax]\label{prop:universal-property-syntax}
	Let $\left( \catV , \monadT, \fixpoint , \iterate \right) $  be a $CBV$ model with chosen finite products, coproducts and exponentials.  For each consistent assignment  
	\begin{eqnarray}
		H(\reals) &\in & \obb{\catV} \\
		H (\cnst c) &\in & \ehomV{\terminall}{H(\reals)} \\
		H (\op)&\in  &\ehomC{H(\reals)^n }{ H(\reals) } =  \ehomV{H(\reals)^n }{ TH(\reals) }, \mbox{ for each } \op\in\Op_n  \\
		H (\tSign) & \in &\ehomC{H(\reals)} { \terminall\sqcup \terminall }  = \ehomV{H(\reals)} {T\left( \terminall\sqcup \terminall\right)} 
	\end{eqnarray}   
	there is a unique $CBV$ model morphism $H$ between $\left( \SynV, \SynT,\Synfix , \Synit   \right)$
	and $\left( \catV , \monadT, \fixpoint , \iterate \right) $ respecting it.	
\end{proposition}

\begin{proposition}[Universal Property of the Syntax]\label{prop:universal-property-syntax-target}
	Let $\left( \catV , \monadT, \fixpoint , \iterate \right) $  be a $CBV$ model with chosen finite products, coproducts and exponentials.  For each consistent assignment  
	\begin{eqnarray}
		H(\reals) &\in & \obb{\catV} \\
		H (\cnst c) &\in & \ehomV{\terminall}{H(\reals)} \\
		H (\op)&\in  &\ehomC{H(\reals)^n }{ H(\reals) } =  \ehomV{H(\reals)^n }{ TH(\reals) }, \mbox{ for each } \op\in\Op_n  \\
		H (\tSign) & \in &\ehomC{H(\reals)} { \terminall\sqcup \terminall }  = \ehomV{H(\reals)} {T\left( \terminall\sqcup \terminall\right)} 
	\end{eqnarray}   
	there is a unique $CBV$ model morphism $H$ between $\left( \SynV, \SynT,\Synfix , \Synit   \right)$
	and $\left( \catV , \monadT, \fixpoint , \iterate \right) $ respecting it.	
\end{proposition}
\subsection{A translation from coarse-grain CBV to fine-grain CBV}\label{subsec:translation-finegrain}
This translation $(-)^\dagger$ operates on types and contexts as the identity.
It faithfully translates terms $\Gamma\vdash \trm{1}:\ty{1}$ of coarse-grain CBV
into computations $\Gamma\vdash^c \trm{1}^\dagger : \ty{1}$ of fine-grain CBV.
This translation illustrates the main difference between coarse-grain and 
fine-grain CBV: in coarse-grain CBV, values are subset of computations,
while fine-grain CBV is more explicit in keeping values and computations separate.
This makes it slightly cleaner to formulate an equational theory,
denotational semantics, and logical relations arguments.

We list the translation $(-)^\dagger$ below where all newly introduced variables 
are chosen to be fresh.
\[
\begin{array}{l|l}
	\textnormal{coarse-grain CBV computation } \trm{1} & \textnormal{fine-grain CBV translation } \trm{1}^\dagger\\
	\hline 
	\var{1} & \tReturn \var{1} \\
	\letin{\var{1}}{\trm{1}}{\trm{2}} & \toin{\var{1}}{\trm{1}^\dagger}{\trm{2}^\dagger}\\
	\cnst c & \tReturn \cnst c \\
	\tInl \trm{1} & \toin{\var{1}}{\trm{1}^\dagger}{\tReturn\tInl \var{1}}\\
	\tInr \trm{1} & \toin{\var{1}}{\trm{1}^\dagger}{\tReturn\tInr \var{1}}\\
	\tUnit & \tReturn \tUnit\\
	\tPair{\trm{1}}{\trm{2}} & \toin{\var{1}}{\trm{1}^\dagger}{\toin{\var{2}}{\trm{2}^\dagger}{\tReturn\tPair {\var{1}} {\var{2}}}}\\
	\fun{\var{1}}{\trm{1}} & \tReturn\fun{\var{1}}{\trm{1}^\dagger}\\
	\op(\trm{1}_1,\ldots,\trm{1}_n) & \toin{\var{1}_1}{\trm{1}_1^\dagger}{\ldots\toin{\var{1}_n}{\trm{1}_n^\dagger}{\op(\var{1}_1,\ldots,\var{1}_n)}}\\
	\nvMatch{\trm{1}} & \toin{\var{1}}{\trm{1}^\dagger}{\nvMatch{\var{1}}}\\
	\bvMatch{\trm{1}}{\var{1}}{\trm{2}}{\var{2}}{\trm{3}} & \toin{\var{3}}{\trm{1}^\dagger}{\bvMatch{\var{3}}{\var{1}}{\trm{2}^\dagger}{\var{2}}{\trm{3}^\dagger} }\\
	\pMatch{\trm{1}}{\var{1}}{\var{2}}{\trm{2}} & \toin{\var{3}}{\trm{1}^\dagger}{\pMatch{\var{3}}{\var{1}}{\var{2}}{\trm{2}^\dagger}}\\
	\trm{1}\,\trm{2} & \toin{\var{1}}{\trm{1}^\dagger}{\toin{\var{2}}{\trm{2}^\dagger}{\var{1}\,\var{2}}}\\
	\tItFrom{\trm{1}}{\var{1}}{\trm{2}} & \toin{\var{2}}{\trm{2}^\dagger}{\tItFrom{\trm{1}^\dagger}{\var{1}}{\var{2}}}\\
	\tSign\,\trm{1} & \toin{\var{1}}{\trm{1}^\dagger}{\tSign\,\var{1}}\\
	\rec{\var{3}}{\trm{1}} &   \rec{\var{3}}{\fun{\var{1}}{          \toin{\var{2}}{\trm{1}^\dagger}{\var{2}\var{1} } }}
\end{array}
\]
\subsection{Dual numbers forward AD transformation}
As before, we fix, for all $n\in\NN$, for all $\op\in\Op_n$, for all $1\leq i \leq n$, 
computations $\var{1}_1:\reals,\ldots,\var{1}_n:\reals\vdash^c \partial_i\op(\var{1}_1,\ldots,\var{1}_n):\reals$,
which represent the partial derivatives of $\op$.
Using these terms for representing partial derivatives,
we define, in Figure \ref{fig:fine-ad1}, a structure preserving macro $\Dsynsymbol$ on the types, values, and computations of our language for performing 
forward-mode AD.
We observe that this induces the following AD rule for our sugar: $\DsynC{\ifelse{\val{1}}{\trm{1}}{\trm{2}}}=\pMatch{\DsynV{\val{1}}}{\var{1}}{\_}{\ifelse{\var{1}}{\DsynC{\trm{1}}}{\DsynC{\trm{2}}}}$.
\begin{figure}[b]
   \fbox{\parbox{0.98\linewidth}{\begin{minipage}{\linewidth}\noindent
\input{TEX/d-types1}
\hrulefill
\input{TEX/fine-d-computations1}
\end{minipage}}}
\caption{A forward-mode AD macro defined on types as $\Dsyn{-}$, values as $\DsynV{-}$, and computations as $\DsynC{-}$.
All newly introduced variables are chosen to be fresh.\label{fig:fine-ad1}}
\end{figure}
In fact, by the universal property of $\SynJ$, $\Dsynsymbol$ is the unique structure preserving functor 
on $\Dsynsymbol$ that has the right definition for constants, primitive operations and $\tSign$.
It automatically follows that $\Dsynsymbol$ respects $\beta\eta$-equality.

Under the translation of coarse-grain CBV into fine-grain CBV, this code transformation induces precisely that of \S\ref{sec:our-cbv-language}.

%% file: TEX/fine-types-values-computations.tex
\begin{syntax}
	\ty{1}, \ty{2}, \ty{3} & \gdefinedby & & \syncat{types}                          \\
	&\gor& \reals                      & \synname{numbers}\\
	&\gor & \Init  \gor \ty{1} + \ty{2}  & \synname{sums}\\
	&&&\\
	\val{1}, \var{2}, \val{3} & \gdefinedby & & \syncat{values}                          \\
	&\gor& \var{1},\var{2},\var{3}                      & \synname{variables}\\
	&\gor& \cnst{c}                   & \synname{constant}\\
	&\gor & \nvMatch{\val{1}} & \synname{sum match}\\
	&\gor& \tInl{\val{1}} \gor   \tInr{\val{1}} & \synname{inclusions}\\
	&&&\\
	\trm{1}, \trm{2}, \trm{3} & \gdefinedby & & \syncat{computations}                          \\
	&\gor & \toin{\var{1}}{\trm{1}}{\trm{2}} & \synname{sequencing}\\
	&\gor & \tReturn \val{1} & \synname{pure comp.}\\
	&\gor & \op(\val{1}_1,\ldots,\val{1}_n) & \synname{operation}\\
	&\gor & \nvMatch{\val{1}} & \synname{sum match}\\
\end{syntax}%
~
\begin{syntax}
	&\gor\quad\, & \Unit  \gor  \ty{1}_1 \t* \ty{1}_2 & \synname{products}\\
	&\gor& \ty{1} \To \ty{2}              & \synname{function}      \\
	& & &\\
	&\gor & \vMatch{\val{1}}{\begin{array}{l}\;\;\tInl\var{1}\To\val{2}\\
			\gor \tInr\var{2}\To \val{3}\end{array}}\hspace{-15pt}\; & \synname{sum match}\\	
	&\gor\quad\, & \tUnit \ \gor  \tPair{\val{1}}{\val{2}} & \synname{tuples}\\
	&\gor\quad\, & \pMatch{\val{1}}{\var{1}}{\var{2}}{\val{2}} & \synname{product match}\\
	&\gor& \fun{\var{1}}{\trm{1}}             & \synname{abstractions}      \\
	& \gor & \rec{\var{1}}\val{1} & \synname{term recursion}\\
	&&&\\
	&\gor & \vMatch{\val{1}}{\begin{array}{l}\;\;\tInl\var{1}\To\trm{1}\\
			\gor \tInr\var{2}\To \trm{2}\end{array}}\hspace{-15pt}\; & \synname{sum match}\\
	&\gor\quad\, & \pMatch{\val{1}}{\var{1}}{\var{2}}{\trm{1}} & \synname{product match}\\
	&\gor& \val{1}\ \val{2}             & \synname{function app.}      \\
	&\gor&\tItFrom{\trm{1}}{\var{1}}{\val{1}} & \synname{iteration}\\
	&\gor&\tSign\val{1} & \synname{sign function}\\
\end{syntax}

%% file: TEX/fine-type-system.tex
\[
\begin{array}{c}
	\inferrule{
		~
	}{
		\Ginfv {\var{1}}{\ty{1}}
	}((\var{1} : \ty{1}) \in \ctx) \quad
	\inferrule{\Ginfc {\trm{1}}{\ty{1}}\quad \Ginfc[,\var{1}:{\ty{1}}]{\trm{2}}{\ty{2}}}
	{\Ginfc {\toin{\var{1}}{\trm{1}}{\trm{2}}}{\ty{2}}}\quad
	\inferrule{\Ginfv {\val{1}}{\ty{1}}}{\Ginfc {\tReturn\val{1}}{\ty{1}}}\quad 
	\inferrule{ ~}{\Ginfv {\cnst{c}}\reals}(c\in\RR)\\ \\
	\inferrule{
		\Ginfv {\val{1}_1} \reals\quad\cdots \quad  \Ginfv {\val{1}_n} \reals
	}{
		\Ginfc {\op(\val{1}_1,\ldots,\val{1}_n)}\reals
	}(\op\in\Op_n)\quad
	\inferrule{\Ginfv{\val{1}}\Init}{\Ginfv {\nvMatch{\val{1}}}{\ty{1}}}\quad 
	\inferrule{\Ginfv{\val{1}}\Init}{\Ginfc {\nvMatch{\val{1}}}{\ty{1}}}
	\\\\ 
	\inferrule{\Ginfv {\val{1}} {\ty{1}}}
	{\Ginfv{\tInl\val{1}}{\ty{1}\t+\ty{2}}}
	\quad
	\inferrule{\Ginfv {\val{1}} {\ty{2}}}
	{\Ginfv{\tInr\val{1}}{\ty{1}\t+\ty{2}}}\quad
	\inferrule{\Ginfv{\val{1}}{\ty{2}\t+\ty{3}}\quad
		\Ginfv[,\var{1}:{\ty{2}}]{\val{2}}{\ty{1}}\quad 
		\Ginfv[,{\var{2}}:{\ty{3}}]{\val{3}}{\ty{1}}
	}{\Ginfv{\bvMatch{\val{1}}{\var{1}}{\val{2}}{\var{2}}{\val{3}}}{\ty{1}}}
	\\\\ 
	\inferrule{\Ginfv{\val{1}}{\ty{2}\t+\ty{3}}\quad
		\Ginfc[,\var{1}:{\ty{2}}]{\trm{1}}{\ty{1}}\quad 
		\Ginfc[,{\var{2}}:{\ty{3}}]{\trm{2}}{\ty{1}}
	}{\Ginfc{\bvMatch{\val{1}}{\var{1}}{\trm{1}}{\var{2}}{\trm{2}}}{\ty{1}}}
	\quad 
	\inferrule{~}{\Ginfv \tUnit \Unit}\quad
	\inferrule{\Ginfv {\val{1}}{ \ty{1}}\quad \Ginfv {\val{2}} {\ty{2}}}
	{\Ginfv {\tPair{\val{1}}{\val{2}}} {\ty{1}\t*\ty{2}}}
	\\\\ 
	\inferrule{\Ginfv {\val{1}}{\ty{2}\t*\ty{3}}\quad 
		\Ginfv[{,\var{1}:\ty{2},\var{2}:\ty{3}}] {\val{2}}{\ty{1}}}{\Ginfv{\pMatch{\val{1}}{\var{1}}{\var{2}}{\val{2}}}{\ty{1}}}
	\quad \!
	\inferrule{\Ginfv {\val{1}}{\ty{2}\t*\ty{3}}\quad 
		\Ginfc[{,\var{1}:\ty{2},\var{2}:\ty{3}}] {\trm{1}}{\ty{1}}}{\Ginfc{\pMatch{\val{1}}{\var{1}}{\var{2}}{\trm{1}}}{\ty{1}}}
	\quad \!
	\inferrule{\Ginfc[,\var{1}:{\ty{2}}]{\trm{1}}{\ty{1}}}{\Ginfv{\fun{\var{1}}{\trm{1}}}{\ty{2}\To\ty{1}}} 
	\\ \\ 
	\inferrule{\Ginfv {\val{1}} {\ty{2}\To\ty{1}}
		\quad\Ginfv {\val{2}}{\ty{2}} 
	}{\Ginfc {\val{1}\,\val{2}}{\ty{1}}}\quad 
	\inferrule{\Ginfc[,\var{1}:{\ty{2}}] {\trm{1}}{\ty{2}\t+\ty{1}}
		\quad \Ginfv {\val{1}} {\ty{2}} }{\Ginfc {\tItFrom{\trm{1}}{\var{1}}{\val{1}}}{\ty{1}}}
	\\\\
	\inferrule{ \Ginfv[,\var{1}:\ty{1}]{\val{1}}{\ty{1}}}
	{\Ginfv{\rec{\var{1}}{\val{1}}}{\ty{1}}}(\ty{1}=\ty{2}\To\ty{3})
	\quad
	\inferrule{\Ginfv {\val{1}} \reals}{\Ginfc {\tSign\val{1}} {\Unit\t+ \Unit}
	}
\end{array}
\]

%% file: TEX/fine-beta-eta1.tex
\[
\begin{array}{l@{\ }l@{\ }l@{\ }l@{\ }l@{\ }l}
    \toin{\var{2}}{(\toin{\var{1}}{\trm{1}}{\trm{2}})}{\trm{3}}&\beeq & \toin{\var{1}}{\trm{1}}{(\toin{\var{2}}{\trm{2}}{\trm{3}})}\hspace{-30pt}
    \\
    \toin{\var{1}}{\tReturn\val{1}}{\trm{1}}&\beeq & \subst{\trm{1}}{\sfor{\var{1}}{\val{1}}}
    &
    \\
    \bvMatch{\tInl\val{1}}{\var{1}}{\val{2}}{\var{2}}{\val{3}}&\beeq & \subst{\val{2}}{\sfor{\var{1}}{\val{1}}}&
    \subst{\val{2}}{\sfor{\var{3}}{\val{1}}}&\freeeq{\var{1},\var{2}}&\vMatch{\val{1}}
    {\begin{array}{l}\;\,\, \tInl\var{1}\To\subst{\val{2}}{\sfor{\var{3}}{\tInl\var{1}}}\\ \gor \tInr\var{2}\To\subst{\val{2}}{\sfor{\var{3}}{\tInr\var{2}}}
    \end{array}}
    \\
    \bvMatch{\tInr\val{1}}{\var{1}}{\val{2}}{\var{2}}{\val{3}}&\beeq & \subst{\val{3}}{\sfor{\var{2}}{\val{1}}}
    &\\
    \pMatch{\tPair{\val{1}}{\val{2}}}{\var{1}}{\var{2}}{\val{3}}&\beeq & \subst{\val{3}}{\sfor{\var{1}}{\val{1}},\sfor{\var{2}}{\val{2}}}\;&
    \subst{\val{3}}{\sfor{\var{3}}{\val{1}}}&\freeeq{\var{1},\var{2}}&\pMatch{\val{1}}{\var{1}}{\var{2}}{\subst{\val{3}}{\sfor{\var{3}}{\tPair{\var{1}}{\var{2}}}}} \\
    \bvMatch{\tInl\val{1}}{\var{1}}{\trm{1}}{\var{2}}{\trm{2}}&\beeq & \subst{\trm{1}}{\sfor{\var{1}}{\val{1}}}&
    \subst{\trm{1}}{\sfor{\var{3}}{\val{1}}}&\freeeq{\var{1},\var{2}}&\vMatch{\val{1}}
    {\begin{array}{l}\;\,\, \tInl\var{1}\To\subst{\trm{1}}{\sfor{\var{3}}{\tInl\var{1}}}\\ \gor \tInr\var{2}\To\subst{\trm{1}}{\sfor{\var{3}}{\tInr\var{2}}}
    \end{array}}
    \\
    \bvMatch{\tInr\val{1}}{\var{1}}{\trm{1}}{\var{2}}{\trm{2}}&\beeq & \subst{\trm{2}}{\sfor{\var{2}}{\val{1}}}
    &\\
    \pMatch{\tPair{\val{1}}{\val{2}}}{\var{1}}{\var{2}}{\trm{1}}&\beeq & \subst{\trm{1}}{\sfor{\var{1}}{\val{1}},\sfor{\var{2}}{\val{2}}}&
    \subst{\trm{1}}{\sfor{\var{3}}{\val{1}}}&\freeeq{\var{1},\var{2}}&\pMatch{\val{1}}{\var{1}}{\var{2}}{\subst{\trm{1}}{\sfor{\var{3}}{\tPair{\var{1}}{\var{2}}}}} \\
    (\fun{\var{1}}{\trm{1}})\ \val{1} &\beeq &  \subst{\trm{1}}{\sfor{\var{1}}{\val{1}}}&
    \val{1} &\freeeq{\var{1}\phantom{,\var{2}}}& \fun{\var{1}}{\val{1}\, \var{1}}
\end{array}    
\]

%% file: TEX/fine-d-computations1.tex
\[
\begin{array}{l}
    \DsynV{\var{1}} \defeq \var{1}\\
    \DsynV{\nvMatch{\val{1}}}\defeq \nvMatch{\DsynV{\val{1}}}\\
    \DsynV{\tInl\val{1}} \defeq \tInl{\DsynV{\val{1}}}\\
   \DsynV{\tInr\val{1}} \defeq \tInr{\DsynV{\val{1}}} \\
   \DsynV{\vMatch{\val{1}}{\begin{array}{l}\;\;\tInl\var{1}\To\val{2}\\
    \gor \tInr\var{2}\To \val{3}\end{array}}}\defeq
  \vMatch{\DsynV{\val{1}}}{\begin{array}{l}\;\;\tInl\var{1}\To\DsynV{\val{2}}\\
    \gor \tInr\var{2}\To \DsynV{\val{3}}\end{array}}\\ 
   \DsynV{\tUnit}\defeq \tUnit \\
   \DsynV{\tPair{\val{1}}{\val{2}}} \defeq \tPair{\DsynV{\val{1}}}{\DsynV{\val{2}}}\\
   \DsynV{\pMatch{\val{1}}{\var{1}}{\var{2}}{\val{3}}}\defeq \pMatch{\DsynV{\val{1}}}{\var{1}}{\var{2}}{\DsynV{\val{3}}}\\
   \DsynV{\fun {\var{1} }   {\trm{1}}} \defeq \fun{\var{1}}{\DsynC{\trm{1}}}\\
\DsynC{\toin{\var{1}}{\trm{1}}{\trm{2}}}\defeq \toin{\var{1}}{\DsynC{\trm{1}}}{\DsynC{\trm{2}}}\\
\DsynC{\tReturn\val{1}}\defeq \tReturn{\DsynV{\val{1}}}\\
\DsynC{\nvMatch{\val{1}}}\defeq  \nvMatch{\DsynV{\val{1}}} \\
\DsynC{\vMatch{\val{1}}{\begin{array}{l}\;\;\tInl\var{1}\To\trm{1}\\
\gor \tInr\var{2}\To \trm{2}\end{array}}}\defeq \vMatch{\DsynV{\val{1}}}{\begin{array}{l}\;\;\tInl\var{1}\To\DsynC{\trm{1}}\\
\gor \tInr\var{2}\To \DsynC{\trm{2}}\end{array}}\\ 
\DsynC{\pMatch{\val{1}}{\var{1}}{\var{2}}{\trm{1}}}\defeq \pMatch{\DsynV{\val{1}}}{\var{1}}{\var{2}}{\DsynC{\trm{1}}}\\
\DsynC{\val{1}\,\val{2}}\defeq \DsynV{\val{1}}\,\DsynV{\val{2}} \\
\DsynC{\tItFrom{\trm{1}}{\var{1}}{\val{1}}}\defeq \tItFrom{\DsynC{\trm{1}}}{\var{1}}{\DsynV{\val{1}}}\\
\DsynC{\rec{\var{1}}{\trm{1}}}\defeq \rec{\var{1}}{\DsynC{\trm{1}}}
\end{array}
\]
\hrulefill
\[
    \begin{array}{ll}
        \DsynV{\cnst{c}} \defeq &\tPair{\cnst{c}}{\cnst{0}}\\
        \DsynC{\op(\val{1}_1,\ldots,\val{1}_n)}\defeq~
                           &\pMatch{\DsynV{\val{1}_1}}{\var{1}_1}{\var{1}_1'}
                           { \ldots \to\pMatch{\DsynV{\val{1}_n}}{\var{1}_n}{\var{1}_n'}
                           {\\
                           &\toin{\var{2}}{\op(\var{1}_1,\ldots,\var{1}_n)}{} 
                           \\
                           &\toin{\var{3}_1}{\partial_1\op(\var{1}_1,\ldots,\var{1}_n)}{\ldots}\toin{\var{3}_n}{\partial_n\op(\var{1}_1,\ldots,\var{1}_n)}{}\\
                           &\tReturn{\tPair{\var{2}}{\var{1}_1' *\var{3}_1+\ldots
                           +\var{1}_n' *\var{3}_n}}}}
                           \\
                           \DsynC{\tSign{\val{1}}}\defeq &\tSign{(\tFst\DsynV{\val{1}})}
    \end{array} 
\]

%% file: TEX/efficient-sign-derivative.tex
\section{A more efficient derivative for $\tSign{}$}\label{sec:efficient-sign-derivative}
We can define by mutual induction (for both $\Dsynplainsymbol=\Dsynsymbol,\Dsynrevsymbol[k]$)

\begin{align*}
    x : \Dsynplain{\ty{1}} &\vdash \typeproj{\ty{1}}(x) : \ty{1}\\ 
 x : \Dsynplain{\reals}  &\vdash \tFst (x) : \reals \\
x : \Dsynplain{\ty{1}} \t*  \Dsynplain{\ty{2}} &\vdash \tPair{\typeproj{\ty{1}}(\tFst  x)}{\typeproj{\ty{2}}(\tSnd  x)} : \ty{1} \t*  \ty{2} \\
x : \Dsynplain{\ty{1}} \t+  \Dsynplain{\ty{2}} &\vdash \vMatch{x}{\tInl y\To \tInl  \typeproj{\ty{1}}(y) \mid \tInr  z \To \tInr  \typeproj{\ty{2}}(z)} : \ty{1} \t+  \ty{2}\\
x : \Dsynplain{\ty{1}} \To \Dsynplain{\ty{2}} &\vdash \fun{y}{\typeproj{\ty{2}}(x(\typezero{\ty{1}}(y)))} : \ty{1} \To \ty{2}\\
x : \trec{\tvar{1}} \Dsynplain{\ty{1}} &\vdash  \rMatch{x}{y}{\tRoll  \typeproj{\ty{1}}(x) } : \trec{\tvar{1}}\ty{1}\\
x : \tvar{1} &\vdash x : \tvar{1}
\end{align*}
and
\begin{align*}
x : \ty{1} &\vdash \typezero{\ty{1}}(x) : \Dsynplain{\ty{1}} \\
x : \reals &\vdash \tPair{x}{\cnst 0} : \Dsynplain{\reals} \\
x : \ty{1} \t*  \ty{2} &\vdash \tPair{\typezero{\ty{1}}(\tFst  x)}{\typezero{\ty{2}}(\tSnd  x)} : \Dsynplain{\ty{1}} \t*  \Dsynplain{\ty{2}}\\
x : \ty{1} \t+  \ty{2} &\vdash \vMatch{x}{\tInl  y \To \tInl     \typezero{\ty{1}}(y) \mid \tInr  z \To \tInr  \typezero{\ty{2}}(z)} : \Dsynplain{\ty{1}} \t+  \Dsynplain{\ty{2}}\\
x : \ty{1} \To \ty{2} &\vdash \fun{y}{\typezero{\ty{2}}(x(\typeproj{\ty{1}}(y)))} : \Dsynplain{\ty{1}} \To \Dsynplain{\ty{2}} \\
x : \trec{\tvar{1}} \ty{1} &\vdash \rMatch{x}{y}{\tRoll  \typezero{\ty{1}}(x) }: \trec{\tvar{1}}\Dsynplain{\ty{1}}\\
x : \tvar{1} &\vdash x : \tvar{1}.
\end{align*}
Then, observe that, for any 
$\var{1}_1:\ty{1}_1,\ldots,\var{1}_n:\ty{1}_n\vdash \trm{1}:\reals$, we have\\
$\sem{\tSign{(\tFst\Dsynplain{\trm{1}})}}=\sem{\letin{x_1}{\typeproj{\ty{1}_1}(x_1)}{\cdots}\letin{x_n}{\typeproj{\ty{1}_n}(x_n)}{\cdots}{\tSign \trm{1}}}$.\\
Therefore, we can define, for $\var{1}_1:\ty{1}_1,\ldots,\var{1}_n:\ty{1}_n\vdash \trm{1}:\reals$,
\begin{align*}
\Dsynplain{\tSign\trm{1}}&\defeq \letin{x_1}{\typeproj{\ty{1}_1}(x_1)}{\cdots}\letin{x_n}{\typeproj{\ty{1}_n}(x_n)}{\cdots}{\tSign \trm{1}}.
\end{align*}
This yields more efficient definitions of the forward and reverse derivatives of $\tSign{}$ and $\ifelse{}{}{}$
as we do not need to differentiate $\trm{1}$ at all.

%% file: TEX/wCPO-enriched-scone.tex
\section{Enriched scone}\label{app-wCPO-enriched-scone}

We present straightforward generalizations (enriched versions) of the results presented in \cite[Section~9]{VAKAR-LUCATELLI2021} below. 

Considering the $\wCpo $-category $\morphismcatTwo $ with two objects and only one non-trivial morphism between them, the
$\wCpo $-category $\morphismcatTwo\pitchfork\catD $ of morphisms of $\catD $
can be described as the $\wCpo $-category $\ihom{\wCpo\textrm{-}\catCat}{\morphismcatTwo}{\catD} $ of $\wCpo $-functors 
$\morphismcatTwo\to \catD $.

Explicitly, the objects of  $\morphismcatTwo\pitchfork\catD $  are 
morphisms $f: Y_0\to Y_1 $ of $\catD $. A morphism between $f$ and $g$ is a pair $\alpha = \left( \alpha _0 , \alpha _1\right) : f\to g $ such that $\alpha _1 f = g\alpha _0 $, that is to say, a ($\wCpo$-)natural transformation. Finally, 
the $\wCpo$-structure is defined by $\left( \alpha _0 , \alpha _1\right) \leq \left( \beta _0 , \beta _1\right) $
if $\alpha _0\leq \beta _ 0 $ and $\alpha _1\leq \beta _ 1 $ in $\catD $. 


Given an $\wCpo$-functor $G:\catC\to\catD$, the comma category $\catD\downarrow G$ of the identity on $\catD $ along $G$ in 
$\ecat{\wCpo}$ is also known as the $\wCpo$-\textit{scone}
or \textit{Artin glueing} of $G$. It can be described as the pullback \eqref{eq:pullback-comma-codomain} in $\ecat{\wCpo}$, in which $\codomm : \morphismcatTwo\pitchfork\catD\to\catD $, defined by 
$\left( \alpha = \left( \alpha _0 , \alpha _1\right) : f \to g\right)\mapsto \alpha _1 $,
is the codomain $\wCpo $-functor.
\begin{equation}\label{eq:pullback-comma-codomain}
	\diag{wCpo-category-of-morphisms-pullback-comma}
\end{equation}

Since $\codomm $ is an isofibration, 
the pullback \eqref{eq:pullback-comma-codomain}  is equivalent to the pseudo-pullback of $\codomm $ along $G$, which is the $\wCpo $-category defined as follows. The objects of the pseudo-pullback 
are triples $$\left( \left( f: Y_0\to Y_1\right)\in \morphismcatTwo\pitchfork\catD , C\in \catC , \xi : \left(\codomm f\right)\xto{\cong} G(C) \right) $$ where $\xi $ is an isomorphism in $\catD $. A morphism $(f,C,\xi) \to (f',C',\xi ' )$ is a pair of morphisms  $\left( \alpha : f\rightarrow f', h: C\to C' \right) $ such that
$ G(h)\circ \xi =      \xi '\circ  \codomm \left( \alpha\right)  $. Finally, the $\wCpo $-structure of the homs are given pointwise. That is to say, $\left( \alpha  , h \right)\leq \left( \alpha ' , h ' \right) $ if $\alpha \leq \alpha $ in $\morphismcatTwo\pitchfork\catD $   and $h\leq h '$ in $\catC$. 
  
\begin{lemma}
	The forgetful $\wCpo$-functor 	$\forgetfulS : \catD \downarrow G\to \catD\times\catC $, defined in \eqref{eq:forgetful-scone-notation}, creates all absolute (weighted) limits and colimits. 
\end{lemma}
\begin{proof}  
Clearly, the $\wCpo $-functor $\forgetfulS$ reflects isomorphisms.
	
Let $D$ be a diagram in $\catD\downarrow G$ such that the weighted (co)limit $(co)\limm{W}{\forgetfulS D} $ exists and is preserved by any $\wCpo$-functor.
Since $\catD\downarrow G$ is the pullback \eqref{eq:pullback-comma-codomain}, there is a unique pair of diagrams $\left( D_0, D_1 \right) $ such that 
\begin{equation*}
	\morPULL\circ D = D_0, \quad  \cCPULL\circ D = D_1, \quad\codomm\circ D_0 = G\circ D_1 , 
\end{equation*}
hold.
  
Since $\domm\circ D_0 =  \proj{\catD}\circ\forgetfulS\circ D $ and  $\codomm\circ D_0 = G\circ \proj{\catC}\circ\forgetfulS\circ D $, we get that $(co)\limm{W}{\domm D_0}\cong \proj{\catD}\left( (co)\limm{W}{\forgetfulS\circ D }\right)  $ and $ (co)\limm{W}{\codomm\circ D_0}\cong G\circ\proj{\catC}\left( (co)\limm{W}{\forgetfulS\circ D }\right) $. Therefore, $(co)\limm{W}{\forgetfulS\circ D_0 }$ exists in $\morphismcatTwo\pitchfork\catD$, pointwise constructed out of $(co)\limm{W}{\domm\circ D_0}$ and  $(co)\limm{W}{\codomm\circ D_0}$.

Moreover, since $D_1 = \proj{\catC}\circ \forgetfulS\circ D $, we have that
$(co)\limm{W}{D_1}\cong \proj{\catC}\left(  (co)\limm{W}{\forgetfulS\circ D }\right)  $.

Therefore, the isomorphism $\xi $ given by
\begin{eqnarray*} 
\codomm\left(  (co)\limm{W}{D_0}\right) &\cong & (co)\limm{W}{\codomm\circ D_0} \\
 & \cong & G\circ\proj{\catC}\left( (co)\limm{W}{\forgetfulS\circ D }\right)\\
 & \cong & G \left( (co)\limm{W}{D_1} \right) 
\end{eqnarray*} 
together with the pair $\left( (co)\limm{W}{D_0} , (co)\limm{W}{D_1} \right) $
defines, up to isomorphism, an object of $\catD\downarrow G$, which satisfies the universal property for 
$(co)\limm{W}{D} = (co)\limm{W}{\left( D_0, D_1\right) } $.

Moreover, by the construction above, we conclude that $(co)\limm{W}{D} $ is preserved by $\forgetfulS $. In particular:
$$ \forgetfulS\left( (co)\limm{W}{D_0} , (co)\limm{W}{D_1}, \xi  \right) =
\left(  (co)\limm{W}{\domm\circ D_0}  , (co)\limm{W}{D_1}\right) . 
$$
The above completes the proof that the $\wCpo$-functor $\forgetfulS $ creates $(co)\limm{W}{D}$.
\end{proof} 
The $\wCpo $-functor $\forgetfulS $ has a right $\wCpo $-adjoint provided that 
$\catD $ has binary $\wCpo $-products. It is given by $\left(  D\in\catD , C\in\catC  \right)\mapsto\left( D\times G\left(  C\right), C, \proj{G(C)}  \right)   $. 
Therefore:
\begin{theorem}
The forgetful $\wCpo$-functor 	$\forgetfulS : \catD \downarrow G\to \catD\times\catC $ 
is $\wCpo $-comonadic provided that $\catD $ has binary $\wCpo $-products.
\end{theorem} 	 
By duality, we get that the forgetful $\wCpo$-functor $F \downarrow \catC \to \catD\times\catC$ is $\wCpo $-monadic
provided that $\catC $ has finite $\wCpo $-coproducts. Therefore:
\begin{theorem}
	The forgetful $\wCpo$-functor 	$\forgetfulS : \catD \downarrow G\to \catD\times\catC $ 
	is $\wCpo $-monadic whenever $G$ has a left $\wCpo $-adjoint and $\catC $ has finite $\wCpo $-coproducts.
\end{theorem} 
\begin{proof}
	Indeed, by the $\wCpo$-adjunction $F\dashv G $, we get an isomorphism $\catD \downarrow G\cong F \downarrow \catC$ which composed with
	the forgetful $\wCpo$-functor $ F \downarrow \catC \to \catD\times\catC$ is
	equal to $\forgetfulS : \catD \downarrow G\to \catD\times\catC $.  
\end{proof}

As a consequence, we conclude that:
\begin{theorem}
Let  $G: \catC\to\catD $ be a right $\wCpo $-adjoint functor between $\wCpo $-bicartesian closed categories. We have that the forgetful $\wCpo $-functor $\forgetfulS $ is $\wCpo$-monadic and comonadic. In particular, $\catD \downarrow G$ is an $\wCpo $-bicartesian closed category.
\end{theorem}

%% file: TEX/app-rnn.tex
\section{Some Haskell Code for a Recursive Neural Network}
\label{ap:rnn}
\begin{lstlisting}[
    basicstyle=\footnotesize, %or \small or \footnotesize etc.
]
-- example implementation of https://icml.cc/2011/papers/125_icmlpaper.pdf
-- Some of the basic datatypes we use -- we elide the implementation of some
data Tree a
  = Leaf a
  | Node (Tree a) (Tree a)
  deriving (Eq) -- \mu b. a + (b x b), leaf a = roll (iota_1 a), node l r = roll (iota_2 (l, r))

data Vector

data Scalar

data Matrix

type ActivationVectors = [Vector]

type AdjacencyMatrix = [(Tree Int, Tree Int)]

-- Some basic data and operations that we need for the RNN
-- Again, we elide much of the implementation as it is beside the point of this example
f :: Vector -> Vector --  some non-linear function, usually elementwise applied sigmoid function
f = undefined

conc :: Vector -> Vector -> Vector -- concatenate vectors
conc = undefined

mult :: Matrix -> Vector -> Vector -- matrix vector multiplication
mult = undefined

add :: Vector -> Vector -> Vector -- elementwise vector addition
add = undefined

innerprod :: Vector -> Vector -> Scalar -- vector inner product
innerprod = undefined

a :: ActivationVectors
a = undefined -- input (for example, sequence of words as vectors or image segments as vectors)

adjMat :: AdjacencyMatrix
-- start with matrix that only stores (Leaf i, Leaf j) pairs in case i is a neighbour of j;
-- we later extend adjacency to parent nodes
adjMat = undefined -- input (specify which words/image segments are neighbours  )

w :: Matrix
w = undefined -- parameter to learn: weights

b :: Vector
b = undefined -- parameter to learn: biases

wScore :: Vector
wScore = undefined -- parameter to learn: scoring vector

-- The implementation of the RNN
-- version 1, without caching
modelH ((w, b, wScore), (adjMat, globalScore)) =
  let getNode (Leaf i) = a !! i
   in let getNode (Node l r) = f (w `mult` conc (getNode l) (getNode r)) `add` b
       in let parentsScores =
                map
                  (\i -> (i, innerprod wScore (getNode (uncurry Node i))))
                  adjMat -- compute scores for all parent nodes of neighbours;
                  -- super inefficient without caching getNode, but conceptually cleaner
           in let ((bp1, bp2), bestScore) =
                    foldl
                      (\(i, s) (i', s') ->
                         if s > s'
                           then (i, s)
                           else (i', s'))
                      (head parentsScores)
                      parentsScores -- find the neighbours that have the higest score
               in let globalScore2 = globalScore + bestScore
               -- add the local contribution of our chosen neighbour pair to the global score
                   in let bestPar = Node bp1 bp2
                   -- actually compute our favourite parent;
                   -- I guess we'd already done this before but it's cheap to redo
                       in let mergeParH i
                                | i == bp1 || i == bp2 = bestPar
                           in let mergeParH i
                                    | otherwise = i
                               in let mergePar (i, j) =
                                        (mergeParH i, mergeParH j)
                                   in let adjMat2 =
                                            filter
                                              (/= (bestPar, bestPar))
                                              [ mergePar (i, j)
                                              | (i, j) <- adjMat
                                              ]
                                              -- replace bp1 and bp2 with bestPar in adjacencyMatrix,
                                              -- but we have a convention that nodes are not neighbours
                                              -- of themselves
                                       in if null adjMat2
                                            then Right globalScore2
                                            else Left (adjMat, globalScore2)
                                            -- if we run out of neighbours that can be merged, we are done;
                                            -- otherwise iterate with new adjacency matrix and score

it :: ((c, a) -> Either a b) -> (c, a) -> b -- functional iteration
it f (c, a) =
  case f (c, a) of
    Left a' -> it f (c, a')
    Right b -> b

model :: ((Matrix, Vector, Vector), (AdjacencyMatrix, Scalar)) -> Scalar
model = it modelH

-- The implementation of the RNN
-- version2, with caching of getNode
modelH2 ((w, b, wScore), (adjMat, values, globalScore)) =
  let getNode (Leaf i) = look (Leaf i) values
   in let getNode (Node l r) =
            let lv = look l values
             in let rv = look r values
                 in f (w `mult` conc lv rv) `add` b
       in let parentsValScores =
                map
                  (\i ->
                     let v = getNode (uncurry Node i)
                      in (i, v, innerprod wScore v))
                  adjMat
           in let ((bp1, bp2), bestVal, bestScore) =
                    foldl
                      (\(i, v, s) (i', v', s') ->
                         if s > s'
                           then (i, v, s)
                           else (i', v', s'))
                      (head parentsValScores)
                      parentsValScores
               in let globalScore2 = globalScore + bestScore
                   in let bestPar = Node bp1 bp2
                       in let mergeParH i
                                | i == bp1 || i == bp2 = bestPar
                           in let mergeParH i
                                    | otherwise = i
                               in let mergePar (i, j) =
                                        (mergeParH i, mergeParH j)
                                   in let adjMat2 =
                                            filter
                                              (/= (bestPar, bestPar))
                                              [ mergePar (i, j)
                                              | (i, j) <- adjMat
                                              ]
                                       in if null adjMat2
                                            then Right globalScore2
                                            else Left
                                                   ( adjMat
                                                   , (bestPar, bestVal) : values
                                                   , globalScore2)

-- initial values will be zip (map Leaf [0..], a)
look :: Tree Int -> [(Tree Int, b)] -> b --  a map operation for looking up cache
look k m =
  case lookup k m of
    Just x -> x

model2 :: ((Matrix, Vector, Vector), (AdjacencyMatrix, Scalar)) -> Scalar
model2 ((w, b, wScore), (adjMat, globalScore)) =
  it modelH2 ((w, b, wScore), (adjMat, zip (map Leaf [0 ..]) a, globalScore))

\end{lstlisting}